\documentclass[acmsmall,screen]{acmart}

\acmJournal{PACMPL}
\acmVolume{1}
\acmNumber{CONF} 
\acmArticle{1}
\acmYear{2025}
\acmMonth{1}
\acmDOI{} 
\startPage{1}

\setcopyright{acmlicensed}

\usepackage[inline]{enumitem} 
\usepackage{amsmath}
\usepackage{fixltx2e}
\usepackage{graphicx}
\usepackage{multibib}
\usepackage{mdframed}
\usepackage{nicefrac}
\usepackage{amsfonts}
\usepackage{mathpartir}
\usepackage{amsthm}
\usepackage{bm}
\usepackage{relsize}
\usepackage{setspace}
\usepackage{stackengine}
\usepackage{thmtools}
\usepackage{thm-restate}
\usepackage{mathpartir}
\usepackage{wrapfig}

\usepackage{ebproof} 
\usepackage{tcolorbox}
\usepackage{dashbox}
\usepackage{fontenc}
\usepackage{textcomp}
\usepackage{newtxmath}

\usepackage{algorithm}
\usepackage{algpseudocode}
\usepackage{algorithmicx}

\usepackage{tikz}
\usepackage{thmtools}
\usepackage{thm-restate}
\usepackage{mathtools}
\usepackage{xspace}

\usetikzlibrary{shapes, shapes.geometric, arrows, arrows.meta, calc, decorations, decorations.pathreplacing, calligraphy, positioning, patterns, backgrounds}
\tikzstyle{rect} = [rectangle, rounded corners, minimum width=3cm, minimum height=1cm,text centered, draw=black, fill=black!10]
\tikzstyle{medrect} = [rectangle, rounded corners, minimum height=.7cm,text centered, draw=black, fill=black!10]
\tikzstyle{smallrect} = [rectangle, rounded corners,text centered, draw=black, fill=black!10]
\tikzstyle{circ} = [circle, minimum height=0.55cm,text centered, draw=black, fill=black!10]
\tikzstyle{smallcirc} = [circle,text centered, draw=black, fill=black!10]
\tikzstyle{arrow} = [thick,-{Latex[length=2mm, width=1.5mm]},>=stealth]
\tikzstyle{arrow2} = [thick,{Latex[length=2mm, width=1.5mm]}-{Latex[length=2mm, width=1.5mm]},>=stealth]
\tikzstyle{garrow} = [thick,-{Latex[length=2mm, width=1.5mm]},>=stealth,draw=black!40]
\tikzstyle{garrow2} = [thick,{Latex[length=2mm, width=1.5mm]}-{Latex[length=2mm, width=1.5mm]},>=stealth,draw=black!40]
\tikzstyle{venn} = [preaction={fill, #1},opacity=0.6,anchor=south,rounded corners=2pt,draw=black]
\tikzstyle{exnode} = [circle,draw=black,inner sep=1pt,fill=gray!20]

\usepackage{textcomp}
\usepackage{subcaption}

\usepackage[hypcap=false]{caption}
\usepackage[list=true]{subcaption}

\usepackage{xr-hyper}
\usepackage{xr}

\usepackage[export]{adjustbox}

\usepackage{subject_reduction}
\RequirePackage{ifthen}
\usepackage{extra_definitions}

\newcommand{\proofreference}[1]{(\text{Proof in {\color{blue}#1}})}

\newif\iffull\fullfalse
\fulltrue

\begin{document}
\title{Top-Down or Bottom-Up? ({\color{blue} Full Version})}
\subtitle{Complexity Analyses of Synchronous Multiparty Session Types}

\author{Thien Udomsrirungruang}
\email{thien.udomsrirungruang@keble.ox.ac.uk}
\affiliation{%
    \institution{University of Oxford}
    \city{Oxford}
    \country{UK}
}
\orcid{0009-0005-2320-9178}
\author{Nobuko Yoshida}
\email{nobuko.yoshida@cs.ox.ac.uk}
\affiliation{%
    \institution{University of Oxford}
    \city{Oxford}
    \country{UK}
}
\orcid{0000-0002-3925-8557}





\begin{abstract}
Multiparty session types (MPST) provide a type discipline for 
ensuring communication safety, deadlock-freedom and liveness 
for multiple concurrently running participants. 
The original formulation of MPST takes 
the \emph{top-down} approach, where a
\emph{global type} specifies a bird's eye view of 
the intended interactions between participants,  
and each distributed process is locally type-checked
against its end-point projection. 
A more recent one takes 
the \emph{bottom-up} approach,  
where a desired property $\varphi$ of 
a set of participants is ensured if 
the same property $\varphi$ is true for
an ensemble of end-point types (a typing context) 
inferred from each participant.

This paper compares these two main procedures of MPST, 
giving their detailed complexity analyses.
To this aim, we build 
several new algorithms missing from the
bottom-up or top-down workflows by using 
graph representation of session types
(\emph{type graphs}). 
We first propose a subtyping system
based on type graphs, 
offering more efficient (quadratic) subtype-checking 
than the existing  (exponential) inductive algorithm 
\cite{Ghilezan2019}.
Next for the top-down,
we measure complexity of the four end-point projections 
from the literature \cite{Honda2008,YH2024,Tirore2023,Ghilezan2019}.  
For the coinductive projection in \cite{Ghilezan2019},  
we build a new sound and complete \PSPACE-algorithm using type graphs. 
For bottom-up, 
we develop a novel \emph{type inference system} from MPST processes
which generates \emph{minimum type graphs}, succinctly capturing
covariance of internal choice and contravariance of external choice. 
For property-checking of typing contexts,  
we achieve \PSPACE-hardness by 
reducing it from the quantified Boolean formula (QBF) problem,
and build nondeterministic algorithms that search for counterexamples
to prove membership in \PSPACE.
We also present deterministic analogues of these algorithms that run in exponential time.
Finally, we calculate the total complexity of the top-down 
and the bottom-up approaches. 
Our analyses reveal that the top-down based on 
global types is more efficient than the bottom-up in many realistic cases;  
liveness checking for typing contexts in the bottom-up 
has the highest complexity; and  
the type inference costs exponential against the size of a process,
which impacts the total complexity. 
\end{abstract}

\maketitle

\section{Introduction}
\label{sec:introduction}
With the growing complexity of communication behaviours 
of distributed and concurrent systems, 
their protocol specification has become of central 
importance \cite{KuhnMT23,Montesi_2023}. 
\emph{Multiparty session types} (MPST)~\cite{Honda2008} 
offer a type discipline 
for ensuring desirable properties such as deadlock-freedom 
and liveness for multiple interacting participants  
with guidance by a protocol specification, called \emph{global types}. 
A global type gives a system designer 
an entire interaction scenario. 
It is projected into a collection of \emph{end-point} or \emph{local  types}, and   
each participant is type-checked locally and independently 
against its local type. 
This can ensure that the execution among participants 
does not get stuck, following a protocol, offering \emph{correctness by construction}. 
This approach, called \emph{top-down}, has 
been implemented and integrated into diverse tools and 
programming languages. The theory of the end-point 
projections has been advanced over the years based on 
a variety of
techniques from e.g., event structures, 
communicating 
automata, linear and separation logics, 
message sequence charts and algebraic rewriting 
\cite{CastellaniDG23,DBLP:conf/concur/MajumdarMSZ21,Li2023,Li2024,Stutz2023,JY2020,CY2020,BarbaneraLT23}, and 
is mechanised by several tools \cite{Tirore2023,JacobsBK22a,CFGY2021}.  

\newcommand{\introTextboxWidth}{4.5cm}
\newcommand{\introTextboxShift}{1.7cm}
\newcommand{\introBraceShift}{-2.6cm}

\begin{figure}[t]
  \begin{center}
\hspace{1cm}
\begin{subfigure}[T]{.27\textwidth}
\begin{center}
     \small
      \begin{tikzpicture}[>=Latex, node distance=1.2cm]
        \node (gt)[smallrect] {Global type $\G$};
        \node (t1) [below of=gt, xshift=-0.9cm] {$\T_1$};
        \node (t2) [below of=gt, xshift=-0.2cm] {$\T_2$};
        \node (tn) [below of=gt, xshift=0.9cm] {$\T_n$};
        \node (u1) [below of=t1, yshift=0.5cm] {$\T_1'$};
        \node (u2) [below of=t2, yshift=0.5cm] {$\T_2'$};
        \node (un) [below of=tn, yshift=0.5cm] {$\T_n'$};
        \node (p1) [below of=u1, yshift=0.3cm] {$\PP_1$};
        \node (p2) [below of=u2, yshift=0.3cm] {$\PP_2$};
        \node (pn) [below of=un, yshift=0.3cm] {$\PP_n$};
        
        \node[rotate=90] (s1) at ($(t1)!0.5!(u1)$) {{\large $\subt$}};
        \node[rotate=90] (s2) at ($(t2)!0.5!(u2)$) {{\large $\subt$}};
        \node[rotate=90] (sn) at ($(tn)!0.5!(un)$) {{\large $\subt$}};

        \draw[->] (gt) -- (t1);
        \draw[->] (gt) -- (t2);
        \draw[->] (gt) -- (tn);
    
        \draw[->] (u1) -- (p1);
        \draw[->] (u2) -- (p2);
        \draw[->] (un) -- (pn);

        \node (td) at ($(t2)!0.5!(tn)$) {$\cdots$};
        \node (ud) at ($(u2)!0.5!(un)$) {$\cdots$};
        \node (pd) at ($(p2)!0.5!(pn)$) {$\cdots$};

        annotations
        \node[text width=\introTextboxWidth] (text1) [right of=tn, xshift=\introTextboxShift, yshift=0.60cm]
        {};
        \node[text width=\introTextboxWidth] (text1') [yshift=0.2cm] at (text1)
        {(i)};
       \node[text width=\introTextboxWidth] (text2) [right of=un, xshift=\introTextboxShift, yshift=0.4cm]
        {(ii)};
        \node[text width=\introTextboxWidth] (text3) [right of=pn, xshift=\introTextboxShift, yshift=0.40cm]
        {(iii-a)};

        \newcommand{\drawBrace}[2]{%
          \node (brace1-#1) at (#1) [xshift=\introBraceShift, yshift=-#2] {};
          \node (brace2-#1) at (#1) [xshift=\introBraceShift, yshift=#2] {};
          \draw[decorate,decoration={calligraphic brace,amplitude=5pt,mirror,raise=3pt},yshift=0pt,line width=1.25pt] (brace1-#1)--(brace2-#1);
        }
        \drawBrace{text1}{0.5cm}
        \drawBrace{text2}{0.5cm}
        \drawBrace{text3}{0.55cm}
      \end{tikzpicture}
    \caption{
      Top-down\\ 
      (i) Projection 
      (\Sec\ref{sec:projections})\\ 
      (ii) Subtype checking (\Sec\ref{section:subtyping})\\
      (iii-a) Type checking
    }
    \label{fig:top_down}
\end{center}
  \end{subfigure}%
\begin{subfigure}[T]{.30\textwidth}
\begin{center}
     \small  
      \begin{tikzpicture}[>=Latex, node distance=1.2cm]
        \node (gt)[smallrect] {Global type $\G$};
        \node (t1) [below of=gt, xshift=-0.9cm] {$\T_1$};
        \node (t2) [below of=gt, xshift=-0.2cm] {$\T_2$};
        \node (tn) [below of=gt, xshift=0.9cm] {$\T_n$};
        \node (u1) [below of=t1, yshift=0.5cm] {$\Tmin_1$};
        \node (u2) [below of=t2, yshift=0.5cm] {$\Tmin_2$};
        \node (un) [below of=tn, yshift=0.5cm] {$\Tmin_n$};
        \node (p1) [below of=u1, yshift=0.3cm] {$\PP_1$};
        \node (p2) [below of=u2, yshift=0.3cm] {$\PP_2$};
        \node (pn) [below of=un, yshift=0.3cm] {$\PP_n$};
        
        \node[rotate=90] (s1) at ($(t1)!0.5!(u1)$) {{\large $\subt$}};
        \node[rotate=90] (s2) at ($(t2)!0.5!(u2)$) {{\large $\subt$}};
        \node[rotate=90] (sn) at ($(tn)!0.5!(un)$) {{\large $\subt$}};

        \draw[->] (gt) -- (t1);
        \draw[->] (gt) -- (t2);
        \draw[->] (gt) -- (tn);
    
        \draw[->] (p1) -- (u1);
        \draw[->] (p2) -- (u2);
        \draw[->] (pn) -- (un);

        \node (td) at ($(t2)!0.5!(tn)$) {$\cdots$};
        \node (ud) at ($(u2)!0.5!(un)$) {$\cdots$};
        \node (pd) at ($(p2)!0.5!(pn)$) {$\cdots$};

        annotations
        \node[text width=\introTextboxWidth] (text1) [right of=tn, xshift=\introTextboxShift, yshift=0.60cm]
        {};
        \node[text width=\introTextboxWidth] (text1') [yshift=0.2cm] at (text1)
        {(i)};
       \node[text width=\introTextboxWidth] (text2) [right of=un, xshift=\introTextboxShift, yshift=0.4cm]
        {(iii)};
        \node[text width=\introTextboxWidth] (text3) [right of=pn, xshift=\introTextboxShift, yshift=0.40cm]
        {(ii-b)};

        \newcommand{\drawBrace}[2]{%
          \node (brace1-#1) at (#1) [xshift=\introBraceShift, yshift=-#2] {};
          \node (brace2-#1) at (#1) [xshift=\introBraceShift, yshift=#2] {};
          \draw[decorate,decoration={calligraphic brace,amplitude=5pt,mirror,raise=3pt},yshift=0pt,line width=1.25pt] (brace1-#1)--(brace2-#1);
        }
        \drawBrace{text1}{0.5cm}
        \drawBrace{text2}{0.5cm}
        \drawBrace{text3}{0.55cm}
      \end{tikzpicture}
    \caption{
      Top-down\\ 
      (i) Projection 
      (\Sec\ref{sec:projections})\\
      (ii-b) Type inference (\Sec~\ref{section:minimum_typing})\\
      (iii) Subtype checking\\
    }
        \label{fig:top_downtwo}
    \end{center}
  \end{subfigure}%
\begin{subfigure}[T]{.33\textwidth}
\begin{center}
  \small
      \begin{tikzpicture}[>=Latex, node distance=1.2cm]
        \node (gt)[smallrect] {Property $\varphi$};
        \node (u1) [below of=t1, yshift=1.2cm] {$\Tmin_1$};
        \node (u2) [below of=t2, yshift=1.2cm] {$\Tmin_2$};
        \node (un) [below of=tn, yshift=1.2cm] {$\Tmin_n$};
        \node (p1) [below of=u1, yshift=-0.4cm] {$\PP_1$};
        \node (p2) [below of=u2, yshift=-0.4cm] {$\PP_2$};
        \node (pn) [below of=un, yshift=-0.4cm] {$\PP_n$};
        

        \draw[->] (u1) -- (gt);
        \draw[->] (u2) -- (gt);
        \draw[->] (un) -- (gt);
    
        \draw[->] (p1) -- (u1);
        \draw[->] (p2) -- (u2);
        \draw[->] (pn) -- (un);

        \node (ud) at ($(u2)!0.5!(un)$) {$\cdots$};
        \node (pd) at ($(p2)!0.5!(pn)$) {$\cdots$};

        annotations
       \node[text width=\introTextboxWidth] (text2) [right of=un, xshift=\introTextboxShift, yshift=0.6cm]
        {(ii)};
        \node[text width=\introTextboxWidth] (text3) [right of=pn, xshift=\introTextboxShift, yshift=0.75cm]
        {(i)};

        \newcommand{\drawBrace}[2]{%
          \node (brace1-#1) at (#1) [xshift=\introBraceShift, yshift=-#2] {};
          \node (brace2-#1) at (#1) [xshift=\introBraceShift, yshift=#2] {};
          \draw[decorate,decoration={calligraphic brace,amplitude=5pt,mirror,raise=3pt},yshift=0pt,line width=1.25pt] (brace1-#1)--(brace2-#1);
        }
        \drawBrace{text2}{0.6cm}
        \drawBrace{text3}{1.0cm}
      \end{tikzpicture}
    \caption{
      Bottom-up\\ 
     (i) Type inference (\Sec\ref{section:minimum_typing})\\
      (ii) Property checking 
      (\Sec\ref{section:modelchecking})\\
      safety, deadlock-freedom, liveness 
      \label{fig:bottom_up}
    }
 \end{center}
 \end{subfigure}%
\end{center}
\vspace{-0.6cm}
\caption{Overview of the top-down and bottom-up procedures.
}
\label{fig:overview}
\end{figure}
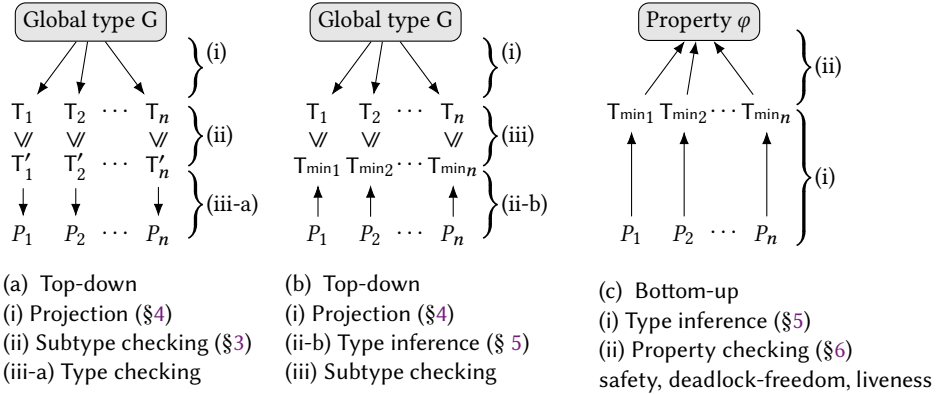

Motivated by a growth of model checking tools, 
Scalas and Yoshida \cite{Scalas2019} proposed a general MPST 
typing system which does \emph{not} require global types:  
it can ensure desired properties of processes 
by checking the same properties of typing contexts (a collection of 
local types). 
They give specifications for these properties in modal
$\mu$-calculus formulae checkable by the mCRL2 checker \cite{bunte2019mcrl2}. 
Their typability is not limited by 
projectability or 
implementability of 
global types, 
giving a \emph{complete} set of processes 
which satisfy the desired properties. 
We call their framework, \emph{bottom-up}. 

\myparagraph{This paper} gives 
detailed and comprehensive complexity analyses of the \emph{top-down} and 
\emph{bottom-up} procedures developed in 
various MPST literature. 
Our motivation is very simple: 
the top-down approach 
gives often less typable processes    
than the bottom-up, but it is
useful in practice and believed to be more tractable.
Can we measure and compare the two approaches precisely? 

Figure~\ref{fig:top_down} gives a workflow of the top-down approach: 
first, we project a global type $\G$ to a set of 
local types $\{\T_i\}_{i\in I}$ (Step (i)).
Each process $P_i$ is then type-checked against a subtype $\T_i'$ of 
$\T_i$ such that $\T_i\subt \T_i'$ (Steps (ii) and (iii-a))
\cite{YH2024}. 
The bottom-up approach \cite{Scalas2019}  
is given in Figure~\ref{fig:bottom_up}. 
Here no global type is used;
instead, we first infer  
the \emph{minimum local type} $\Tmin$
from each process $\PP_i$ 
by the \emph{type inference system} (Step (i)). 
The minimum type $\Tmin$ of $\PP$ characterises, in a succinct way,
the tightest behaviour of $\PP$.
In Step (ii), 
we directly check whether a property $\varphi$ is satisfied or not
against a collection of the minimum types.  This guarantees
a set of processes satisfy property $\varphi$. 


\myparagraph{Road map.\ }
For Step (i) of the top-down, we consider four kinds of projections, 
and give analyses of their complexity:
\emph{inductive plain merge}, introduced in the first MPST paper
\cite{Honda2008,HYC2016}
(Definition~\ref{def:inductive_plain_merging}, \Sec\ref{subsec:plainmerge}), 
\emph{inductive full merge}, refined from plain merge  
to obtain a larger set of
valid global types (Definition~\ref{def:inductive_full_merging}, \Sec\ref{subsec:fullmerge}), 
a recent projection algorithm proposed 
by Tirore \etal~\cite{Tirore2023}, called \TBC, which is 
sound and complete against the coinductive
projection (\Sec\ref{section:coinductive_projection}) 
and a new projection algorithm which is sound and complete 
against the coinductive projection with full merge \cite{Ghilezan2019}
(\Sec\ref{subsec:coinductivefull}). 
For Step (ii), we develop a subtyping system
based on a graph representation of types (called \emph{type graphs}), 
which is more efficient than the existing inductive algorithm 
in \cite{Ghilezan2019} (\Sec\ref{section:subtyping}).
This type graph representation turns out
useful for building algorithms for other complexity analyses. 

Step (iii-a) of the top-down is type-checking each process by 
some subtype $\T_i'$ of its projection $\T_i$.
This problem is reducible to 
subtype-checking 
with the minimum type $\Tmin_i$ inferred from the process 
($\Tmin_i\subt \T_i$) since the projected type $\T_i$ must be a
supertype of $\Tmin_i$. 
Hence the total complexity of the top-down 
is analysed as described in Figure~\ref{fig:top_downtwo}. 

In \Sec\ref{section:minimum_typing},
we give a sound and complete algorithm for finding a minimum type
by building \textit{type constraints} (Definition \ref{def:constraint_rules})
and a \textit{minimum type graph} (Definition
\ref{def:minimum_type_graph_of_process}) which yields the minimum type
of a given process. We then analyse its complexity. 
This corresponds to Step (ii-b) of the top-down
(Figure~\ref{fig:top_downtwo}) 
and Step (i) in the bottom-up (Figure~\ref{fig:bottom_up}).
In \Sec\ref{section:modelchecking}, we analyse the complexity to check
whether a set of local types satisfy 
\emph{safety} (no type nor communication error),
\emph{deadlock-freedom} (no stuck)  
and \emph{liveness} (all pending actions can always 
eventually communicate with their dual actions under fair paths). 
In \Sec\ref{sec:related} we measure 
the total of complexities of the top-down 
(Steps (i,iii,ii-b) in Figure~\ref{fig:top_downtwo})
and the bottom-up 
(Steps (i,ii) in Figure~\ref{fig:bottom_up}).    

To examine the complexity analyses  
covering a wide range of the literature of session types, 
we focus on the simplest calculus from \cite{Ghilezan2019}, namely  
a {\em synchronous single multiparty session without type
 annotations}.
This choice is made because (1) the bottom-up approach is 
\emph{undecidable} if the calculus is asynchronous (communications use 
infinite FIFO queues); and 
(2) most of the binary and multiparty session calculi studied in the literature
and a textbook \cite{GV2025} are synchronous (both in theory and practice). 
See \Sec~\ref{sec:related} for more discussions on synchronous and asynchronous MPST.


\myparagraph{Contributions. }
To our best knowledge, all the complexity results proven in this
paper were unknown. 

To calculate the total complexity of the top-down and the bottom-up
approaches, we build the three new systems currently missing from the
workflows: \textbf{(1)} a type inference system for MPST 
(\Sec\ref{section:minimum_typing}); \textbf{(2)} a quadratic subtyping
algorithm (\Sec\ref{section:subtyping}): and \textbf{(3)} a sound and
complete \PSPACE-algorithm for a coinductive projection with full
merging (\Sec\ref{subsec:coinductivefull}).

Throughout this paper, we use different formats of \emph{local and
global type graphs} to construct the new systems and measure the
complexity. 
Specifically, the type inference system (1) 
uses a new form of type graphs (Definition~\ref{def:minimum_type_graph_of_process}) and a relationship between them (Lemma~\ref{def:subsets_are_subtypes}) 
for finding the minimum 
type $\Tmin$ of $\PP$.     
The algorithm generates the \emph{minimum
type graph}, taking care of contravariance of
branching (external choice) and covariance of selections (internal
choice).
 
The subtyping system (2) (Definition~\ref{def:type_simulation}) uses 
a product graph of two local type graphs 
ensuring more efficiency (quadratic) (Theorem~\ref{thm:quadratic_subtyping})
than the existing  (exponential) inductive algorithm 
\cite{Ghilezan2019} (Theorem~\ref{thm:complexity_ghilezan_subtyping}).
(3) is constructed by a \emph{projection graph} of $\G$ whose node is 
a set of subglobal types of $\G$.  
Apart from (1)-(3) listed above, we also use 
a product of global and local type graphs for complexity analysis 
of a projection algorithm by Tirore \etal~\cite{Tirore2023} (Theorem~\ref{thm:coinductiveprojection}). 

For the complexity analysis of the projection algorithm with 
\emph{inductive full merging}, if we use 
a na\"ive approach, 
the complexity stays as quadratic as 
\emph{inductive plain merging} (Theorem~\ref{thm:inductive_plain_merge}). 
This is improved 
to subquadratic time (Theorem~\ref{thm:naive_full_merging_algorithm_is_quadratic}) by  
applying the \emph{balanced binary tree search technique} to 
the choices (branching and selection) in the session types. 
The proof requires a detailed analysis. 

Another main contribution is the complexity analysis of the
type context property-checking which was (as far as we have known)
never tackled before our work (Figure~\ref{fig:bottom_up}(ii)).
One of the most non-trivial proofs in the paper is \PSPACE-completeness 
for the property checking. 

For \PSPACE-membership (\Sec\ref{subsec:pspacecomplete}), 
safety and deadlock-freedom are proven as a reachability problem 
of checking unsafe and deadlock states, respectively. 
For liveness, 
we give the equivalent characterisation 
by constructing the \emph{counterwitness for liveness} using 
observational and barbed predicates on the typing contexts. 
With this, we construct a nondeterministic polynomial-space algorithm that aspects if and only if the input action is not live. 
We further give deterministic algorithms for safety, deadlock-freedom and 
liveness that run in exponential time. 
For liveness, we use a similar construction to the proof of Savitch's theorem
\cite[Theorem 8.5]{Sipser2012}. 

For \PSPACE-hardness (\Sec\ref{subsec:pspacehard}), 
we reduce it from the \emph{quantified Boolean formula (QBF) problem}. 
We construct a very large typing context that reaches an undesirable state if 
a quantified boolean formula is true. 
The complexity analysis of 
the liveness checking is most involved, 
resulting the liveness incurring the highest complexity among the three properties. 

Our proof methods for \PSPACE-completeness are uniformly applicable 
to analyse other properties such as termination and other forms of liveness 
studied in \cite{Scalas2019}. 

The detailed complexity results of the subtyping checking,  
type-level and total complexity of the top-down
and the bottom-up approaches are summarised in
Table~\ref{tab:subtyping_complexity_summary} (subtyping), 
Table~\ref{table:projection} (top-down) 
and Table~\ref{tab:bot_complexity} (bottom-up).

\iffull{Appendices include\ }%
\else{The full version \cite{UY2024} includes\ }\fi
detailed full proofs, omitted algorithms, examples and definitions.

\section{Multiparty Session Calculus}
\label{sec:calculus}
This section introduces the syntax and semantics of a 
multiparty synchronous session calculus (MPST-calculus)
\cite{Ghilezan2019} and define its properties.   

\myparagraph{Syntax}
\emph{Expressions} ($\e, \e', \dots$) are variables ($x,y,z,\dots$), 
a value $\val$, which is either a natural number $\valn$, an integer
$\valr$,
or a boolean $\true$/$\false$, 
or a term built from expressions by applying operators such as 
${\tt neg}$, $\neg$, $\oplus$ or $+$.
The operator $\oplus$ models non-determinism: %
$\e_1 \oplus \e_2$ is an expression %
that might yield either $\e_1$ or $\e_2$.\\[1mm]%
\centerline{
$
\e ::= \true \SEP \false \SEP \valn \SEP \vali \SEP x 
     \SEP \e \vee \e \SEP \neg \e \SEP \e + \e \SEP \e \oplus \e \SEP \fsqrt{\e}
$}\\[1mm]
\emph{Processes} ($\PP, \Q, \dots$) and \emph{multiparty sessions} ($\N, \N', \dots$) are defined by the following syntax:\\[1mm] 
\centerline{
$\begin{array}{rll}
\PP   ::= & \inact  \SEP \procout \pp{\e} \PP \SEP \procin \pp{x} \PP \SEP 
\procsel \pp{l}{\PP} \SEP \procbrasub \pp{l_i: \PP_i}{i\in \I}& \text{\footnotesize (nil, output, input, selection, branch)} \\
     \SEP & \cond{\e}{\PP}{\PP} \SEP \mu X. \PP \SEP X& \text{\footnotesize (conditional, recursion,  proc variable)}\\
\N  ::=   & \pa \pp \PP  \SEP \N \pc\N & \text{\footnotesize (role, parallel composition)}
\end{array}
$}\\[1mm]
$\pp, \pq, \pr, \ps, \dots$ are \emph{participants} from a countably infinite set. The nil $\inact$ denotes termination; 
the output process $\procout{\pp}{\e}{\Q}$ sends the value of expression $\e$
to participant $\pp$; the input process $\procin \pp{x} \PP$ 
inputs a value from $\pp$; 
the selection $\procsel \pp{l}{\PP}$ selects a label $l$ at participant $\pp$ (internal choice); the branching 
$\procbrasub \pp{l_i: \PP_i}{i\in \I}$ 
(often denoted by $\procbra \pp{l_1: \PP_1, \dots, l_n: \PP_n}$)
is a process that can accept label $ l_i $ from participant $ \pp$  for any $ i\in I\neq \emptyset$ (external choice).
Conditional process $\cond{\e}{\PP}{Q}$ 
and a recursive process $\mu X.\PP$ are standard. 
We assume that the recursive processes are \emph{guarded}, i.e., 
$\mu X.X$ is not a valid process, but 
$\mu X.\procout{\pp}{\e}{X}$ is valid. 
Process $\PP$ plays the role of participant $\pp$ (denoted: $\pa\pp\PP$),
and can interact with other processes playing other roles.
We assume $\pc$ is commutative and associative, 
and write  $\prod_{i \in \I}\pa{\pp_i}{\PP_i}$ for 
$\pa{\pp_1}{\PP_1}\pc\cdots \pc \pa{\pp_n}{\PP_n}$ with $I=\set{1,\dots,n}$. 

\myparagraph{{Reductions}}
The value $\val$ of expression $\e$ (notation $\eval\e\val$) is
computed as expected
\iffull{(Table~\ref{tab:evaluation} in
  Appendix~\ref{app:calculus}).}\else{\cite[Table 1]{Ghilezan2019}.}\fi
The internal choice $\e_1\oplus\e_2$ evaluates either to the value of $\e_1$ or to the value of $\e_2$.  The {\em reduction rules of multiparty sessions} are given in Figure~\ref{fig:reduction_sessions}.
\begin{figure}
  \begin{center}
\small
$\begin{array}{c}
\begin{array}{@{}rll}
\rulename{r-comm} & 
\pa\pp \procin \q{x} \PP \ | \ \pa \q\procout \pp{\e} \Q  \; \pc\;\N \red 
\pa\pp{\PP}\sub{\val}{\x}\;\pc\;\pa\q\Q \; \pc\;\N
& (\eval{\e}{\val})\\
\rulename{r-bra} & 
\pa\pp\procsel \q{l}{\PP} \pc \pa\q\procbrasub \pp{l_i: \PP_i}{i\in \I}
\; \pc\;\N
\red 
\pa\pp{\PP} \pc \pa\q \PP_k
\; \pc\;\N
& (k\in \I)\\
\rulename{t-cond}
&{
    \pa\pp{\cond{\e}{\PP}{\Q}} \; \pc \;  \N \red \pa\pp\PP\; \pc \;  \N
   }
& (\eval{\e}{\true}) \\
\rulename{f-cond}
&{
    \pa\pp{\cond{\e}{\PP}{\Q}} \; \pc \;  \N \red \pa\pp\Q\; \pc \;  \N
   }
& (\eval{\e}{\false}) \\
\rulename{r-str} 
& 
\N'_1\prestruct \N_1 \quad \N_1\red \N_2 \quad \N_2 \prestruct \N'_2
 \ \Longrightarrow \ 
\N'_1 \red \N'_2\\
\rulename{v-err} 
& 
\pa\pp{\cond{\e}{\PP}{\Q}} \; \pc \; \N \red \error 
& (\eval{\e}{\val} \text{ and } \val\not\in \{\true,\false\}) \\
\rulename{c-err} 
& 
\pa\pp\procsel \q{l}{\PP} \pc \pa\q\procbrasub \pp{l_i: \PP_i}{i\in \I}
\; \pc\;\N
\red \error & (\forall i\in \I, l\not=l_k)\\
\end{array}
\end{array}
$
\vspace{-3mm}
\end{center}
\caption{Reductions of sessions}
\label{fig:reduction_sessions}
\end{figure}

Rule \rulename{t-comm} defines a communication between 
the input and output processes; in \rulename{t-bra}, the selection 
process chooses label $l$ at the branching; 
rules \rulename{t,f-cond} are standard; and 
rule \rulename{t-str} defines that a reduction is 
closed under \emph{structure rules} $\prestruct$ (preorder) defined by 
$\pa\pp\mu X.\PP \pc \N \prestruct \pa\pp\PP\sub{\mu X.\PP}{X} \pc \N$, 
and $\prod_{i \in \I}\pa{\pp_i}{\PP_i} \prestruct 
\prod_{j \in \J}\pa{\pp_j}{\PP_j}$ where $I$ is a permutation of $J$. 
Rules \rulename{v-err} and \rulename{s-err} define 
value and session errors, respectively.   
We use $\red^\ast$ as the reflexive and transitive closure of $\red$. 

\myparagraph{Properties} 
We recall various desirable properties of a session $\M$  
from \cite[Definition 5.1]{Scalas2019}. 

\begin{definition}[Properties]\label{def:properties}
  A session $\M$ is:
  \begin{enumerate}
    \item \emph{communication safe} iff 
$\M \red^* \M'$, then $\M'\not\prestruct\error$. 
    \item \emph{deadlock-free} iff $\M \red^* \M' \not\red$ implies 
$\M' \prestruct \prod_{i \in I}\pp_i :: \inact$.
    \item \emph{live} iff $\M \red^* \M' \prestruct \M_1 \pc \pp :: \PP$ implies:
    \begin{enumerate}
      \item If $\PP = \procin{\pq}{x} \PP'$ then there exists $\M_2, \val$ such that $\M' \red^* \M_2 \pc \pp :: \PP'\sub{\val}{x}$.
      \item If $\PP = \procout{\pq}{\e} \PP'$ then there exists $\M_2$ such that $\M' \red^* \M_2 \pc \pp :: \PP'$.
      \item If $\PP = \procsel{\pq}{l} \PP'$ then there exists $\M_2$ such that $\M' \red^* \M_2 \pc \pp :: \PP'$.
      \item If $\PP = \procbrasub{\pq}{l_i: \PP_i}{i \in \I}$ then there exists $\M_2, i$ such that $\M' \red^* \M_2 \pc \pp :: \PP_i$.
    \end{enumerate}
  \end{enumerate}
\end{definition}
\noindent A session $\M$ is 
\emph{communication safe} 
if there is neither type and label errors; 
\emph{deadlock-free} when it only stops 
reducing to nil; 
$\M$ is \emph{live} when all
the pending actions can always 
eventually communicate.

\section{Complexity Analysis of Multiparty Session Subtyping}
\label{section:subtyping}
We start from complexity analysis of the two subtyping algorithms,
introducing \emph{type graphs}. 
The first algorithm is the inductive algorithm of 
multiparty synchronous session subtyping given in 
\cite{Ghilezan2019}; the second
one is a new coinductive algorithm for multiparty, extending 
the algorithm for binary session types in \cite{Udomsrirungruang2024a}.

\myparagraph{Local Multiparty Session Types} ($\T, \T', \dots$) give a
local specification 
of each participant, which are defined with \emph{sorts} (base types) ($\S, \S',
\dots$). \\[1mm]
\centerline{
$\begin{array}{c}
\S ::=\: \tbool \SEP \tnat \SEP \tint \qquad 
\T ::=  \tend \SEP \tout \pp{S} \T \SEP \tin \pp{S} \T 
         \SEP \tselsub{\pp}{l_i: \T_i}{i\in \I}
         \SEP \tbrasub{\pp}{l_i: \T_i}{i\in \I}
         \SEP \mu \ty. \T 
         \SEP \ty 
\end{array}
$}\\[1mm]
Sorts include base (atomic) types; in general, the choice
of sorts is unimportant and our system will work with arbitrary atomic
sorts. 
Type $\tend$ represents a termination; $\tout \pp{S} \T$ is the \emph{output} 
to participant $\pp$ with 
its value typed by $S$; $\tin \pp{S} \T$ is its dual \emph{input};
$\tselsub{\pp}{l_i: \T_i}{i\in \I}$ (with $I \neq \emptyset$) 
is a \emph{selection type} which selects one of the labels $l_i$
and $\tbrasub{\pp}{l_i: \T_i}{i\in \I}$ (with $I \neq \emptyset$) is
a \emph{branching type} which offers branches with label $l_i$.
We assume all labels are distinct, and often write 
$\tsel \pp{l_1: \T_1, \dots, l_n: \T_n}$ or 
$\tbra \pp{l_1: \T_1, \dots, l_n: \T_n}$.   
$\mu \ty. \T$ and $\ty$ are recursive types and type variables.  
We often omit $\tend$ and write $\tselset{\pp}{l_1:\T_1}$ 
for $\tsel{\pp}{l_1:\T_1}$. 
We assume a countably infinite set of recursive type variables and, 
as with processes, all recursive types to be guarded. 
We denote $\ftv{\T}$ as a set of \emph{free type variables} in $\T$,
and call $\T$ \emph{closed} if $\ftv{\T}=\emptyset$. 
We denote $\Sort$ (resp. $\Type$) as the set of all sorts
(resp. closed local types).

We calculate the complexity against the size of a local type;   
and use the set of \emph{subformulas} of $\T$.



\begin{definition}[Size of Local Types]\label{def:size_of_local_type}
The \emph{size} of a local type $\T$, denoted $|\T|$, is defined as:
$|\tend| = |\ty| = 1$, 
$|\mu \ty. \T| = |\T| + 1$, 
$|\tin\pp\S\T| = |\tout\pp\S\T| = |\T| + 1$ and  
$|\tselsub{\pp}{l_i: \T_i}{i\in \I}| = |\tbrasub{\pp}{l_i: \T_i}{i\in
  \I}| = \sum_{i\in \I} |\T_i| + 1$.    
\end{definition}
\begin{definition}[Subformulas of Local Types]\label{def:subformulas_of_local_type} The set of \emph{subformulas}
  of a local type $\T$, denoted $\Sub(\T)$, is defined as: 
   $\Sub(\tend) = \{\tend\}$, 
$\Sub(\tout\pp{\S}\T) = \{\tout\pp{\S}\T\} \cup \Sub(\T)$, 
$\Sub(\tin\pp{\S}\T) = \{\tin\pp{\S}\T\} \cup \Sub(\T)$, 
$\Sub(\tselsub{\pp}{l_i: \T_i}{i\in \I}) = \{\tselsub{\pp}{l_i: \T_i}{i\in \I}\}
  \cup \bigcup_{i\in \I} \Sub(\T_i)$, 
$\Sub(\tbrasub{\pp}{l_i: \T_i}{i\in \I}) = \{\tbrasub{\pp}{l_i: \T_i}{i\in \I}\}
  \cup \bigcup_{i\in \I} \Sub(\T_i)$, 
   $\Sub(\ty) = \{\ty\}$ and  
   $\Sub(\mu \ty. \T) = \{\mu \ty. \T\} \cup \{\T'\sub{\mu \ty.\T}{\ty}
  \mid \T' \in \Sub(\T)\}$.  

\end{definition}

\myparagraph{Subtyping.}
We use the (coinductive) synchronous multiparty subtyping relation
(denoted by $\T \subt \T'$) following
\cite{DemangeonH11,Ghilezan2019,Chen2017}
(called \emph{process subtyping} 
in \cite{Gay16}):
when a type $\T$ is ``smaller'' than $\T'$, %
it is allowed to use a process typed by the former %
whenever a process typed by the latter %
is required.
The original subtyping in \cite{Gay2005} 
is called \emph{channel subtyping} in \cite{Gay16}, and  
the complexity results in this paper do not change
whether it is process- or channel-subtyping. 




We use the following unfolding function from \cite{Gay2005}.

\begin{definition}[Unfolding function]\label{def:unfold}
Define the \emph{unfold} function by: $\unfold{\mu \ty. \T} = \unfold{\T\sub{\mu \ty. \T }{ \ty}}$, and $\unfold{\T} = \T$ otherwise.
\end{definition}
\begin{example}[Unfolding]
\label{ex:unfold}
Let $\T_1 = \mu\ty. \tsel\pp{l_1: \tsel\pp{l_1: \ty}, l_2: \tend}$. 
Then we have:\\
\centerline{
$\begin{array}{rl}
\unfold{\T_1} = & 
\unfold{\tsel\pp{l_1: \tsel\pp{l_1: \ty}, l_2: \tend}\sub{\T_1}{\ty}}
= \unfold{\tsel\pp{l_1: \tsel\pp{l_1: \T_1}, l_2: \tend}}\\
= & \tsel\pp{l_1: \tsel\pp{l_1: \T_1}, l_2: \tend}
\end{array}$} 
\end{example}

We recall the standard definition of coinductive synchronous subtyping 
from \cite{Gay2005}. 

\begin{definition}[Subtyping]\label{def:subtyping_local_types}
The \emph{subtyping relation} $\subt$ is the largest relation between session types such that, for all $\T_1 \subt \T_2$:
\begin{itemize}[leftmargin=*]
  \item If $\unfold{\T_1} = \tend$ then $\unfold{\T_2} = \tend$.
  \item If $\unfold{\T_1} = \tout{\pp}{\S}{\T_1'}$ then $\unfold{\T_2} = \tout{\pp}{\S}{\T_2'}$ and $\T_1' \subt \T_2'$.
  \item If $\unfold{\T_1} = \tin{\pp}{\S}{\T_1'}$ then $\unfold{\T_2} = \tin{\pp}{\S}{\T_2'}$ and $\T_1' \subt \T_2'$.
  \item If $\unfold{\T_1} = \tselsub{\pp}{l_i: \T_{1i}}{i\in \I}$ then $\unfold{\T_2} = \tselsub{\pp}{l_j: \T_{2j}}{j\in \J}$ and $\I \subseteq \J$ and $\forall i\in \I.\:\T_{1i} \subt \T_{2i}$.
  \item If $\unfold{\T_1} = \tbrasub{\pp}{l_i: \T_{1i}}{i\in \I}$ then $\unfold{\T_2} = \tbrasub{\pp}{l_j: \T_{2j}}{j\in \J}$ and $\J \subseteq \I$ and $\forall j\in \J.\:\T_{1j} \subt \T_{2j}$.
\end{itemize}
\end{definition}
Notice that the selection subtyping is covariant and 
the branching subtyping is contravariant 
with respect to the indices of choices.
The relation $\subt$ coincides with the coinductive 
synchronous subtyping formulation in 
\cite[Definition~3.15]{Ghilezan2019}.

We represent types syntactically as labelled transition systems
where each node is a subformula of the type. Type graphs are
used throughout the paper.

\begin{definition}[Type graph]\label{def:typegraph}
Define \emph{actions}: 
${\gell, \gell',\dots ::= \ \trin{\pp}{\S} 
\SEP \trout{\pp}{\S} \SEP \trbra{\pp}{l} \SEP
\trsel{\pp}{l} \SEP \trend}$ 
and the \emph{labelled transition relation} (LTS), $\T \trans{\gell}
\T'$, 
where $\gell$ is an action and $\T'$ is extended with the additional
type $\Skip$:\\[2mm]
\centerline{\small
$\begin{array}{c}
\mbox{
\begin{prooftree}
    \hypo{\unfold{\T} = \tend}
    \infer1{\T \trans{\trend} \Skip}
 \end{prooftree}\qquad
  \begin{prooftree}
   \hypo{\unfold{\T} = \tdag\pp {\S} \T' \quad \dagger\in \{ !,?\}}
   \infer1{\T \trans{\trdag{\pp}{\S}} \T'}
 \end{prooftree}\qquad
   \begin{prooftree}
 \hypo{\unfold{\T} = 
\tbradag{{\pp}}{{l_i: {\T_i}}}{i\in \I}  \quad \dagger \in \{ \oplus, \&\} \quad j \in \I}
 \infer1{\T \trans{\trbradag{\pp}{l_j}}{\T_j}}
   \end{prooftree}
}
     \end{array}$}\\[2mm]
The type graph is a directed graph whose edges are local types or 
$\Skip$, and an edge is given by $\T\trans{\gell} \T'$
(from $\T$ to $\T'$ labelled by $\gell$).
The type graph for $\T$, denoted by $\GG(\T)$, 
is the graph reachable by the transitions from $\T$.
We often write $\T$ for $\GG(\T)$ when it is clear from the context. 
\end{definition}
\begin{example}[Transitions]
\label{ex:typegraph}
Recall $\T_1$ in Example~\ref{ex:unfold}. 
Then we have:
$\T_1 \trans{\trsel{\pp}{l_1}} 
\tsel\pp{l_1: \T_1}
\trans{\trsel{\pp}{l_1}} 
\T_1$ and 
$\T_1 \trans{\trsel{\pp}{l_2}} 
\tend
\trans{\trend} \Skip
$. 
Its type graph $\GG(\T)$ is in Figure \ref{fig:transitions_type_graph}.
\end{example}
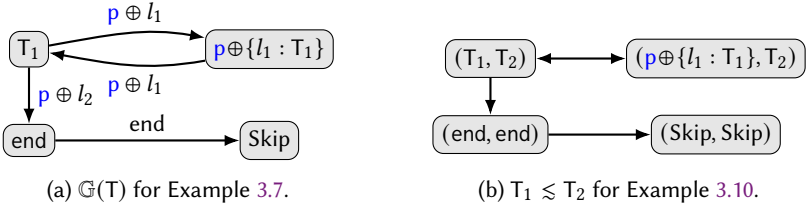
\begin{figure}[t]
\begin{center}
\begin{subfigure}[B]{0.40\textwidth}  
\centering
\small
  \tikzstyle{arrow} = [thick,-{Latex[length=2mm, width=1.5mm]},>=stealth]
  \tikzstyle{arrow2} = [thick,{Latex[length=2mm, width=1.5mm]}-{Latex[length=2mm, width=1.5mm]},>=stealth]
  \begin{tikzpicture}[>=Latex, node distance=1.2cm]
    \node (start) [smallrect] {$\T_1$};
    \node (n1) [smallrect, right of=start, xshift=2cm] {$\tsel{\pp}{l_1: \T_1}$};
    \node (n2) [smallrect, below of=start] {$\tend$};
    \node (n3) [smallrect, right of=n2, xshift=2cm] {$\Skip$};
    \draw[arrow] (start) to[out=10, in=170] node[midway, auto, xshift=6mm] {$\trsel{\pp}{l_1}$} (n1);
    \draw[arrow] (n1) to[out=-170, in=-10] node[midway, auto, xshift=6mm] {\ $\trsel{\pp}{l_1}$} (start);
    \draw[arrow] (start) to node[midway, auto] {$\trsel{\pp}{l_2}$} (n2);
    \draw[arrow] (n2) to node[midway, auto] {$\trend$} (n3);
  \end{tikzpicture}
  \caption{$\GG(\T)$ for Example~\ref{ex:typegraph}.}
  \label{fig:transitions_type_graph}
\end{subfigure}
\begin{subfigure}[B]{.45\textwidth}
\begin{center}
    \small
    \begin{tikzpicture}[>=Latex, node distance=1cm]
      \node (start) [smallrect] {$(\T_1, \T_2)$};
      \node (n1) [smallrect, right of=start, xshift=2cm] {$(\tsel\pp{l_1: \T_1}, \T_2)$};
      \node (n2) [smallrect, below of=start] {$(\tend, \tend)$};
      \node (n3) [smallrect, right of=n2, xshift=2cm] {$(\kf{Skip}, \kf{Skip})$};
      \draw[arrow2] (start) -- (n1);
      \draw[arrow] (start) -- (n2);
      \draw[arrow] (n2) -- (n3);
    \end{tikzpicture}
  \end{center}
\caption{$\T_1\subtsim\T_2$ for Example~\ref{ex:typesimulation}.\label{fig:typesimulation}}
\end{subfigure}
\vspace{-3mm}
\caption{Local type graphs}
\label{fig:typegraph}
\vspace{-3mm}
\end{center}
\end{figure}

We redefine the subtyping relation equivalently in terms of a \emph{simulation}
between the type graphs.
\begin{definition}[Type simulations]\label{def:type_simulation}
A \emph{type simulation} $\cR$ is a relation on 
type graphs $\T_1$ and $\T_2$ such that 
$\T_1 \cR \T_2$ implies: 
\begin{itemize}
  \item If $\T_1 \trans{\gell} \T_1'$ then $\T_2 \trans{\gell} \T_2'$,
    and $\T_1' \cR \T_2'$, for
    $\gell \in \{\trsel\pp l, \tend, \trin{\pp}{S}, \trout{\pp}{S}\}$.
  \item If $\T_2 \trans{\gell} \T_2'$ then $\T_1 \trans{\gell} \T_1'$,
    and $\T_1' \cR \T_2'$, for
    $\gell \in \{\trbra\pp l, \tend\}$.
\end{itemize}
The subtyping simulation $\subtsim$ is defined by $\T_1 \subtsim \T_2$ 
if $(\T_1, \T_2) \in \cR$ 
in some type simulation $\cR$, i.e.~subtyping is the largest type
simulation.
We write $\subteq$
for the type graph equivalence.
\end{definition}
\begin{restatable}[Subtyping as a type
    simulation]{lemma}{subtypingtypesimulation}\label{thm:subtyping_as_simulation}
  \proofreference{Appendix~\ref{app:subtypingtypesimulation}}
$\T_1 \subt \T_2$ iff $\T_1 \subtsim \T_2$. 
\end{restatable}



\begin{example}[Type simulation]\label{ex:typesimulation}
Recall $\T_1 = \mu\ty. \tsel\pp{l_1: \tsel\pp{l_1: \ty}, l_2: \tend}$ 
from Example~\ref{ex:typegraph}
and $\T_2 = \mu\ty. \tsel\pp{l_1: \ty, l_2: \tend}$.
$\T_1$ is a subtype of $\T_2$; the graph for the type simulation is in Figure~\ref{fig:typesimulation}.
\end{example}
Ghilezan \etal\ give an inductive subtype-checking 
algorithm 
defined
in~\iffull{Figure~\ref{fig:ghilezhan_subtyping_rules}\ }\else{\cite[Table
    6]{Ghilezan2019}\ }\fi (extending the one 
for binary in \cite{Gay2005}). The rules take 
the judgement of the form $\Theta \vdash \T \leq \T'$, which means ``assuming the relations in $\Theta$, we may conclude that $\T$ is a subtype of $\T'$''. 
They defined, but did not analyse its complexity. 
We prove the following result that 
their algorithm takes \emph{exponential 
time}. 

\begin{restatable}{theorem}{complexityghilezansubtyping}
\label{thm:complexity_ghilezan_subtyping}
Let $n=|\T_1|+|\T_2|$.
The upper bound of the worst-case complexity of 
the algorithm to check
$\T_1\subt \T_2$ 
given in \cite[Table 6]{Ghilezan2019} is $\bigO{n^{n^3}}$, 
and the lower bound of the worst-case complexity 
is $\bigOmega{(\sqrt{n})!}$. 
\end{restatable}
Subtyping checking can be performed
in quadratic time using the simulation $\subtsim$.
For fixed types $\T_1$ and $\T_2$, in a type derivation tree, one may need to prove $\Gamma \vdash \T_1 \subt \T_2$
multiple times (possibly with different values of $\Gamma$), whereas using a subtyping
graph all of these judgements are combined into a single node.

To describe the subtyping algorithm, 
we use the following definition, called 
\emph{inconsistent nodes}. 
In $\GG(\T, \T')$, a state $(\T_1', \T_2')$ is \emph{inconsistent} 
if it immediately violates Definition \ref{def:type_simulation}.

\begin{definition}[Inconsistent nodes]
  \label{def:inconsistent_node}
A state $(\T_1', \T_2')$ in 
$\GG(\T, \T')$ 
is \emph{inconsistent} iff at least one of the following hold, for some $\T_1''$ and $\T_2''$:
  \begin{itemize}
    \item $\T_1' \trans{\gell} \T_1''$ and $\T_2' \not\trans{\gell} \T_2''$, for $\gell \in \{\trbra\pp l, \tend, \trin{\pp}{S}, \trout{\pp}{S}\}$.
    \item $\T_2' \trans{\gell} \T_2''$ and $\T_1' \not\trans{\gell} \T_1''$, for $\gell \in \{\trsel\pp l, \tend\}$.
  \end{itemize}
\end{definition}

\begin{restatable}{theorem}{quadraticsubtyping}
  \label{thm:quadratic_subtyping}
  \proofreference{Appendix~\ref{section:quadratic_subtyping_proof}}
There exists an algorithm for checking $\T_1 \subt \T_2$
with worst-case time complexity
$\bigTheta{|\T_1| \cdot |\T_2|}$. 
\end{restatable}
\begin{proof}
{\bf (upper bound)} 
For the algorithm for checking $\T_1 \subt \T_2$,  
we first construct a subtyping graph $\GG(\T_1, \T_2)$.
Then for each node $(\T_1', \T_2')$ in the graph, 
check if it is inconsistent. Next we add 
the edge $(\T_1,\T_2)\rightarrow (\T_1',\T_2')$ if $\T_1 \,\cR\,\T_2$ 
directly implies $\T_1' \,\cR\,\T_2'$ in the definition of $\subtsim$. 
If there is any reachable inconsistent node, 
then $\T_1 \not\subt \T_2$; otherwise, $\T_1 \subt \T_2$.
To show correctness of the algorithm, 
observe that any set of consistent nodes closed
under reachability is a type simulation, 
directly from Definition \ref{def:type_simulation}.
The minimal such set containing $(\T_1, \T_2)$, 
if it exists, must be the set of reachable nodes from $(\T_1, \T_2)$.
Thus we can check if any inconsistent nodes are contained in this set to solve the problem. 
Since the graph has $\bigO{|\T_1| \cdot |\T_2|}$ nodes and edges, 
the algorithm runs in $\bigO{|\T_1| \cdot |\T_2|}$ time with a reachability search.

\noindent{\bf (lower bound)} 
  For $i \in \{1, 2\}$, consider 
  $\T_i = \mu\ty.\tin\pp\tint \dots \tin\pp\tint \ty$ where
  $\size{\T_i}=n_i+2$ (i.e., there are $n_i$-times $\tinn\pp\tint$
  under the
  recursive body).   
  The type graph of $\T_i$ is isomorphic to the graph on $n_i$ vertices $\{0, \dots, n_i - 1\}$, and $j \trans{\trin\pp{\tint}} ((j+1) \bmod n_i)$.
  Thus the subtyping graph has transitions $(j, k) \rightarrow ((j+1) \bmod n_1, (k+1) \bmod n_2)$ with starting node $(0, 0)$.
  If $n_1$ and $n_2$ are coprime then all $n_1 \cdot n_2$ nodes are reachable.
  We have that $\T_1 \subt \T_2$ thus the algorithm must search the entire graph,
  and therefore the complexity in this case is $\bigTheta{|\T_1| \cdot |\T_2|}$.
\end{proof}

A summary of the results in this section is in Table~\ref{tab:subtyping_complexity_summary}.



\section{Complexity Analyses of Four End-Point Projections} 
\label{sec:projections}
This section studies complexity of the four existing 
end-point projections. More precisely, 
for the \emph{inductive plain and full merging}, as well as the \emph{coinductive plain merging}
projections, we give the upper and lower bounds of
the worst-case complexity of the projection algorithms.
For the \emph{coinductive full merging} projection,
we define a new sound and complete 
projection algorithm for coinductive projection with full merging using type graphs (Definition~\ref{def:subsetconstruction}).
We then give the upper and lower bounds of 
the worst-case complexity of 
\emph{the projection problem} (which asks the projection of given $\G$) for 
the coinductive full merging projection.
Figure~\ref{fig:venn} depicts the relationship 
between sets of local types produced by 
the four projections. 
The table in the right hand side (1-4) lists examples 
which represent different projections. 
They are explained throughout this section.  
Liveness, deadlock-freedom and safety properties (5-9) are explained in 
\Sec\ref{section:checking_safety_properties}.


\input{fig_venn.tex}

\subsection{Global Types}
\label{subsec:global}
We define \textit{global types}, which give a global view 
of interactions between all the participants. Global types are 
projected to a set of local types by the algorithms defined in the next subsection. 
%
\begin{definition}[Global type] 
\label{def:global} \emph{Global types} $\G, \G', \G_i, \dots$ 
are defined by the following syntax:\\[1mm]
\centerline{$
\begin{array}{c}
  \G ::= \tend \SEP \GMsg{\pp}{\q}{S}\G \SEP \GvtPair{\pp}{\q}{l_i: \G_i}{i \in \I} \SEP \mu\ty. \G \SEP \ty
\end{array}
$}
\end{definition}
The \emph{message type}, $\GMsg{\pp}{\q}{S}\G$, denotes 
a message from participant $\pp$ to participant $\q$ with a payload of type 
$S$; the \emph{branching type}, $\GvtPair{\pp}{\q}{l_i: \G_i}{i \in \I}$, 
defines that participant $\pp$ can select label $l_i$ at participant $\q$ where $I\not=\emptyset$ and $l_i\not=l_j$ with $i\not = j$ for all $i,j\in I$; 
$\mu\ty. \G$ is a \emph{recursive type} (where we assume $\G$ is guarded);  
and $\ty$ is a \emph{type variable}. 
The set of participants of a global type $\G$ is defined as:
$\pt{\ty} = \pt{\tend} = \emptyset$; 
$\pt{\mu \ty. \G} = \pt{\G}$; 
$\pt{\GMsg{\pp}{\q}{S}\G} = \pt{\G}\cup\{\pp, \q\}$; and 
$\pt{\GvtPair{\pp}{\q}{l_i: \G_i}{i \in \I}} = \left(\bigcup_{i \in \I}{\pt{\G_i}}\right) \cup \{\pp, \q\}$.  
We denote $\ftv{\G}$ as the set of \emph{free type variables} in $\G$ 
and if $\ftv{\G}=\emptyset$, we call $\G$ \emph{closed}. 

We define the size $|\G|$ and subformulas of global types $\Sub(\G)$ analogously to local types. 

\begin{definition}[Size of a global type] 
The size of a global type $\G$, denoted $|\G|$, is defined as:
$|\ty| = |\tend| = 1$; $|\mu\ty. \G| = 1 + |\G|$; $|\GMsg{\pp}{\q}{S}\G| = 1 + |\G|$; and $|\GvtPair{\pp}{\q}{l_i: \G_i}{i \in \I}| = 1 + \sum_{i \in \I} |\G_i|$. 
\end{definition}
\begin{definition}[Subformulas of global types]  
\label{def:subg}
The set of \emph{subformulas of a global type $\G$}, denoted $\Sub(\G)$, is defined as:
$\Sub(\tend) = \tend$; $\Sub(\GMsg{\pp}{\q}{S}\G) = \{\GMsg{\pp}{\q}{S}\G\} \cup \Sub(\G)$; 
$\Sub(\GvtPair{\pp}{\q}{l_i: \G_i}{i \in \I}) = \{\GvtPair{\pp}{\q}{l_i: \G_i}{i \in \I}\} \cup \bigcup_{i \in \I} \Sub(\G_i)$; and 
$\Sub(\mu\ty. \G) = \{\G\} \cup \{\G'[\mu\ty. \G/\ty] \mid \G' \in \Sub(\G)\}$.\end{definition}
 
\begin{definition}[Unfolding of a global type]\label{def:unfold_global}
We define $\unfold{\mu\ty. \G} = \unfold{\G\sub{\mu\ty. \G}{\ty}}$, and $\unfold{\G} = \G$ otherwise.
\end{definition}




\label{section:projection}
\subsection{Two Inductive Projections}
\label{section:inductive_projection}
Firstly, we consider the two algorithms of \textit{inductive projections},
which are defined purely syntactically. 
This comes in two variants: \textit{plain merging} 
from \cite{Honda2008,HYC2016} which merges only 
identical types if the participant is not involved. 
The second one, called \textit{full merging} from \cite{YH2024}, 
gives more projectable global types, 
enabling merging of two different branchings. 

\begin{definition}[Inductive projection] \label{def:inductive_projection}
The projection operator on participant $\pp$ for global type $\G$, 
denoted by $\proj{\G}{\pp}$, is a partial function defined as:
{\small
  \begin{itemize}[leftmargin=*]
    \item $\proj{\left(\Gvt{\pq}{\pr}{S} \G\right)}{\pp} =
      \begin{cases}
        \tout\pr{S} \proj{\G}{\pp} & \text{if } \pp = \pq\\
        \tin\pq{S} \proj{\G}{\pp} & \text{if } \pp = \pr\\
        \proj{\G}{\pp} &\hspace{-2mm}\text{otherwise}\\
      \end{cases}$
      \quad
      $\proj{\left(\GvtPair{\pq}{\pr}{l_i: \G_i}{i \in \I}\right)}{\pp} =
      \begin{cases}
        \tselsub\pr{l_i: \proj{\G_i}{\pp}}{i \in \I} & \text{if } \pp = \pq\\
        \tbrasub\pq{l_i: \proj{\G_i}{\pp}}{i \in \I} & \text{if } \pp = \pr\\
        \bigmerge_{i \in \I} \proj{\G_i}{\pp} & \text{otherwise}\\
      \end{cases}
    $
    \item 
$\proj{\left(\mu \ty. \G\right)}{\pp} =
      \begin{cases}
        \tend & \text{if }\pp \not\in \pt{\G}\text{ and }
\ftv{\mu\ty. \G}=\emptyset\\
        \mu \ty. \left(\proj{\G}{\pp}\right) & \text{otherwise}\\
      \end{cases}
    $
\quad $\proj{\tend}{\pp}=\tend$ 
  \end{itemize}
}
\noindent where the merging operator $\mergep$ is defined in Definitions~\ref{def:inductive_plain_merging} and \ref{def:inductive_full_merging} below.  
If undefined by the above rules, then $\proj{\G}{\pp}$ is undefined. 
\end{definition}

\begin{definition}[Inductive plain merging]\label{def:inductive_plain_merging}
The plain merge is defined by $\T \mergep \T = \T$, and undefined otherwise.
\end{definition}
\begin{definition}[Inductive full merging]\label{def:inductive_full_merging}
The full merge is defined inductively by:\\[1mm]
{\small
$\tdag\pp{S}{\T_1} \mergef \tdag\pp{S}{\T_2} = \tdag\pp{S}(\T_1 \mergef \T_2)$ with $\dagger\in \{!,?\}$;
$\ty \mergef \ty = \ty$;
$\tend \mergef \tend = \tend$;\\[1mm]
$\tselsub\pp{l_i: \T_i}{i \in \I} \mergef \tselsub\pp{l_i: \T_i'}{i \in \I} = \tselsub\pp{l_i: \T_i \mergef \T_i'}{i \in \I}$; and\\[1mm]
$\tbrasub\pp{l_i: \T_i}{i \in \I} \mergef \tbrasub\pp{l_j: \T_j'}{j \in \J}
= \tbraset\pp\left(\{l_k: (\T_k \mergef \T_k') \mid k \in \I \cap \J\} \cup \{l_i: \T_i \mid i \in \I \minus \J\} \cup \{l_j: \T_j' \mid j \in \J \minus \I\}\right)$;
}\\
\noindent and undefined otherwise.
\end{definition}
\noindent Note that we can show that both notions of merging are associative and commutative \cite{Ghilezan2019}.
We introduce the notation: $\bigmergep_{i \in \I}\T_i = (\dots(\T_{n_1} \mergep \T_{n_2})\dots) \mergep \T_{n_m}$ for finite non-empty set $\I = \{n_1, \dots, n_m\}$, and similarly for $\bigmergef$.


Notice that in full merging (Definition~\ref{def:inductive_plain_merging}), while the internal choice does not increase 
the branchings, the two external choices can be merged to make 
a larger type with more branchings.  
This is consistent with subtyping: given some types
with external choices, merging the choices together keeps them as
subtypes of the original types, and such can be used in place of
them. For internal choice, adding more choices makes
an ill-behaved global type well-formed. See \cite[Example
  3.13]{YH2024} and \cite[\Sec 4.3]{HYC2016}.

The following example illustrates that full merging is more expressive than plain merging.

\begin{example}\label{ex:full_merging_is_more_expressive} 
Consider two global types where $\pr$ is sending 
a label to $\pp$ after receiving from $\pq$:\\[1mm]
\centerline{
$\begin{array}{rl}
\G_{\text{ip}} \ = & \mu\ty.\:\GvtPairs\pq\pr{l_1: \GvtPairs\pr\pp{l_1: \ty}, l_2: \GvtPairs\pr\pp{{\color{violet}{l_1: \ty}}}}\\
\G_{\text{if}} \ = & \mu\ty.\:\GvtPairs\pq\pr{l_1: \GvtPairs\pr\pp{l_1: \ty}, l_2: \GvtPairs\pr\pp{{\color{violet}{l_2: \tend}}}}
\end{array}$
}\\[1mm]
Then we have 
$\proj{\G_{\text{ip}}}\pp \ = \ 
\mu\ty.\:\tbra\pr{l_1: \ty}$ by both plain and full merging,  
but:\\[1mm]
\centerline{
$\begin{array}{ll}
\proj{\G_{\text{if}}}\pp
    \ = & 
       \mu\ty.\:(\proj{\GvtPairs\pr\pp{l_1: \ty}}\pp \mergep\: 
       \proj{\GvtPairs\pr\pp{l_2: \tend}}\pp)
     = \mu\ty.\:(\tbra\pr{l_1: \ty} \mergep \tbra\pr{l_2: \tend})
\end{array}$}
\\[1mm]
By plain merging, this is \emph{undefined}; by full merging, it is \emph{defined} as: $\mu\ty.\tbra\pr{l_1: \ty, l_2: \tend}$. 
\end{example}
\subsubsection{Complexity of the projection with plain merging}
\label{subsec:plainmerge}
We show that inductive projection of a global type $\G$ with plain merging is in $\bigO{|\G| \log |\G|}$ time. 
Most cases of inductive projection can be performed in constant time and reduce the total sizes of the types to be projected by a constant amount; the only exception is the projection of $\proj{\left(\GvtPair{\pq}{\pr}{l_i: \G_i}{i \in \I}\right)}{\pp}$ where $|\I| \geq 2$ and $\pp \notin \{\q, \pr\}$.
In this case the syntactic equality check takes linear time and the size of each projected branch is at most half of the size of the original type. Theorem~\ref{thm:inductive_plain_merge} states this argument.

\begin{restatable}{theorem}{thmplain}\label{thm:inductive_plain_merge}
\label{thm:plain:lower}  
\label{thm:plain}
\proofreference{Appendix~\ref{app:plainprojection}}
There exists an algorithm for inductive projection with plain merging
with worst-case time complexity $\bigTheta{\size{\G}\log\size{\G}}$. 
\end{restatable}
\begin{proof}
\textbf{(upper bound)}  
First, checking 
$\ftv{\G}=\emptyset$
and $\pp \notin \pt{\G}$
for every syntax of subformulas of $\G$ 
is performed in $\bigO{n}$ time by 
$\size{\Sub(\G)}= \bigO{|\G|}$. 
We also note that 
$\size{\proj{\G}{\pp}}\leq\size{\G}$.
Then, our algorithm performs the projection in the same recursive sense as in Definition~\ref{def:inductive_projection}, except for the aforementioned checking of $\ftv{\G}=\emptyset$. 
First, we prove that the algorithm 
computes $\proj\G\pp$ in $2|\G|\log|\G|$ time, if $\proj\G\pp$ is defined.
Then, we prove that the time taken when $\proj\G\pp$ is undefined is
$2|\G|\log|\G| + |\G|$. 

\noindent\textbf{(lower bound)}  
Consider: 
$\G_m = \Gvt\pp\pr\tint \GvtPairs\pp\pq{l_1: \G_{m-1}, l_2: \G_{m-1}}$
$(m > 0)$
and $\G_0 = \tend$.  
We compute $\proj{\G_m}\pp$. Type $\G_m$ has size $\bigO{2^m}$ but
computing $\proj{\G_k}\pp \mergep \proj{\G_k}\pp$ requires $\bigTheta{2^k}$ time (because $|\proj{\G_k}\pp| = \bigTheta{2^k}$). Thus, the
time taken to compute $\proj{\G_m}\pp$ is $\bigTheta{m \cdot 2^m}$, which
gives rise to the $\bigOmega{n \log n}$ lower bound.
\end{proof}


\subsubsection{Complexity of projection with full merging: Na\"ive algorithm}
\label{subsec:fullmerge}
In the full merging, as the requirement for merged branches to be
syntactically equal is dropped, unlike the plain merging case, the
branch sizes may be unbalanced. 
If we represent type graphs in a \emph{na\"ive} way (that is, syntactically representing types as finite trees),
the following example introduces a proof that the worst-case runtime of such an algorithm is $\bigOmega{n^2}$.\\[1mm] \centerline{
$\begin{array}{rcll}
  \G^{(n+1)} &= & \GBra \pq \pr {l_1: \G^{(n)}, l_2: \GBra \pq \pp {l_{n+1}: \tend}} & \text{for } n \geq 0\\
\G^{(0)} &= & \GBra \pq \pp {l_0: \tend}
\end{array}
$}



\begin{theorem}\label{thm:naive_full_merging_algorithm_is_quadratic} 
Finding $\proj {\G^{(n)}} \pp$ under full merging with a na\"ive representation of types as syntax trees takes $\bigTheta{n^2}$ time.
\end{theorem}
\begin{proof}
We have $\proj{\G^{(n)}}\pp = \tbrasub{\pq}{l_{i}: \tend}{1 \leq i \leq n}$.
 Therefore, when computing $\proj{\G^{(n+1)}}\pp$ from $\proj{\G^{(n)}}\pp$, the na\"ive algorithm will need to compute $\tbrasub{\pq}{l_{i}: \tend}{1 \leq i \leq n} \:\mergef\: \tbra {\pq}{l_{n+1}: \tend}$. As the size of the types is linear, and merging takes linear time in the na\"ive algorithm, the $n$-th step uses $\bigTheta{n}$ time, so the total time is $\bigTheta{n^2}$. 
\end{proof} 

\subsubsection{Complexity of projection with full merging: Optimised algorithm}
We construct the optimised algorithm 
using the classical 
technique of small-to-large merging: by representing each branch of a \emph{branching type} 
as a leaf of a binary search tree on their labels, we can efficiently represent
the set of branches of a type, and 
merge branching types adding each branch of the smaller type to the
bigger type. Intuitively, each branch is ``merged'' into a
larger type only $\bigO{\log n}$ times, but a careful analysis is
needed as merging affects the size of types and is recursive.
Note that our method works for non-binary branching types, and the
nontriviality in the above construction is handling \emph{arbitrary types}
with \emph{arbitrarily many branches} in each choice.

\begin{restatable}{lemma}{complexityfullmerge}
\proofreference{Appendix~\ref{app:fullprojection}}
\label{thm:complexity_of_full_merge} 
\quad If $\:\T_1 \mergef \T_2 = \T$, then $\T_1 \mergef \T_2$ can be computed in\\$\bigO{|\T_1| + |\T_2| - |\T| + |\T_2| \cdot \log |\T_1|}$ time.
\end{restatable}
\begin{proof}
By structural induction on $\T_1$. 
First note that 
$\size{\proj{\G}{\pp}}\leq\size{\G}$ under the full merge. 
We show that, up to constant
factors, the merge can be computed in time at most $2(|\T_1| + |\T_2|
- |\T|) - 1 + (|\T_2| - 1) \cdot \log |\T_1|$.  
Note this is positive by
the fact that $\T_1 \mergef \T_2 = \T$ implies $|\T| < |\T_1| + |\T_2|$\iffull{\ (Lemma~\ref{thm:merge_is_smaller})}\else{}\fi.    
To allow for the complexity logarithmic in $|\T_1|$, 
we represent branching types as binary search trees. 
\end{proof}


\begin{restatable}{theorem}{thmfullmergeopt}\label{thm:inductive_full_merge}
  \label{thm:fullmergeopt}\label{thm:fullmergeopt:lower}
  \proofreference{Appendix~\ref{app:fullprojection}} 
There exists an algorithm for inductive projection with full merging
with worst-case time complexity
$\bigTheta{\size{\G}\log^2\size{\G}}$. 
\end{restatable}
\begin{proof}
{\bf (upper bound)}   
We start by checking every syntactic subterm $\G$ satisfies $\ftv{\G}=\emptyset$
and $\pp \notin \pt{\G}$ in linear time.
Then, we show inductively that (up to constant factors) we can compute $\proj\G\pp$ in time $|\G| \log^2 |\G| + 2|\G| - |\T|$, where $\T = \proj\G\pp$. 

\noindent{\bf (lower bound)}
Consider the following global types: 
$\G_k^{(j)} = \GvtPairs\pp\pq{l_1: \G_{k-1}^{(2j)}, l_2:
  \G_{k-1}^{(2j+1)}}$  with $k > 0$ and 
    $\G_0^{(j)} = \GvtPairs\pp\pr{l_j: \tend}$. 
  We shall compute $\proj{\G_m^{(0)}}\pr = \tbrasub\pr{l_i: \tend}{i \in \{0, \dots, 2^m - 1\}}$.
  We have $\proj{\G_k^{(j)}}\pr = \tbrasub\pr{l_i: \tend}{i \in \{2^k \cdot j, \dots, 2^k \cdot (j+1) - 1\}}$ and $|\proj{\G_k^{(j)}}\pr| = 2^j$.
Computing 
  $\proj{\G_{k-1}^{(2j)}}\pr \mergef \proj{\G_{k-1}^{(2j+1)}}\pr$
  $=\tbrasub\pr{l_i: \tend}{i \in \{2^{k-1} \cdot (2j), \dots, 2^{k-1} \cdot (2j+1) - 1\}} \mergef \tbrasub\pr{l_i: \tend}{i \in \{2^{k-1} \cdot (2j+1), \dots, 2^{k-1} \cdot (2j+2) - 1\}}$
  requires merging $2^{k-1}$ branches into an existing branching type with $2^{k-1}$ branches,
thus this requires $2^{k-1}$ operations on a balanced binary search tree with $2^{k-1}$ leaves.
Hence the merging requires time $\bigTheta{k \cdot 2^k}$.

We can deduce that the total time required to compute $\proj{\G_m^{(0)}}\pr$ is $\bigTheta{m \cdot 2^m}$ and $|\G_m^{(0)}| = \bigTheta{2^m}$, which gives rise to the $\bigOmega{n \log^2 n}$ lower bound.
\end{proof}


\subsection{Tirore \etal's Projection \cite{Tirore2023}}
\label{section:coinductive_projection}
In the next two subsections, we discuss complexity of 
the \emph{coinductive projection} with the plain and full merging
operators.
We define the \emph{coinductive projection} using the unfolding 
function following the style in Definition~\ref{def:subtyping_local_types}. 

\begin{definition}[Coinductive projection]\label{def:coinductiveproj}
The \emph{coinductive projection} is the largest relation 
${}\fproj{\pr}{}$ between global and local types such that, when 
${\G}\fproj{\pr}{\T}$:
  
\begin{itemize}
\item $\pr \notin \pt{\G}$ implies $\unfold{\T} = \tend$
\item $\unfold{\G} = \Gvt{\pp}{\pr}{\S}\G'$ implies 
$\unfold{\T} = \tin{\pp}{\S} \T'$ and ${\G'}\fproj{\pr}{\T'}$.
\item $\unfold{\G} = \Gvt{\pr}{\pp}{\S}\G'$ implies $\unfold{\T} = \tout{\pp}{\S} \T'$ and ${\G'}\fproj{\pr}{\T'}$.
\item $\unfold{\G} = \Gvt{\pp}{\pq}{\S}\G'$ and $\pr \notin \{\pp, \pq\}$ implies $\G' \fproj{\pr} {\T}$.
\item $\unfold{\G} = \GvtPair{\pp}{\pr}{l_i: \G_i}{i \in \I}$ implies $\unfold{\T} = \tbra{\pp}{l_i: \T_i}$ and $\G_i \fproj{\pr} \T_i\ \forall i\in\I$.
\item $\unfold{\G} = \GvtPair{\pr}{\pp}{l_i: \G_i}{i \in \I}$ implies $\unfold{\T} = \tsel{\pp}{l_i: \T_i}$ and $\G_i \fproj{\pr} \T_i\ \forall i\in\I$.
\item $\unfold{\G} = \GvtPair{\pp}{\pq}{l_i: \G_i}{i \in \I}$ implies $\exists \T_i.\ \T = \bigmergec_{i \in \I} \T_i$ and $\G_i \fproj{\pr} \T_i\ \forall i\in\I$
\end{itemize}

    for some \emph{merging operator} $\mergec$.
\end{definition}

\begin{definition}[Coinductive plain merging]\label{def:coinductive_plain_merge}
Recall $\subteq$ denotes the type graph equivalence
(Definition~\ref{def:type_simulation}).
The \emph{coinductive plain merge} is defined by $\T \mergec \T' = \T$ 
if $\T\subteq \T'$ and undefined otherwise. 
\end{definition}

We first recall \emph{balanced global types} from 
Definition~3.3 in \cite{Ghilezan2019}.
\begin{definition}[Global type graph] 
A \emph{global type graph} of $\G$ (denoted by $\GG(\G)$) has 
nodes $\Sub(\G)$ and has edges defined by the following rules:\\[1mm]
\centerline{\small
    \begin{prooftree}
      \hypo{\unfold{\G} = \GMsg \pp \pq \S \G'}
      \infer1{\G \rightarrow \G'}
    \end{prooftree}\quad \quad \quad 
    \begin{prooftree}
      \hypo{\unfold{\G} = \GvtPair \pp \pq {l_i: \G_i}{i \in \I}}
      \infer1{\G \rightarrow \G_i}
    \end{prooftree}
}
\end{definition}
\begin{definition}A global type $\G$ \emph{involves} participant $\pp$ 
if $\unfold{\G}=\Gvt{\pr}{\ps}{S}\G'$ 
or $\unfold{\G}=\GvtPair{\pr}{\ps}{l_i: \G_i}{i \in \I}$
with $\pp\in \set{\pr,\ps}$.  We write $\involve{\G}{\pp}$
if $\G$ involves $\pp$. 
\end{definition}
\begin{definition}[Balanced global types, Definition~3.3 in \cite{Ghilezan2019}]
A global type $\G$ is \emph{balanced} iff, for each node $\G'$ in 
a global type graph $\GG(\G)$,
  whenever a node involving some $\pp$ is reachable from
 $\G'$, there exists $k$ such that all paths from $\G'$ reach a node involving $\pp$ within $k$ steps.
\end{definition}
\noindent The balanced condition is required for
ensuring liveness of session $\M$ (Definition~\ref{def:properties}(3))
under the coinductive projection. 
For example, 
the local types projected from 
the global type 
$\G'=\mu\ty.\:\GvtPairs\pp\pq{l:\ty, \ l':\GvtPairs{\pq}{\pr}{l'':\tend}}$ 
in \cite[Example~3.12]{Ghilezan2019}
can type \emph{non-live} session $\N$. 

\begin{lemma}[Balanced global types] \label{lem:balanced} 
Suppose $\proj{\G}{\pp}$ (Definition~\ref{def:inductive_projection}) is defined for 
$\forall\pp. \ \pp\in \pt{\G}$.
Then $\G$ is balanced.
\end{lemma}

\begin{restatable}{theorem}{thmbalancedcomplexity}
\label{thm:balanced_complexity}
\proofreference{Appendix~\ref{app:coinductivefull}}
Checking $\G$ balanced can be performed in $\bigO{|\G|^2}$ time.
\end{restatable}

\noindent We now give the complexity analysis for the most recent projection 
algorithm proposed by Tirore \etal~\cite{Tirore2023}, which we call 
\TBC.  
This algorithm computes a projection on inductive types with plain merge 
that is sound and complete with respect to its coinductive projection; 
Example~\ref{ex:coinductive_projection_is_more_expressive} shows that
syntactic inductive projections may fail to compute a projection when 
\TBC\ may not.
This is because inductive projection considers the syntax of types,
so types with the same infinite tree but different patterns of $\mu$-binders 
are not considered equal in merging. 
Their algorithm is equivalent to considering types 
as infinite trees.

\begin{example}\label{ex:coinductive_projection_is_more_expressive} Consider the following global type:\\[1mm]
\centerline{
   $\G_{\text{cp}} = \GvtPairs\pq\pr{l_1: \mu\ty. \GMsg\pq\pp\S \ty, l_2: \mu\ty'. \GMsg\pq\pp\S \GMsg\pq\pp\S \ty'}$ 
}\\[1mm]
\noindent Both branches of the communication~$\pq \rightarrow \pr$ represent the same tree, but they are not equal syntactically. Therefore, inductive projection (with either kind of merging) will fail to compute a projection. The \TBC\ projection is defined as these types represent the same infinite trees. On the other hand, ${\G_{\text{if}}}$ in 
Example~\ref{ex:full_merging_is_more_expressive} is not projectable 
by \TBC. We shall see ${\G_{\text{if}}}$ is projectable 
by coinductive projection 
with the full merging defined in Definition~\ref{def:coinductive_full_merge}. 
\end{example}

The result of the following lemma is visualised in 
Figure~\ref{fig:venn}.

\begin{restatable}{proposition}{propifcp}
  \label{pro:ifcp}
  \proofreference{Appendix~\ref{app:coinductiveprojection}} 
If $\G$ is inductively projectable by full merging but 
not inductively projectable by plain merging,  
then $\G$ is not coinductively projectable by 
plain merging. 
\end{restatable}

The \TBC\ algorithm is non-trivial and consists of a number of functions
over graph representations of $\G$ and $\T$,  
taking several pages of their paper. Rather than listing 
all of the definitions needed for \TBC, we summarise how we analyse 
its complexity. The \TBC\ algorithm reduces the problem to one of checking whether a
global type $\G$ is projectable to a local type $\T$: first, a
candidate projection $\ptrans(\G)$ is defined, then the
\textit{intermediate projection}
is checked. It is then shown that this transformation 
is sound and complete with respect to a coinductive definition of projection with plain merging.   
Specifically, if $\G$ is projectable, it is projectable to a coinductive 
type that unravels to $\ptrans(\G)$ \cite[Definition~9]{Tirore2023}. 
In short, two predicates are checked: one for all nodes in a graph representation of $\G$ and another for all nodes in a \emph{product graph} of $\G$ and $\T$, which has size $|\G| \cdot |\T|$.
As $|\T| \leq |\G|$, 
we arrive at the following theorem.

\begin{restatable}{theorem}{theconductiveprojectionupperlower}
\proofreference{Appendix~\ref{app:coinductiveprojection}} 
\label{thm:coinductiveprojection}
The projection in~\cite{Tirore2023} takes $\bigTheta{|\G|^2}$ time in the worst case.
\end{restatable}

\subsection{Coinductive Projection with Full Merging}
\label{subsec:coinductivefull}
We define the coinductive projection with 
full merging extending from \cite[Definition 3.6]{Ghilezan2019}. 

\begin{definition}[Full coinductive merge]\label{def:coinductive_full_merge}
The \emph{full coinductive merge} $\mergecf$ is coinductively  
defined by:

\begin{itemize}[leftmargin=*]
\item If $\T_1 = \tend$ and $\T_2 = \tend$, then $\T_1 \mergecf \T_2 = \tend$.
\item If $\unfold{\T_1} = \tdag\pp{\S}{\T_1'}$ and $\unfold{\T_2} = \tdag\pp{\S}{\T_2'}$,
then $\T_1 \mergecf \T_2 = \tdag\pp{\S}{\T_1' \mergecf \T_2'}$ with $\dagger\in \set{!,?}$.
\item If $\unfold{\T_1} = \tselsub\pp{l_i: \T_i}{i \in \I}$ and $\unfold{\T_2} = \tselsub\pp{l_i: \T_i'}{i \in \I}$,
then $\T_1 \mergecf \T_2 = \tselsub\pp{l_i: \T_i \mergecf \T_i'}{i \in \I}$.
\item If $\unfold{\T_1} = \tbrasub\pp{l_i: \T_i}{i \in \I}$ and $\unfold{\T_2} = \tbrasub\pp{l_j: \T_j'}{j \in \J}$,\\
then $\T_1 \mergecf \T_2 = \tbrasub\pp{l_k: \T_k''}{k \in \I \cup \J}$ where 
(1) $\T_k'' =\T_k \mergecf \T_k'$ if 
$k \in \I \cap \J$; 
(2) $\T_k'' = \T_k$ if  $k \in \I \minus \J$; 
or (3) $\T_k'' = \T_k'$ if $k \in \J \minus \I$; undefined otherwise. 
\end{itemize}
\end{definition}

\subsubsection{Lower bounds for coinductive projection with full merging}
We exhibit a global type $\G$ such that the size of the projected
local type graph must be exponential in the syntactic size of
$\G$. 

\begin{theorem}\label{thm:coinductive:size}
There exists a (balanced) global type $\G$ whose projection onto some participant requires a local type graph of size $\bigOmega{2^{|\G|^c}}$ for some $c > 0$.
\end{theorem}
\begin{proof}
We prove this for $c = \frac{1}{2} - \varepsilon$, $\varepsilon > 0$.
Fix integers $n_i$ for $i$ in some finite indexing set
$\I$. Consider the following global type:\\[1mm]
{$\begin{array}{rl}
    \G_{\text{cf}} &= \mathrlap{\GvtPair{\pp}{\pr}{l_i: \mu\ty. \GvtBigPairs{\pp}{\pq}{\text{\color{red}$a$}: \GvtPairs{\pp}{\pq}{\text{\color{red}$a$}: \dots \GvtPairs{\pp}{\pq}{\text{\color{red}$a$}: \ty} \dots },\
\text{\color{red}$b$}: \GvtPairs{\pp}{\pq}{l_i: \tend}}}{i \in \I}}
{\hspace{5.9em}{\underbrace{\hspace{14.6em}}_{n_i{\text{ times}}}}}
  \end{array}
$}\\[1mm]
\noindent Participant $\pp$ chooses one $i \in \I$ to follow: 
after every multiple of $n_i$ communications with label {\color{red}$a$} to $\pq$,
it is possible to send a label {\color{red}$b$} to $\pq$.
Thus, $\pq$ must be able to accept {\color{red}$b$} precisely after $m$ {\color{red}$a$} messages
if and only if $m$ is a multiple of some $n_i$.
For coprime $n_i$, the sequence of allowable communications is acyclic, so we would require at least $\prod_{i \in \I} n_i$ states of a local type graph.
We have that $|\G| = 1 + 2 \cdot |\I| + \sum_{i \in \I} n_i$.
Taking $n_i$ as the first $k$ primes, as a loose bound we have $|\G| = \bigO{k^{2+\varepsilon'}}$ (for all $\varepsilon' > 0$)
by the prime number theorem,
and $\prod_{i \in \I} n_i = \bigOmega{2^k}$.
Thus, the required number of states is $\bigOmega{2^m}$ with
$m={|\G|^{\left(\frac{1}{2} - \varepsilon\right)}}$. 
\end{proof}
Type $\G_{\text{cf}}$ also serves as the example which is projecatble 
only by the coinductive projection with full merging, but not any of other 
three projections. 

\begin{corollary}\label{cor:coinductive:lower}
The coinductive projection with full merging has worst-case time complexity lower bounded by $\bigOmega{2^{|\G|^{\left(\frac{1}{2} - \varepsilon\right)}}}$ for some $\varepsilon > 0$.
\end{corollary}
\begin{proof}
  The time complexity of the algorithm is lower-bounded by the size of the output local type graph.
\end{proof}

\subsubsection{A \PSPACE\ algorithm for coinductive projection with full merging}
We give a sound and complete algorithm for coinductive projection with full merging. We first introduce a \emph{subset construction} procedure, 
inspired by \cite{Li2023}. We first make a closure 
of $\G$ on $\pp$ as a finite set of subglobal types of $\G$ 
which are reachable without involving $\pp$.

\begin{definition}[Closure of a participant]
Let $\cG$ denote 
a set of closed global types. 
  For participant $\pp$, define the \emph{$\pp$-closure}
  $\gcl{\pp}(\cG)$ of a set of global types $\cG$ as follows:
  (1) $\G \in \cG$ implies $\G \in \gcl{\pp}(\cG)$; 
(2) $\G \in \gcl{\pp}(\cG)$, 
    $\unfold{\G} = \Gvt{\pq}{\pr}{\S'} \G'$ and 
    $\pp \notin \{\pq, \pr\}$ imply 
    $\G' \in \gcl{\pp}(\cG)$; and  
  (3) 
  $\G \in \gcl{\pp}(\cG)$,  
    $\unfold{\G} = \GvtPair{\pq}{\pr}{l_i: \G_i}{i \in \I}$, $j \in \I$ 
and $\pp \notin \{\pq, \pr\}$ imply  
    $\G_j \in \gcl{\pp}(\cG)$. 
\end{definition}

Starting from a given global type $\G$, 
the following algorithm produces a local type graph of participant $\pp$ whose 
nodes are a set of global types and edges are 
actions from a node to a node. The produced 
type graph represents a projection of $\G$ onto $\pp$. 
\begin{definition}[Subset construction]
\label{def:subsetconstruction}
Let $\G$ be a closed global type and $\pp\in \pt{\G}$. 
We compute the \emph{projection graph} $\gproj{\pp}(\G)$ of $\G$ onto $\pp$ as the local type graph whose node is a set of subglobal types of $\G$. 
The initial node is set to $\gcl{\pp}(\{\G\})$.\\[1mm]
\begin{tabular}{rl}
\small  
\RULE{G-End} & $\forall \G\in \cG.\ \neg(\involve{\G}{\pp})$ implies   
$\cG \trans{\trend} \Skip$.\\
\RULE{G-Out} &
$\pp\in \pt{\cG}$, 
$\G_i'\in \cG$ and 
$\involve{\G_i'}{\pp}$ and 
$\G_i'=\Gvt{\pp}{\pq}{\S}\G_i$ imply  
$\cG \trans{\trout{\pq}{\S}} \gcl{\pp}(\{G_i\}_{i \in \I})$.\\


\RULE{G-In} & 
$\pp\in \pt{\cG}$, 
$\G_i'\in \cG$ and  
$\involve{\G_i'}{\pp}$ and 
$\G_i'=\Gvt{\pq}{\pp}{\S}\G_i$ imply  
$\cG \trans{\trin{\pq}{\S}} \gcl{\pp}(\{G_i\}_{i \in \I})$.\\

\RULE{G-Sel} & 
$\pp\in \pt{\cG}$, 
$\G_i'\in \cG$ and  
$\involve{\G_i'}{\pp}$ and 
$\G_i'=\GvtPair{\pp}{\pq}{l_j: \G_{ij}}{j \in \J}$ imply\\
&  
$\cG \trans{\trsel{\pq}{l_j}} \gcl{\pp}(\{\G_{ij}\}_{j \in \J, i \in \I})$ for all $j \in \J$.\\

\RULE{G-Bra} & 
$\pp\in \pt{\cG}$ and  
$\G_i'\in \cG$ and  
$\involve{\G_i'}{\pp}$ and 
$\G_i'=\GvtPair{\pq}{\pp}{l_j: \G_{ij}}{j \in \J_i}$ imply\\
& $\cG \trans{\trbra{\pq}{l_j}} \gcl{\pp}(\{\G_{ij}\}_{j \in \J_i, i \in \I})$ for all $j \in \bigcup_{i \in \I} \J_i$.
\end{tabular}\\[1mm]
The projection graph $\gproj{\pp}(\G)$ is the graph reachable from the initial node. If it is not a valid local type graph, 
then the projection is undefined.

\end{definition}
Note that all nodes are of the form 
$\gcl{\pp}{(\cG)}$ for some set of global types $\cG$.
Intuitively, a node $\cG$ in the subset construction $\gproj{\pp}(\G)$ represents all possible states in which
the global protocol could be in, from the perspective of $\pp$.
Rule \RULE{G-End} allows $\pp$ to terminate given that it is not possible for $\pp$ to be involved in any further communications.
Rules \RULE{G-Out}, \RULE{G-In} and \RULE{G-Sel} allow $\pp$ to communicate if all possible communications are identical.
Rule \RULE{G-Bra} additionally accommodates for the full merge by
becoming looser with labels.
\iffull{\ Example~\ref{ex:subsetcontruction} shows how the subset construction
is built.}\else{\ See \cite{UY2024} for an example of the subset construction.}\fi

\begin{restatable}[Soundness of the subset construction]{theorem}{thmcoprojsound}
  \label{thm:proj:sound}
  \proofreference{Appendix~\ref{app:coinductivefull}} 
Let $\G$ be a balanced global type. 
If the projection graph $\gproj{\pp}(\G)$ is defined, 
then $\G \fproj{\pp} \gproj{\pp}(\G)$.
\end{restatable}
\begin{proof}
By proving $\cR = \{(\G', \GG(\gcl{\pp}(\{\G'\})))\ | \ \G' \in
\Sub(\G)\}$ is a subset of relation $\fproj{\pp}$. 
\end{proof}


Completeness is proved by contradiction, 
assuming that $\G \fproj{\pp} \T$ but some local type $\cG \in \gproj{\pp}(\G)$ is invalid (i.e. its outgoing edges do not match any of the rules in the subset construction).

\begin{restatable}[Completeness of the subset
    construction]{theorem}{thmprojcomplete} \proofreference{Appendix~\ref{app:coinductivefull}}
  Let $\G$ be a balanced global type such that $\G \fproj{\pp} \T$. Then the subset construction $\gproj{\pp}(\G)$ is defined.
\end{restatable}

\begin{theorem} The coinductive full projectability problem for global types is in $\PSPACE$. 
\end{theorem}
\begin{proof}
In polynomial space, we can check all reachable subsets from the initial node of the subset construction, and check that the resulting transitions form a correct type graph.
Then the result follows from the soundness and completeness of the subset construction.
\end{proof}

\begin{theorem} Finding a coinductive projection of a balanced global type $|\G|$ onto a participant $\pp$ with full merging takes time $\bigO{\size{\G} \cdot 2^{\size{\G}}}$.
\end{theorem}
\begin{proof}
We can use the subset construction: there are at most $2^{|\G|}$ nodes in the projection graph, and each node can be generated in $\bigO{|\G|}$ time.\
\end{proof}

\section{Minimum Type Inference and its Complexity Analysis}
\label{section:minimum_typing}
This section first gives a quick summary of 
the top-down typing system from \cite{Ghilezan2019} 
and develops a \emph{minimum type inference system}, 
and gives its complexity analysis.
The term ``minimum'' means the smallest type $P$ can have with respect 
to the subtyping relation (following the same term in \cite{AbadiCardelli96}).

\subsection{Top Down Type Checking}
\label{subsec:topdown}
For the top-down procedure (Figure \ref{fig:top_down}), 
a typing system uses a projection of global types to type a multiparty session $\M$. Let us define a \emph{typing environment} 
($\Gamma, \Gamma',\dots$) as $\Gamma::=\emptyset \SEP \Gamma, x{:}\S \SEP 
\Gamma, X{:}\T$, and a \emph{(session) typing context}
$\Delta::= \pp{:}\T \SEP \Delta, \Delta'$ (notice that there should be at least one participant in $\M$). We write $\Delta(\pp)$ 
for $\T$ with $\pp{:}\T\in \Delta$ and denote $\dom{\Delta}$. 
The typing judgments of expressions and processes are defined as 
$\Gamma \vdash \e : \S$ and $\Gamma \vdash \PP : \T$, respectively,  
as defined in \cite{Ghilezan2019}. 

Following \cite{BHYZ2023,YH2024},  
to type a multiparty session $\M$, we associate a typing context $\Delta$ 
to a global type $\G$ via projection and subtyping.

\begin{definition}[Association, Definition 10 in \cite{BHYZ2023} and 
Definition 9 in \cite{YH2024}]\label{def:association}
A \emph{typing context $\Delta$ is associated with a global type $\G$}, 
written $\Delta\associated\G$ if $\dom{\Delta}=\pt{\G}$ and 
(1) $\forall\pp\in\pt{\G}$, \ $\Delta(\pp)\subt \proj{\G}{\pp}$ where 
$\proj{\G}{\pp}$ is defined by the inductive projection 
(Def.~\ref{def:inductive_projection}) with 
the plain merge (Def.~\ref{def:inductive_plain_merging}) or  
the full merge (Def.~\ref{def:inductive_full_merging}); or (2)  
$\forall\pp\in\pt{\G}$, $\Delta(\pp) \subt \T_{\pp}$ 
such that $\cprojt{\G}{\pp}{\T_\pp}$ where 
$\cprojt{\G}{\pp}{\T_\pp}$ is defined by the coinductive projection 
(Def.~\ref{def:coinductiveproj}) with 
the plain merge (Def.~\ref{def:coinductive_plain_merge}) or  
the coinductive projection with the full merge 
(Def.~\ref{def:coinductive_full_merge}) with $\G$ balanced.
\end{definition}
Notice that the association enlarges a typing context beyond
directly projected local types. 
Recall $\Delta_5$ in the table in
Figure~\ref{fig:venn} and 
$\G=\GvtPairs{\pp}{\pq}{l_1:\GvtPairs{\pr}{\pq}{l_2:\tend},
  l_4:\GvtPairs{\pr}{\pq}{l_2:\tend}}$. Then
$\Delta_5\associated \G$. 
Figure~\ref{fig:processtyping} lists the rules for processes.  
The rules for expressions are standard and found in 
\iffull{\ Appendix~\ref{app:typing}}\else{\ \cite{UY2024}}\fi.
The top-down typing of $\M$ is defined as 
$\provestop \M :: \Delta$ by \RULE{T-Sess} where 
a set of participants are typed by a typing context $\Delta$ 
associated by a global type \cite{YH2024}. 

\begin{figure}
{\small
\fbox{Processes $\Gamma \vdash \PP : \T$}\raggedright
\begin{center}
\begin{spacing}{2}
  \begin{prooftree}
    \hypo {\Gamma \vdash \PP : \T}
    \hypo {\T \subt \T'}
    \infer2[\RULE{T-Sub}] {\Gamma \vdash \PP : \T'}
  \end{prooftree}
\quad
  \begin{prooftree}
    \infer0[\RULE{T-Inact}] {\Gamma \vdash \inact : \tend}
  \end{prooftree}
\quad
  \begin{prooftree}
    \hypo {\Gamma, X: \T \vdash \PP: \T}
    \infer1[\RULE{T-Rec}] {\Gamma \vdash \mu X. \PP: \T}
  \end{prooftree}
\quad
  \begin{prooftree}
    \hypo {X: \T \in \Gamma}
    \infer1[\RULE{T-Var}] {\Gamma \vdash X: \T}
  \end{prooftree}

\vspace*{3mm}

\begin{prooftree}
    \hypo {\Gamma, x: S \vdash \PP: \T}
    \infer1[\RULE{T-In}] {\Gamma \vdash \procin \pp{x} \PP: \tin \pp{\S} \T}
\end{prooftree}
\quad
\begin{prooftree}
    \hypo {\Gamma \vdash \PP: \T}
    \hypo {\Gamma \vdash \e: S}
    \infer2[\RULE{T-Out}] {\Gamma \vdash \procout \pp{\e} \PP: \tout \pp{S} \T}
  \end{prooftree}
\quad
\begin{prooftree}
    \hypo {\forall i \in \I.\:\Gamma \vdash \PP_i: \T_i}
    \infer1[\RULE{T-Bra}] {\Gamma \vdash \procbrasub \pp{l_i: \PP_i}{i \in \I}: \tbra \pp{l_i: \T_i}_{i \in \I}}
\end{prooftree}

\vspace*{3mm}

  \begin{prooftree}
    \hypo {\Gamma \vdash \PP: \T}
    \infer1[\RULE{T-Sel}] {\Gamma \vdash \procsel \pp{l} \PP: \tsel \pp{l: \T}}
  \end{prooftree}
\quad 
  \begin{prooftree}
    \hypo {\Gamma \vdash \e: \tbool}
    \hypo {\Gamma \vdash \PP: \T}
    \hypo {\Gamma \vdash \PP': \T}
    \infer3[\RULE{T-Cond}]{\Gamma \vdash \cond{\e}{\PP}{\PP'}: \T}
  \end{prooftree}
\end{spacing}
\fbox{Session $\vdash \M : \Delta$}\raggedright\\[2mm]
\begin{prooftree}
  \hypo{\forall i \in \I.\ \vdash \PP_i : \Delta(\pp_i)}
  \hypo{\pt{\G} \subseteq \{\pp_i \mid i \in \I\}\quad \Delta \associated \G}
  \infer2[\RULE{T-Sess}]{\provestop \prod_{i \in \I} \pp_i :: \PP_i \tri 
  \Delta}
\end{prooftree}
\ 
\begin{prooftree}
\hypo{\forall i \in \I.\ \proves \PP_i : \T_i}
\hypo{\safe{(\{\pp_i{:}\T_i\}_{i\in I})}}
\infer2[\RULE{B-Sess}]{\provesbot \prod_{i \in \I} \pp_i :: \PP_i \tri 
\{\pp_i{:}\T_i\}_{i\in I}}
\end{prooftree}
\end{center}
}
\caption{Typing system for processes and a session (\RULE{T-Sess} top-down  
and \RULE{B-Sess} bottom-up) 
\label{fig:processtyping}
}
\end{figure}

The system satisfies the subject reduction theorem (\cite[Theorem~3.21]{Ghilezan2019}), and \textbf{by construction}, 
a typed multiparty session $\M$ automatically 
satisfies desired properties defined in 
Definition~\ref{def:properties} 
\cite[Theorem 7]{YH2024} using any of 
the four projections.
%
%
The following theorems are proved in \cite{Scalas2019,YH2024,Ghilezan2019}. 

\begin{theorem}[Top-down soundness]\label{thm:topdown}
Suppose $\provestop \M \tri \Delta$ by {\rm\RULE{T-Sess}}.  
Then $\M$ is communication safe, deadlock-free and live.  
\end{theorem}

\paragraph{\bf Complexity of the top-down system}
Since checking $\pt{\G} \subseteq \{\pp_i \mid i \in \I\}$ 
is linear, the complexity of the top-down approach is 
a sum of the complexity of (i) projection and 
(iii-a) type checking of $\PP$ against (ii) a subtype of a projection
of $\G$ (Figure~\ref{fig:top_down}).  
In the case of coinductive projection, 
we additionally require checking $\G$ is balanced. 
As in Figure~\ref{fig:top_downtwo}, 
(iii-a) and (ii) are reduced to (ii-b) the minimum type inference
and (iii) the subtype checking. 
The complexity of (2-2) is given in
Theorems~\ref{thm:complexity_ghilezan_subtyping} 
and \ref{thm:quadratic_subtyping}.  
The rest of this selection develops  
the minimum type inference and analyses its complexity.



\subsection{Minimum Type Inference System}
\label{subsec:minimum}

\myparagraph{Overview.}
The minimum type inference system of process $P$ is constructed in 
two steps:
first, we build the \emph{constraint inference system}
for generating a constraint set of subtyping relations
of $\PP$;  
then we try to solve the constraint set to obtain the minimum type
of $P$. More precisely: 
\begin{enumerate}
\item \label{inf:one}
We construct \emph{the constraint inference system} (Definition~\ref{def:constraint_rules}) for deriving a constraint set $\C$ for process $P$.
\item \label{inf:two} We prove that the constraint inference system satisfies: 
\begin{enumerate}
\item
\label{inf:two:soundness}
\emph{soundness} (Theorem~\ref{thm:soundness_of_constraints}):
if a constraint set
$\C$ is derived from $P$, then $P$ is typable
by a solution of $\C$; and  
\item \emph{completeness}:
\label{inf:two:completeness} (Theorem~\ref{thm:completeness_of_constraints}): 
if $P$ is typable by type $\T$ and a constraint set
$\C$ is derived from $P$, then there exists a solution of $\C$
which corresponds to $\T$.
\end{enumerate}
\item \label{inf:three}  
We define the algorithm to find the \emph{minimum type graph} $\Gmin{\C}$ and
a set of sort constraints $\Csort$ given $P$ and $\C$ (Definition~\ref{def:minimum_type_graph_of_process}); 
\item \label{inf:four}  We prove that the algorithm satisfies: 
\begin{enumerate}
\item \label{inf:four:soundness} \emph{soundness} (Theorem~\ref{the:minimuminf:soundness}): the most general
solution $\Csort$ generated 
with $\Gmin{\C}$ by the algorithm corresponds to the solution of $\C$; and 

\item \label{inf:four:completeness} \emph{completeness} 
(Theorem \ref{thm:minimum_type_of_process}): $\Gmin{\C}$
substituted by the most general solution of $\Csort$ is a type graph of 
the minimum type of $P$. 
\end{enumerate}  
\end{enumerate}
In this subsection, we define (\ref{inf:one}) and prove 
(\ref{inf:two}). 
The second part 
(\ref{inf:three})--(\ref{inf:four}) will be 
presented in \Sec\ref{subsec:minimumtypes}. 

We start by introducing the definition of \emph{minimum types}.
\begin{definition}[Minimum types]\label{def:minimum_type}
  A \textit{minimum type} of a process $\PP$ is a type $\T$ such that:
  {\rm (1)} $\emptyset \vdash \PP: \T$; and 
  {\rm (2)} for all $\T'$ such that $\emptyset \vdash \PP: \T'$, there exists a sort substitution $\pi: \Sv \rightarrow \Sort$ such that $\T \subt \T'\pi$.
\end{definition}

We now introduce
the \emph{constraints} for \emph{unification} \cite[Chapter~22.3]{Pierce2002}. First, we extend the syntax of types 
$\T$ to include a countable set of \emph{type variables} $\Tv$, denoted by $\xi, \psi, \dots$;  
we also extend sorts $\S$ with a countable set of \emph{sort variables} $\Sv$,
ranged over $\alpha, \beta, \dots$. 
A \emph{constraint set} $\C$ includes elements of the form $\T\subt \T'$ and 
$\S = \S'$. 

\label{def:substitution}
Given a constraint set $\C$ on a set of free variables $\cX$, 
a \textit{type substitution} $\sigma: \Tv \prightarrow \Type$ is a partial function from type variables to types. 
Similarly, we define sort substitutions $\pi: \Sv \prightarrow \Sort$. 
We extend these substitutions to act on types and sorts respectively: $\sigma: \Type \prightarrow \Type$ and $\pi: \Sort \prightarrow \Sort$.

\begin{definition}[Constraint inference]
\label{def:constraint_rules}
Define \emph{contexts} $\Gamma::= \emptyset\ | \  \Gamma,x: \alpha \ | \ \Gamma,X: \xi$. Let us write $\dcup$ for disjoint union and alpha-convert all $\mu$-bound variables to be unique.
Figure~\ref{fig:mininfer} defines the judgements for rules 
for process $P$ and selected rules for expression $e$:\\[1mm]
\centerline{
  $\Gamma \vdash \e : \alpha \mid_\cX \C$
  \quad and\quad 
  $\Gamma \vdash \PP : \psi \mid_\cX \C$
}  
which state that
(1) any solution for $\C$ substituted into $\alpha$
gives a type for $\e$; and 
(2) any solution for $\C$ substituted into $\psi$ 
gives a type for $\PP$.
%
\end{definition}
Note that each judgement in the proof tree must type a process 
or an expression as a single type variable. 
In each rule we combine all the constraints $\C$ from the premises,
and add all fresh variables to $\mathcal{X}$.
\iffull{\ Explanations of the rules for expressions are found in
  Appendix~\ref{app:constraint_rules}}\else{\ The rules for expressions are
  standard; see \cite{UY2024}}\fi.
Rule \RULE{C-End} states that a nil must be typed as a supertype of $\tend$.
Rules \RULE{C-Var} and \RULE{C-Rec} are used for recursions: the type of a recursive process must be a subtype of the type of its corresponding 
recursion variable.
In rule \RULE{C-In}, we assume that $x$ has sort $\alpha$ and so the input must have a payload of $\alpha$.
In rule \RULE{C-Out}, 
the expression must have the same sort as the payload of the type.
Rules \RULE{C-Bra} and \RULE{C-Sel} ensure that the type variable has valid types with respect to each branch.
Rule \RULE{C-Cond} verifies that the expression is of sort $\tbool$ and the two sides of the conditional statement have the same type.



\iffull{\begin{figure}
{\small
\fbox{Expressions $\Gamma \vdash \e : \alpha \mid_\cX \C$}\raggedright
\begin{center}
\begin{spacing}{2}
  \begin{prooftree}
     \hypo {x: \alpha \in \Gamma}
     \infer1[\RULE{C-SortVar}] {\Gamma \vdash x: \alpha \mid_{\emptyset} \emptyset}
   \end{prooftree}
   \begin{prooftree}
     \hypo {\alpha \text{ fresh}}
     \infer1[\RULE{C-True}] {\Gamma \vdash \true: \alpha \mid_{\{\alpha\}} \{\alpha = \tbool\}}
   \end{prooftree}
   \begin{prooftree}
     \hypo {\alpha \text{ fresh}}
     \infer1[\RULE{C-False}] {\Gamma \vdash \false: \alpha \mid_{\{\alpha\}} \{\alpha = \tbool\}}
   \end{prooftree}
   
   \vspace*{0.7em}
   
   \begin{prooftree}
     \hypo {\alpha \text{ fresh}}
    \infer1[\RULE{C-Nat}] {\Gamma \vdash \valn: \alpha \mid_{\{\alpha\}} \{\alpha = \tnat\}}
   \end{prooftree}
   \quad
   \begin{prooftree}
     \hypo {\alpha \text{ fresh}}
     \infer1[\RULE{C-Int}] {\Gamma \vdash \valr: \alpha \mid_{\{\alpha\}} \{\alpha = \tint\}}
   \end{prooftree}

   \vspace*{0.7em}

    \begin{prooftree}
      \hypo {\Gamma \vdash \e: \alpha \mid_{\cX} \C}
      \hypo {\C' = \C \cup \{\alpha = \tint\}}
      \infer2[\RULE{C-Neg}] {\Gamma \vdash \fsqrt\e: \alpha \mid_{\cX} \C'}
 \end{prooftree}
  \quad
    \begin{prooftree}
     \hypo {\Gamma \vdash \e: \alpha \mid_\cX \C}
      \hypo {\C' = \C \cup \{\alpha = \tbool\}}
      \infer2[\RULE{C-Not}] {\Gamma \vdash \neg \e: \alpha \mid_{\cX} \C'}
    \end{prooftree}

    \vspace*{0.7em}

    \begin{prooftree}
        \hypo {\Gamma \vdash \e_1: \alpha_1 \mid_{\cX_1} \C_1}
        \hypo {\Gamma \vdash \e_2: \alpha_2 \mid_{\cX_2} \C_2}
        \hypo {\beta \text{ fresh}}
        \hypo {\C' = \C_1 \cup \C_2 \cup
          \{\alpha_1 = \tbool, \alpha_2 = \tbool, \beta = \tbool\}}
      \infer4[\RULE{C-Or}] {\Gamma \vdash \e_1 \vee \e_2: \beta \mid_{\cX_1 \dcup \cX_2 \dcup \{\beta\}} \C'}
    \end{prooftree}

    \vspace*{0.7em}

    \begin{prooftree}
     \hypo {\Gamma \vdash \e_1: \alpha_1 \mid_{\cX_1} \C_1}
     \hypo {\Gamma \vdash \e_2: \alpha_2 \mid_{\cX_2} \C_2}
     \hypo {\beta \text{ fresh}}
     \hypo {\C' = \C_1 \cup \C_2 \cup
       \{\alpha_1 = \beta, \alpha_2 = \beta\}}
     \infer4[\RULE{C-Nondet}] {\Gamma \vdash \e_1 \oplus \e_2: \beta \mid_{\cX_1 \dcup \cX_2 \dcup \{\beta\}} \C'}
 \end{prooftree}

    \vspace*{0.7em}

       \begin{prooftree}
     \hypo {\Gamma \vdash \e_1: \alpha_1 \mid_{\cX_1} \C_1}
      \hypo {\Gamma \vdash \e_2: \alpha_2 \mid_{\cX_2} \C_2}
      \hypo {\beta \text{ fresh}}
     \hypo {\C' = \C_1 \cup \C_2 \cup \{\alpha_1 = \beta, \alpha_2 = \beta, \beta = \tint\}}
      \infer4[\RULE{C-Add}] {\Gamma \vdash \e_1 + \e_2: \beta \mid_{\cX_1 \dcup \cX_2 \dcup \{\beta\}} \C'}
    \end{prooftree}
 \end{spacing}
\end{center}
\vspace*{-2em}
\fbox{Processes $\Gamma \vdash \PP : \psi \mid_\cX \C$}\raggedright
\vspace*{0.2em}
\begin{center}
\begin{spacing}{2}
\newcommand{\treesep}{.5em}
  \begin{prooftree}
    \hypo {\xi \text{ fresh}}
    \infer1[\RULE{C-End}] {\Gamma \vdash \inact: \xi \mid_{\{\xi\}} \{\tend \subt \xi\}}
  \end{prooftree}
\quad
\begin{prooftree}
    \hypo {X: \psi \in \Gamma}
    \hypo {\xi \text{ fresh}}
    \infer2[\RULE{C-Var}] {\Gamma \vdash X: \xi \mid_{\{\xi\}} \{\psi \subt \xi\}}
  \end{prooftree}
\quad
  \begin{prooftree}
    \hypo {\Gamma, X: \xi \vdash \PP: \xi \mid_\cX \C}
    \infer1[\RULE{C-Rec}] {\Gamma \vdash \mu X. \PP: \xi \mid_\cX \C}
  \end{prooftree}

\vspace*{3mm}

  \begin{prooftree}
    \hypo {\Gamma, x: \alpha \vdash \PP: \psi \mid_\cX \C}
    \hypo {\alpha \text{ fresh}}
    \hypo {\xi \text{ fresh}}
    \hypo {\C' = \C \cup \{\tin \pp{\alpha} \psi \subt \xi\}}
    \infer4[\RULE{C-In}] {\Gamma \vdash \procin \pp{x}\PP: \xi \mid_{\cX \dcup \{\xi, \alpha\}} \C'}
  \end{prooftree}

  \vspace*{0.7em}

  \begin{prooftree}
    \hypo {\Gamma \vdash \PP: \psi \mid_\cX \C_1}
    \hypo {\Gamma \vdash \e: \alpha \mid_{\cX'} \C_2}
    \hypo {\xi \text{ fresh}}
    \hypo {\C' = \C_1 \cup \C_2 \cup \{\tout \pp{\alpha} \psi \subt \xi\}}
    \infer4[\RULE{C-Out}] {\Gamma \vdash \procout \pp{\e}\PP: \xi \mid_{\cX \dcup \cX' \dcup \{\xi\}} \C'}
  \end{prooftree}

  \vspace*{0.7em}

  \begin{prooftree}
      \hypo {\forall i \in \I.\: \Gamma \vdash \PP_i: \psi_i \mid_{\cX_i} \C_i}
      \hypo {\xi \text{ fresh}}
      \hypo {\C' = \bigcup_{i \in \I} \C_i \cup \{\tbra \pp {l_i: \psi_i}_{i \in \I} \subt \xi\}}
      \hypo {\cX' = \dbigcup_{i \in \I} \cX_i \dcup \{\xi\}}
    \infer4[\RULE{C-Bra}]{\Gamma \vdash \procbra \pp{l_i: \PP_i}_{i \in \I}: \xi \mid_{\cX'} \C'}
  \end{prooftree}

  \vspace*{0.7em}

  \begin{prooftree}
      \hypo {\Gamma \vdash \PP: \psi \mid_{\cX} \C}
      \hypo {\xi \text{ fresh}}
      \hypo {\C' = \C \cup \{\tsel \pp {l: \psi} \subt \xi\}}
      \hypo {\cX' = \cX \dcup \{\xi\}}
    \infer4[\RULE{C-Sel}] {\Gamma \vdash \procsel \pp{l} \PP : \xi \mid_{\cX'} \C'}
  \end{prooftree}

\vspace*{0.7em}

  \begin{prooftree}
      \hypo {\Gamma \vdash \PP_i: \psi_i \mid_{\cX_i} \C_i}
      \hypo {\xi \text{ fresh}}
      \hypo {\Gamma \vdash \e: \alpha \mid_{\cX_3} \C_3}
      \hypo {\C' = \C_1 \cup \C_2 \cup \C_3 \cup \{\psi_1 \subt \xi, \psi_2 \subt \xi, \alpha = \tbool\}}
    \infer4[\RULE{C-Cond}] 
 {\Gamma \vdash \cond{\e}{\PP_1}{\PP_2}: \xi \mid_{\cX_1 \dcup \cX_2 \dcup \cX_3 \dcup \{\xi\}} \C'}
\end{prooftree}
\end{spacing}
\vspace{-7mm}
\end{center}
}
\caption{Constraint inference rules for expressions and processes}
\label{fig:mininfer}
\end{figure}
}\else{\input{fig_inference_selected}}\fi

\begin{example}[Constraint inference of a conditional]\label{ex:constraint_derivation_of_process_1}
Let $\PP = 
\cond{\true}{\procbra{\pp}{l_1:(\procsel{\pq}{l_2}{\inact}),\ l_3:\inact}}
{\procbra{\pp}{l_1:(\procsel{\pq}{l_4}{\inact}),\ l_5:\inact}}$. 
This process should have a type 
$\T = \tbra{\pp}{l_1: \tsel{\pq}{l_2: \tend, l_4: \tend}}$.

Let $\C_1' = \{\tend \subt \xi_1, \tsel{\pq}{l_2: \xi_1} \subt \xi_2\}$,
$\C_1 = \C_1' \cup \{\tend \subt \xi_3, \tbra{\pp}{l_1: \xi_2, l_3: \xi_3} \subt \xi_4\}$,
$\C_2' = \{\tend \subt \xi_5, \tsel{\pq}{l_4: \xi_5} \subt \xi_6\}$,
$\C_2 = \C_2' \cup \{\tend \subt \xi_7, \tbra{\pp}{l_1: \xi_6, l_5: \xi_7} \subt \xi_8\}$.
\vspace*{2mm}
  \begin{center}
    \small
    \newcounter{constraintderivationinline}
    \renewcommand{\theconstraintderivationinline}{\alph{constraintderivationinline}}
    \newcommand{\constraintderivationinline}[1]{(\refstepcounter{constraintderivationinline}\label{#1}\theconstraintderivationinline)}
    {\normalsize$\constraintderivationinline{ex:constraint_derivation_of_process_1:num1} =$}
    \begin{prooftree}
      \infer0[\RULE{C-End}]{\vdash \inact: \xi_1 \mid_{\{\xi_1\}} \{\tend \subt \xi_1\}}
      \infer1[\RULE{C-Sel}]{\vdash \procsel{\pq}{l_2: \inact}: \xi_2 \mid_{\{\xi_1, \xi_2\}} \C_1'}
      \infer0[\RULE{C-End}]{\vdash \inact: \xi_3 \mid_{\{\xi_1, \xi_2, \xi_3\}} \{\tend \subt \xi_3\}}
      \infer2[\RULE{C-Bra}]{\vdash \procbra{\pp}{l_1: \procsel{\pq}{l_2: \inact}, l_3: \inact}: \xi_4 \mid_{\{\xi_1, \xi_2, \xi_3, \xi_4\}} \C_1}
    \end{prooftree}

    \vspace{3mm}

    {\normalsize$\constraintderivationinline{ex:constraint_derivation_of_process_1:num2} =$}
    \begin{prooftree}
      \infer0[\RULE{C-End}]{\vdash \inact: \xi_5 \mid_{\{\xi_5\}} \{\tend \subt \xi_1\}}
      \infer1[\RULE{C-Sel}]{\vdash \procsel{\pq}{l_4: \inact}: \xi_6 \mid_{\{\xi_5, \xi_6\}} \C_2'}
      \infer0[\RULE{C-End}]{\vdash \inact: \xi_7 \mid_{\{\xi_5, \xi_6, \xi_7\}} \{\tend \subt \xi_3\}}
      \infer2[\RULE{C-Bra}]{\vdash \procbra{\pp}{l_1: \procsel{\pq}{l_4: \inact}, l_5: \inact}: \xi_8 \mid_{\{\xi_5, \xi_6, \xi_7, \xi_8\}} \C_2}
    \end{prooftree}
    
    \vspace{3mm}

    \begin{prooftree}
      \hypo{(\ref{ex:constraint_derivation_of_process_1:num1})}
      \hypo{(\ref{ex:constraint_derivation_of_process_1:num2})}
      \infer0[\RULE{C-True}]{\vdash \true: \alpha \mid_{\{\alpha\}} \{\tbool = \alpha\}}
      \infer3[\RULE{C-Cond}]{\vdash \PP: \xi \mid_{\{\xi_1, \xi_2, \xi_3, \xi_4, \xi_5, \xi_6, \xi_7, \xi_8, \alpha, \xi\}} \C_1 \cup \C_2 \cup \{\tbool = \alpha, \xi_4 \subt \xi, \xi_8 \subt \xi\}}
    \end{prooftree}

  \end{center}

\vspace*{1mm}

\noindent The last derived constraint set is $\C = \{
  \tend \subt \xi_1, \tsel{\pq}{l_2: \xi_1} \subt \xi_2,
  \tend \subt \xi_3, \tbra{\pp}{l_1: \xi_2, l_3: \xi_3} \subt \xi_4,
  \tend \subt \xi_5, \tsel{\pq}{l_4: \xi_5} \subt \xi_6,
  \tend \subt \xi_7, \tbra{\pp}{l_1: \xi_6, l_5: \xi_7} \subt \xi_8, 
  \tbool = \alpha, \xi_4 \subt \xi, \xi_8 \subt \xi
\}$.
Then $(\sigma, \pi)$ is a solution to $\C$ with
$\sigma = 
\{\xi \mapsto \T,
\xi_1 \mapsto \tend,
\xi_2 \mapsto \tsel{\pq}{l_2: \tend},
\xi_3 \mapsto \tend,
\xi_4 \mapsto \T,
\xi_5 \mapsto \tend,
\xi_6 \mapsto \tsel{\pq}{l_4: \tend},
\xi_7 \mapsto \tend,
\xi_8 \mapsto \T\}$
and $\pi = \{\alpha \mapsto \tbool\}$. 
Hence $\xi\pi\sigma = \T$, so by soundness (will be proven in 
Theorem~\ref{thm:soundness_of_constraints}), we derive $\vdash \PP : \T$.
\end{example}


\begin{example}[Constraint derivation of a process (2)]\label{ex:constraint_derivation_of_process_2}
Let $\PP = \mu X. \procin{\pp}{x}\procout{\pp}{x}X$. Then the constraint derivation of $\PP$ is the following:

  \begin{center}
    \small
    \begin{prooftree}[separation=1em]
      \infer0[\RULE{C-Var}]{X: \xi, x: \alpha \vdash X: \xi_1 \mid_{\{\xi_1\}} \{\xi \subt \xi_1\}}
      \infer0[\RULE{C-SortVar}]{X: \xi, x: \alpha \vdash x: \alpha \mid_\emptyset \emptyset}
      \infer2[\RULE{C-Out}]{X: \xi, x: \alpha \vdash \procout{\pp}{x}X: \xi_2 \mid_{\{\xi_1, \xi_2\}} \{\xi \subt \xi_1, \tout{\pp}{\alpha}\xi_1 \subt \xi_2\}}
      \infer1[\RULE{C-In}]{X: \xi \vdash \procin{\pp}{x}\procout{\pp}{x}X: \xi \mid_{\{\xi_1, \xi_2, \xi_3, \alpha\}} \C}
      \infer1[\RULE{C-Rec}]{\vdash \PP: \xi \mid_{\{\xi, \xi_1, \xi_2, \xi_3, \alpha\}} \C}
    \end{prooftree}
  \end{center}

\vspace*{1mm}  
  
\noindent where the generated set of constraints $\C = \{\xi \subt \xi_1, \tout{\pp}{\alpha}\xi_1 \subt \xi_2, \tin{\pp}{\alpha}\xi_2 \subt \xi\}$.
  $(\sigma, \pi)$ is a solution to $\C$ where $\sigma = \{\xi \mapsto
\mu\ty.\tin{\pp}{\alpha}\tout{\pp}{\alpha}\ty, \xi_1 \mapsto
\mu\ty.\tout{\pp}{\alpha}\tin{\pp}{\alpha}\ty, ,\xi_2 \mapsto
\mu\ty.\tin{\pp}{\alpha}\tout{\pp}{\alpha}\ty\}$ and $\pi = \{\alpha \mapsto \tint, \beta \mapsto \tint\}$.
  Soundness implies that $\vdash \PP : \T$ where $\T = \xi\pi\sigma = \mu\ty.\tin{\pp}{\tint}\tout{\pp}{\tint}\ty$; however, note that $\T$ is not the only type of $\PP$ (e.g.~$\mu\ty.\tin{\pp}{\tbool}\tin{\pp}{\tbool}\ty$ is another type of $\PP$).
\end{example}

\begin{definition}[Solutions]\label{def:solution}
Given that $\C$ is a constraint set, we say that ($\sigma, \pi)$ is a solution to $\C$ if for all $\xi \subt \psi \in \C$, $\xi\sigma\pi \subt \psi\sigma\pi$, and for all $\alpha = \beta \in \C$, $\alpha\pi = \beta\pi$.
\end{definition}

\begin{restatable}[Uniqueness of constraints]{lemma}{uniqueconstraintset}
  \label{thm:unique_constraint_set}
\proofreference{Appendix~\ref{app:subsec:soundcomplete}}
Given a process $\PP$, there is a unique constraint set $\C$ (up to renaming) such that $\Gamma \vdash \PP: \xi \mid_\cX \C$ for some $\Gamma, \xi, \cX$. That is, if $\C$ and $\C'$ are two such sets, then there exists one-to-one $\sigma: \Sv \rightarrow \Sv$, $\pi: \Tv \rightarrow \Tv$ such that $\C\sigma\pi = \C'$.
\end{restatable}
\begin{proof}
By inversion, there is only one rule that can be used for any $\PP$, thus the trees have the same shape.
We map the fresh variables from one tree to the corresponding 
fresh variables in the other.
We prove this by induction on the proof trees of 
$\Gamma \vdash \PP: \xi \mid_\cX \C$ 
and $\Gamma' \vdash \PP: \xi' \mid_{\cX'} \C'$.
\end{proof}

\begin{lemma}
\label{lem:possibleconstraints}  
A constraint generated from the rules in 
Definition~\ref{def:constraint_rules} is of the form of 
either $\tend \subt \xi$, $\xi \subt \psi$,  
$\tin \pp{\beta} \xi \subt \psi$,
$\tout \pp{\beta} \xi \subt \psi$,
$\tbra \pp{l_i: \xi_i}_{i \in \I} \subt \psi$,  
$\tsel \pp{l_i: \xi_i}_{i \in \I} \subt \psi$,  
or $\S = \S'$ for $\S, \S' \in \Sort$; 
and the right-hand side type variables (e.g. $\psi$ in 
$\T \subt \psi$) in $\C$ are pairwise distinct. 
\end{lemma}
\begin{restatable}[Soundness of constraints]{theorem}{soundminimum}
  \label{thm:soundness_of_constraints}
\proofreference{Appendix~\ref{app:subsec:soundcomplete}}  
If $\vdash \PP: \xi \mid_\cX \C$ is derivable and $(\sigma, \pi)$ 
is a solution to $\C$, then $\PP$ has type $\xi\sigma\pi$.
\end{restatable}
\begin{proof}
We show that if $\Gamma \vdash \PP: \xi \mid_\cX \C$ is derivable and $(\sigma, \pi)$ is a solution to $\C$,
  then there is a typing derivation for $\Gamma\sigma\pi \vdash \PP: \xi\sigma\pi$
  (where we define $\Gamma\sigma\pi$ to apply $\sigma\pi$ to each type in $\Gamma$).
  For sorts, we prove: if $\Gamma \vdash \e: \alpha \mid_\cX \C$
  is derivable, $\Gamma\sigma\pi \vdash \e: \alpha\sigma\pi$.
  We prove these by induction on the proof tree of the constraint derivations. 
\end{proof}

\begin{restatable}[Completeness of constraints]{theorem}{completeminimum}
  \label{thm:completeness_of_constraints}
\proofreference{Appendix~\ref{app:subsec:soundcomplete}}
Let $\vdash \PP: \T_0$ and $\vdash \PP: \xi \mid_\cX \C$. Then there exists a solution $(\sigma', \pi')$ of $\C$ such that $\xi\sigma'\pi' = \T_0$.
\end{restatable}
\begin{proof}
We show that if $\Gamma \vdash \PP: \xi \mid_\cX \C$, and $\Gamma\sigma\pi \vdash \PP: \T_0$, such that $\dom{\sigma} \cap \cX = \emptyset$ and $\dom{\pi} \cap \cX = \emptyset$,
then there exists a solution $(\sigma', \pi')$ of $\C$ such that $\xi\sigma'\pi' = \T_0$, $\sigma' \minus \cX = \sigma$, and $\pi' \minus \cX = \pi$.
We also show the analogous statement for sorts. The proof is done by
postponing the application of \RULE{T-Sub} in the derivations of
$\PP$.  
\end{proof}

\subsection{Solving the Minimum Types}
\label{subsec:minimumtypes}
This section proves that there exists an algorithm to find the 
minimum type for typable $\PP$ ((\ref{inf:three}) in
\Sec~\ref{section:minimum_typing}). The algorithm constructs
a \emph{minimum type graph} (denoted by $\Gmin{\C}$)
and \emph{sort constraints} $\Csort$ 
with respect to $\vdash \PP: \xi \mid_\cX \C$.
We then prove \emph{soundness} (\ref{inf:four:soundness}) and 
\emph{completeness} (\ref{inf:four:completeness}) of the algorithm, deriving the minimum type of $\PP$ as a type graph instantiated by the most general
solution of $\Csort$. 




\begin{definition}[Most general solutions for sorts]\label{def:most_general_pi}
  Let $\pi$ solve the sort constraints $\Csort$ if $\S_1\pi = \S_2\pi$ for all $\S_1 = \S_2 \in \Csort$.
  Then $\pi$ is a most general solution of $\Csort$ if, for all $\pi'$ solving $\Csort$,
  there exists $\pi''$ such that $\alpha\pi' = \alpha\pi\pi''$ for all $\alpha \in \Sv$.
\end{definition}

Note that, given a set of sort constraints, we may use a union-find to identify the equivalence classes of equality constraints. Then we can find the most general solution of the constraints using these equivalence classes.

\begin{definition} 
Define $\tr(\C)$ to be a \emph{minimum constraint set} 
created by the following procedure: while $\xi \subt \psi$ occurs in $\C$, 
perform the substitution $\sub{\xi}{\psi}$ on $\C$.
\end{definition}



Below Lemma~\ref{thm:remove_transitive} tells us that it is sufficient to find a solution to $\tr(\C)$, which can be used to reconstruct a solution to $\C$.

\begin{lemma}\label{thm:remove_transitive} 
For constraint set $\C$ created from Definition~\ref{def:constraint_rules} containing $\xi \subt \psi$, if $(\sigma, \pi)$ is a solution to $\C' = \C[\xi/\psi]$ then $(\sigma_{[\psi \mapsto \sigma(\xi)]}, \pi)$ is a solution to $\C$.
\end{lemma}
\begin{proof}
  By considering each rule in $\C$, and using transitivity of subtyping.
\end{proof}

We now define the algorithm to generate a pair of 
the minimum type graph $\Gmin{\C}$ 
and sort constraint $\Csort$ from $\vdash \PP: \xi \mid_\cX \C$. 
Intuitively, each node in the minimum type graph corresponds to the set of 
\emph{expected behaviours} of the process.

\begin{definition}[Minimum type graph]\label{def:minimum_type_graph_of_process}
Let $\set{\xi_i}_{i\in I} \in \powerset(\Tv)$ denote 
a subset of type variables. 
\emph{The minimum type graph} (denoted by $\Gmin{\C}$) 
with \emph{the constraint} $\Csort$ of 
$\vdash \PP: \xi \mid_\cX \C'$ with $\C = \tr(\C')$
is a directed connected graph where  
(1) a node 
is a subset of type variables in $\C$ or $\Skip$; and 
(2) the edges are given 
by the transition $\trans{\gell}$ defined below.

A sort constraint is generated from $\set{\xi}$ 
starting from the initial constraint 
$\Csort=\set{S_1=S_2 \ \SEP \ S_1=S_2\in \C}$. 
New constratins will be added to $\Csort$ when rule
\RULE{M-IO} is applied.  

The transition relation 
$\trans{\gell}$ which generates 
a type graph and the sort constraint $\Csort$ is defined below 
(Lemma~\ref{lem:possibleconstraints} constrains the possible
elements of $\C$).




\begin{itemize}[leftmargin=*]
\item \RULE{M-End} 
If $\tend\subt \xi_i \in \C$
for all $\xi_i\in \set{\xi_i}_{i\in I}$, 
then ${\set{\xi_i}_{i\in I}} \trans{\trend} \Skip$. 

\item \RULE{M-IO} 
If $\tdag\pp {\beta_{i}} \xi_{i}\subt \psi_k \in \C$ such that 
$\dagger\in \{ !,?\}$ for all $\psi_k\in \set{\psi_k}_{k\in K}$ with $i\in I$, 
then ${\set{\psi_k}_{k\in K}}\trans{\trdag{\pp}{\alpha}}
{\set{\xi_{i}}_{i\in I}}$; 
and update $\Csort$ to $\Csort\cup \set{\alpha=\beta_i}_{i\in I}$.   

\item \RULE{M-Sel} 
If $\tsel{\pp}{l_j:{\xi_{ij}}}_{j\in \J_{i}, i\in I} \subt \psi_{k} \in \C$
for all $\psi_{k} \in 
\set{\psi_{k}}_{k\in K}$,  
then $\egnode{\set{\psi_{k}}_{k\in K}}
\trans{\trsel{\pp}{l_j}}
\egnode{\set{\xi_{ij}}_{j \in J_i, i\in I}}$ for each 
$j\in \bigcup_{i\in I} J_i$. 

\item \RULE{M-Bra}
If $\tbra{\pp}{l_j:{\xi_{ij}}}_{j\in \J_{i},i\in I} \subt \psi_{k} \in \C$
for all $\psi_{k} \in \set{\psi_{k}}_{k\in K}$,  
then $\egnode{\set{\psi_{k}}_{k\in K}}\trans{\trbra{\pp}{l_{j}}}
\egnode{\set{\xi_{ij}}_{j\in J_i, i\in I}}$ for each 
$j\in \bigcap_{i\in I} J_i$. 
\end{itemize}

\begin{figure}[t]
\footnotesize
\begin{tabular}{cc}
\begin{tabular}{c}
\begin{tikzpicture}[>=Latex, node distance=2cm]
    \node (n1) [smallrect] {$\{\xi\}$};
    \node (n2) [smallrect, right of=n1] {$\{\xi_2, \xi_6\}$};
    \node (n3) [smallrect, right of=n2, yshift=.4cm] {$\{\xi_1\}$};
    \node (n4) [smallrect, right of=n2, yshift=-.4cm] {$\{\xi_5\}$};
    \node (n5) [smallrect, right of=n2, xshift=2cm] {$\kf{Skip}$};

    \draw[arrow] (n1) to node[auto] {$\trbra{\pp}{l_1}$} (n2);
    \draw[arrow] (n2) to node[auto, xshift=3mm] {$\trsel{\pp}{l_2}$} (n3);
    \draw[arrow] (n2) to node[auto, swap, xshift=3mm] {$\trsel{\pp}{l_4}$} (n4);
    \draw[arrow] (n3) to node[auto, xshift=-3mm] {$\trend$} (n5);
    \draw[arrow] (n4) to node[auto, swap, xshift=-3mm] {$\trend$} (n5);
  \end{tikzpicture}\\
(a)
\end{tabular}
& 
\begin{tabular}{c}
    \begin{tikzpicture}[>=Latex, node distance=3cm]
      \node (n1) [smallrect] {$\{\xi\}$};
      \node (n2) [smallrect, right of=n1] {$\{\xi_2\}$};

      \draw[arrow] (n1) to[out=20, in=160] node[auto] {$\trin{\pp}{\alpha_1}$} (n2);
      \draw[arrow] (n2) to[out=-160, in=-20] node[auto, swap] {$\trout{\pp}{\alpha_2}$} (n1);
    \end{tikzpicture}\\[2mm]
(b)\\[2mm]
\end{tabular}
\end{tabular}
\vspace{-5mm}
\caption{Minimum type graphs 
of Example~\ref{ex:building_minimum_type_graph_1} (a) 
and Example~\ref{ex:building_minimum_type_graph_2} (b) 
}
\label{fig:min_type_graph_1}
\end{figure}

\noindent 
The type graph for $\set{\xi_i}_{i\in I}$, denoted by 
$\Gmin{\C}(\set{\xi_i}_{i\in I})$ is the graph reachable by the transitions from 
$\set{\xi_i}_{i\in I}$.  
If the generated type graph is invalid (i.e. not a type graph 
of any local type),
it is \emph{undefined}. 
\end{definition}
We explain the transition rules in Definition~\ref{def:minimum_type_graph_of_process}.
Intuitively, $\{\xi_i\}_{i \in \I}$ is a node that implements the behaviours of $\xi_i$ for all $i \in \I$. 
Rule \RULE{M-End} means a node whose type variables are all supertypes of $\tend$ can be typed by $\tend$;
rule \RULE{M-IO} defines  
a transit with an input or output action with a fresh sort variable $\alpha$. 
Sort constraints $\alpha=\beta_i$ are added to $\Csort$ to enforce that the payloads have the same types;
rule \RULE{M-Sel} adds the edge $\trsel{\pp}{l_j}$ 
for each $j$ in the \emph{union} of $J_i$. This is because  
a supertype of the selection type has \emph{more} choices so that all possible 
edges should be generated; and rule \RULE{M-Bra} states its 
dual--edge $\trbra{\pp}{l_j}$ for each $j$ 
in the \emph{intersection} of $J_i$ should be generated.

\begin{example}[Minimum type graph (1)]\label{ex:building_minimum_type_graph_1}
  We continue Example~\ref{ex:constraint_derivation_of_process_1}. We have (up to renaming of type variables)
  $\tr(\C) = \{
    \tend \subt \xi_1, \tsel{\pq}{l_2: \xi_1} \subt \xi_2,
    \tend \subt \xi_3, \tbra{\pp}{l_1: \xi_2, l_3: \xi_3} \subt \xi,
    \tend \subt \xi_5, \tsel{\pq}{l_4: \xi_5} \subt \xi_6,
    \tend \subt \xi_7, \tbra{\pp}{l_1: \xi_6, l_5: \xi_7} \subt \xi,
    \tbool = \alpha
  \}$.
  Then the minimum type graph $\Gmin{\tr(\C)}$ is given in Figure~\ref{fig:min_type_graph_1}.
 This graph represents that after branching from the label $l_1$,
the process can be in either side of the conditional statement, thus it must respect the behaviours of both.
  This is represented by the node being a two element-set, with one element corresponding to each selection. 
After selecting the labels $l_2$ or $l_4$, this nondeterminism is resolved, as each side of the conditional statement has a disjoint set of branches. Thus we have two separate nodes with one single element, which can now have different behaviours.

  Formally, we have that the sort constraints are $\Csort = \{\alpha = \tbool\}$.
  A most general solution is $\pi = \{\alpha \mapsto \tbool\}$, so the minimum type is the type corresponding to $\Gmin{\tr(\C)}$, with $\pi$ applied to each transition: 
$\Tmin = \tbra{\pp}{l_1: \tsel{\pq}{l_2: \tend, l_4: \tend}}$.
\end{example}

\begin{example}[Minimum type graph (2)]\label{ex:building_minimum_type_graph_2}
We continue Example~\ref{ex:constraint_derivation_of_process_2}. We have (up to renaming of type variables) $\tr(\C) = \{
   \tout{\pp}{\alpha}\xi \subt \xi_2,
   \tin{\pp}{\alpha}\xi_2 \subt \xi
 \}$.
Its minimum type graph is given in Figure~\ref{fig:min_type_graph_1}(b). 
The sort constraints are $\Csort = \{
    \alpha = \alpha,
    \alpha = \alpha_1,
    \alpha = \alpha_2
\}$.
The most general solution is $\pi = \{\alpha \mapsto \alpha, \alpha_1 \mapsto \alpha, \alpha_2 \mapsto \alpha\}$. 
Thus the minimum type is the type corresponding to the graph 
$\Gmin{\tr(\C)}$ with $\pi$ applied to the transitions: $\Tmin = 
\mu\ty.\trin{\pp}{\alpha};\trout{\pp}{\alpha};\ty$.
\end{example}
Example~\ref{ex:untypable_process} shows how the system works 
for the case of an untypable process. 
\begin{example}[Untypable process]
\label{ex:untypable_process}
Let $\PP = \procsel{\pp}{l}\cond{\false}{\procout{\pp}{1}\inact}{\procin{\pp}{x}\inact}$. This process is not typable because the the 
branches of the conditional statement are not instances of a common type. 
The constraints are $\tr(\C) = \{
     \alpha_1 = \tbool,
     \alpha_2 = \tint,
     \tout{\pp}{\alpha_2}\xi_1 \subt \xi_3,
     \tin{\pp}{\alpha_3}\xi_2 \subt \xi_3,
     \tsel{\pp}{l: \xi_3} \subt \xi
   \}$,
with $\vdash \PP: \xi \mid_\cX \C$.
If we try to build a minimum type graph, 
it fails at
{\tiny
   \begin{tikzpicture}[>=Latex, node distance=1.5cm]
     \node (n1) [smallrect] {$\{\xi\}$};
     \node (n2) [smallrect, right of=n1] {$\{\xi_3\}$};
     \draw[arrow] (n1) to node[auto] {$\trsel{\pp}{l}$} (n2);
     \end{tikzpicture}
}
since $\{\xi_3\}$ does not meet the condition of \RULE{M-IO} in  
Definition~\ref{def:minimum_type_graph_of_process}, hence 
the minimum type graph is undefined. Therefore no minimum type exists.
\end{example}

The following lemma states that set inclusion of type variables 
in nodes of the minimum type graph gives rise to the subtyping relation 
of types. 

\begin{restatable}[Subsets as subtyping]{lemma}{subsetsubtyping}
  \label{def:subsets_are_subtypes}
Let $\Gmin{\C}$ be a minimum type graph and $\pi$ be a solution of $\Csort$ generated with $\Gmin{\C}$.   
Suppose $\set{\xi_k}_{k\in K}$ is a node in $\Gmin{\C}$. 
Then for all non-empty $I,J$ 
such that $I\subseteq J\subseteq K$, 
$\Gmin{\C}(\set{\xi_i}_{i\in I})\pi\subtsim 
\Gmin{\C}(\set{\xi_j}_{j\in J})\pi$. I.e.,   
$\T_{1}\pi \subt \T_{2}\pi$ where $\T_{I}$ and 
$\T_{2}$ correspond to types of the type graphs 
$\Gmin{\C}(\set{\xi_i}_{i\in I})$ and 
$\Gmin{\C}(\set{\xi_j}_{j\in J})$, respectively. 
\end{restatable}

\begin{example}[Type graphs and subtyping]\label{ex:type_graphs_subtyping}
We demonstrate that the subset relation $\subseteq$ of variables at a node 
corresponds to the subtyping relation $\subtsim$ of types with two examples. \\
{\bf (1) Selection subtyping:}\ 
Consider the node $\set{\xi_2,\xi_6}$ in
Figure~\ref{fig:min_type_graph_1}(a) (Example~\ref{ex:building_minimum_type_graph_1}).  
The type corresponding to 
$\Gmin{\C}(\set{\xi_2,\xi_6})$ is 
$\T_{2,6}=\tsel{\pq}{l_2: \tend, l_4: \tend}$ and 
the type corresponding to 
$\Gmin{\C}(\set{\xi_2})$ 
is $\T_{2}=\tsel{\pq}{l_2: \tend}$. Hence 
$\set{\xi_2}\subseteq \set{\xi_2,\xi_6}$ implies 
$\Gmin{\C}(\set{\xi_2})\subtsim \Gmin{\C}(\set{\xi_2,\xi_6})$ 
and $\T_{2}\subt\T_{2,6}$.\\ 
{\bf (2) Branching subtyping:}\ 
The following process with branchings demonstrates a subset of variables 
produces \emph{more} branchings (the opposite to the above selection example). 
Consider 
  \begin{equation*}
Q\ = \ \cond{\true}{\procout\pp{1}\procbra{\pp}{l_1: \inact, l_3: (\procsel{\pp}{l_4} \inact)}}{\procout\pp{1}\procbra{\pp}{l_2: \inact, l_3: (\procsel{\pp}{l_5} \inact)}}
  \end{equation*}
Then, $\vdash Q: \xi \mid_\cX \C$ with
  $
    \tr(\C) = \{
      \tin{\pp}{\alpha_1}\xi_1 \subt \xi,
      \tbra{\pp}{l_1: \xi_2, l_3: \xi_3} \subt \xi_1,
      \tend \subt \xi_2,
      \tsel{\pp}{l_4: \xi_4} \subt \xi_3,
      \tend \subt \xi_4,
      \tin{\pp}{\alpha_2}\xi_5 \subt \xi,
      \tbra{\pp}{l_2: \xi_6, l_3: \xi_7} \subt \xi_5,
      \tend \subt \xi_6,
      \tsel{\pp}{l_5: \xi_8} \subt \xi_7,
      \tend \subt \xi_8,
      \alpha_1 = \tint,
      \alpha_2 = \tint
    \}
  $.
The minimum type graph of $Q$ is given in 
Figure~\ref{fig:min_type_graph_3}(a). 

Now let 
$Q_1=\procbra{\pp}{l_1: \inact, l_3: (\procsel{\pp}{l_4} \inact)}$ 
in the truth branch of $P$, which has type $\T_1=\tbra{\pp}{l_1{:}\tend, l_3{:}
\tsel{\pp}{l_4{:}\tend}}$. Its type graph $\Gmin{\C}(\set{\xi_1})$ is given  
in Figure~\ref{fig:min_type_graph_3}(b). 
Consider $\Gmin{\C}(\set{\xi_1,\xi_5})$ in 
Figure~\ref{fig:min_type_graph_3}(a). The corresponding type is 
$\T_{1,5}=\tbra{\pp}{l_3:\tsel{\pp}{l_4:\tend}}$. Hence 
$\set{\xi_1} \subseteq \set{\xi_1,\xi_5}$ implies 
$\Gmin{\C}(\set{\xi_1}) \subtsim \Gmin{\C}(\set{\xi_1,\xi_5})$  
and $\T_{1}\subt\T_{1,5}$. 
\end{example}

\begin{figure}[t]
  \footnotesize
  \begin{tabular}{cc}
  \begin{tabular}{c}
  \begin{tikzpicture}[>=Latex, node distance=2cm]
    \node (n1) [smallrect] {$\{\xi\}$};
    \node (n2) [smallrect, right of=n1] {$\{\xi_1, \xi_5\}$};
    \node (n3) [smallrect, right of=n2] {$\{\xi_3, \xi_7\}$};
    \node (n4) [smallrect, right of=n3, yshift=.7cm, xshift=-1cm] {$\{\xi_4\}$};
    \node (n5) [smallrect, right of=n3, yshift=-.7cm, xshift=-1cm] {$\{\xi_8\}$};
    \node (n6) [smallrect, right of=n2, xshift=2cm] {$\Skip$};

    \draw[arrow] (n1) to node[auto] {$\trin{\pp}{\beta}$} (n2);
    \draw[arrow] (n2) to node[auto] {$\trbra{\pp}{l_3}$} (n3);
    \draw[arrow] (n3) to node[auto] {$\trsel{\pp}{l_4}$} (n4);
    \draw[arrow] (n3) to node[auto, swap] {$\trsel{\pp}{l_5}$} (n5);
    \draw[arrow] (n4) to node[auto] {$\trend$} (n6);
    \draw[arrow] (n5) to node[auto, swap] {$\trend$} (n6);
    \end{tikzpicture}\\
  (a)
  \end{tabular}
  & 
  \begin{tabular}{c}
      \begin{tikzpicture}[>=Latex, node distance=2cm]
        \node (n2) [smallrect] {$\{\xi_1\}$};
        \node (n3) [smallrect, right of=n2, xshift=-1cm, yshift=.7cm] {$\{\xi_3\}$};
        \node (n4) [smallrect, right of=n3] {$\{\xi_4\}$};
        \node (n6) [smallrect, right of=n4] {$\Skip$};
        \node (n7) [smallrect, right of=n2, xshift=-1cm, yshift=-.7cm] {$\{\xi_2\}$};
    
        \draw[arrow] (n2) to node[auto] {$\trbra{\pp}{l_3}$} (n3);
        \draw[arrow] (n3) to node[auto] {$\trsel{\pp}{l_4}$} (n4);
        \draw[arrow] (n4) to node[auto] {$\trend$} (n6);
        \draw[arrow] (n2) to node[auto, swap] {$\trbra{\pp}{l_1}$} (n7);
        \draw[arrow] (n7) to node[auto, swap] {$\trend$} (n6);
      \end{tikzpicture}\\
  (b)
  \end{tabular}
  \end{tabular}
  \vspace{-3mm}
  \caption{Minimum type graphs 
  of Example~\ref{ex:type_graphs_subtyping}
  }
\label{fig:min_type_graph_3}
\end{figure}


\begin{restatable}[Soundness of the minimum type inference]{theorem}{theminsound}\proofreference{Appendix~\ref{app:solveminimum}} 
\label{the:minimuminf:soundness}
Let $\vdash \PP: \xi \mid_\cX \C'$ and $\C = \tr(\C')$. 
Let the minimum type graph be $\Gmin{\C}$ with sort constraints $\Csort$.
Let $\sigma = \{\psi \mapsto \T_\psi 
\SEP \psi \in \C\}$ where $\T_\psi$ is a type which corresponds to 
$\Gmin{\C}(\set{\psi})$. 
Assume $\pi$ is a most general solution of $\Csort$. 
Then $(\sigma, \pi)$ is a solution to $\C$.
\end{restatable}
\begin{proof}
All sort constraints are satisfied because $\pi$ solves $\Csort$ which is a superset of the sort constraints of $\C$. We prove that all type constraints are satisfied by $\sigma\pi$ 
using $\subtsim$. 
\end{proof}


Below we write $\Gmin{\C}\pi$ to be the graph with $\pi$ applied to all transitions.

\begin{restatable}[Completeness of the minimum type inference]{theorem}{themincomplete}
  \label{thm:minimum_type_of_process}
\proofreference{Appendix~\ref{app:solveminimum}}  
Let $\vdash \PP: \xi \mid_\cX \C'$, and $\C = \tr(\C')$, such that $\C$ is solvable.
Assume $\Gmin{\C}$ and $\Csort$ are the minimum graph and sort constraints 
of $\vdash \PP: \xi \mid_\cX \C$. 
If $\pi$ is the most general solution to $\Csort$, then $\Gmin{\C}\pi$ is a type graph (Definition~\ref{def:typegraph}) that corresponds to a minimum type (Definition~\ref{def:minimum_type}) $\Tmin$ of $\PP$.
\end{restatable}
\begin{proof}
Let $(\sigma, \pi)$ be a solution to $\C$. 
Let $\T = \xi\sigma\pi$.
By Theorems \ref{thm:soundness_of_constraints} and \ref{thm:completeness_of_constraints}, it suffices to show that there exists $\pi''$ such that $\Tmin\pi'' \subt \T$.
The proof involves defining a simulation $\psubts$ on 
$\Gmin{\C}(\set{\xi_i}_{i\in I})\pi'$ and $\Sub(\T)$ 
inductively matching each subformula of $\T$ with 
corresponding type variables. 
We show $\psubts$ is a type simulation proving two lemmas. 
\end{proof}

\begin{corollary}
  If $\vdash \PP: \T$ is derivable, then there exists a minimum type graph for $\T$.
\end{corollary}
\begin{proof}
\iffull{By Lemma~\ref{thm:unique_domain_gamma}}\else{First}\fi, there exists a derivation for $\vdash \PP: \xi \mid_\cX \C$.
  Thus, by Theorem~\ref{thm:completeness_of_constraints}, there exists a solution $(\sigma, \pi)$ of $\C$.
  Therefore, a solution also exists for $\tr(C)$ by Lemma~\ref{thm:remove_transitive}.
  The result follows by Theorem~\ref{thm:minimum_type_of_process}.
\end{proof}

\subsection{Complexity of Minimum Type Inference}
\label{subsec:min:complexity}
\begin{definition}[Size of expressions, processes and sessions]\label{def:size_of_process}
\emph{The size of an expression} $\e$, denoted $\size{\e}$ is 
defined as:
$|\true| = |\false| = |\valn| = |\vali| = |x| = 1$,
$|\neg \e| = |\fsqrt{\e}| = |\e| + 1$, and
$|\e \vee \e'| = |\e + \e'| = |\e \oplus \e'| = |\e| + |\e'| + 1$.
\emph{The size of a process} $\PP$, denoted $|\PP|$, 
is  defined as:
$|\inact| = |X| = 1$, 
$|\mu X. \PP| = |\PP| + 1$, 
$|\procin{\pp}{x} \PP| = |\procout{\pp}{\e} \PP| = |\PP| + |\e| + 1$,  
$|\procsel{\pp}{l}{\PP}| = |\PP| + 1$, and
$|\procbrasub{\pp}{l_i: \PP_i}{i\in \I}| = \sum_{i\in \I} |\PP_i| + 1$, and
$|\cond{\e}{\PP_1}{\PP_2}| = |\PP_1| + |\PP_2| + |\e| + 1$.
\emph{The size of a session} $\M$, denoted $\size{\M}$, 
is defined as: $\size{\pa \pp \PP}=\size{\PP}+1$, 
and $\size{\M\ | \ \M'}=\size{\M}+\size{\M'}+1$. 
\end{definition}

\begin{lemma} \label{thm:constraint_proof_tree_linear}
\label{thm:constraint_set_linear}
The number of judgements in a proof tree of $\Gamma \vdash \PP: \xi \mid_\cX \C$ is $\leq |\PP|$ and $|\C|$ = $\bigO{|\PP|}$.
\end{lemma}
\begin{proof}
By induction on the proof tree. 
Note that we never unfold recursions in the judgement and the sum of the sizes of the processes in the premises is strictly less than the size of $\PP$.
Thus, the number of judgements in the proof tree is bounded by $|\PP|$.
We can also easily check 
the number of constraints added per judgement in the proof tree 
is bounded by a constant.
\end{proof}

\begin{theorem} 
\label{thm:typeinfer:complexity}
The minimum type of a process $\PP$ can be found 
in $\bigO{n\cdot 2^n}$, where $n = |\PP|$.
\end{theorem}
\begin{proof}
Finding the subset graph of Definition~\ref{def:minimum_type_graph_of_process} requires $\bigO{n\cdot 2^n}$ time. First, building the set of constraints $\C$ and finding $\tr(C)$ take polynomial time in $n$, as there are only polynomially many type constraints.
  The number of nodes in the proof tree for constraint derivation is at most $n$ by Lemma~\ref{thm:constraint_proof_tree_linear}, and at most one new type variable is introduced per rule.
  Therefore, there are $\bigO{2^n}$ nodes in the subset graph. The transitions outgoing from $\mathcal{S}$ can be computed in $\bigO{n}$ time, as the number of constraints is $\bigO{n}$. Therefore, the total time complexity is $\bigO{n\cdot 2^n}$.
\end{proof}

We can show an explicit example where the minimum type has an exponential size.

\begin{theorem}\label{thm:worst_case_minimum_type} There exists a process $\PP$ s.\ t.\ any minimum type of $\PP$ has size exponential in $\size{\PP}$.
\end{theorem}
\begin{proof}
Define $\PP_n = \mu X. \procbra{\pp}{l_1: \procbra{\pp}{l_1: \dots
    \procbra{\pp}{l_1:\procbra{\pp}{l_1: X, l_2: X}\dots}}}$ where $n$ is the
number of $\& l_1$-nested branches (i.e., $\size{\PP_n}=n+4$).
This has the minimum type $\T_n = \mu\ty. \tbra{\pp}{l_1:
  \tbra{\pp}{\dots \tbra{\pp}{l_1: \tbra{\pp}{l_1: \ty, l_2:
        \ty}\dots}}}$ where there are $n$ number of $\& l_1$ branch nested 
constructs in $\T_n$ (i.e., $\size{\T_n}=n+4$).

Consider $\PP = \cond{\e}{\PP^{(1)}}{\PP^{(2)}}$ for some expression $\e$. If the minimum types of $\PP^{(1)}$ and $\PP^{(2)}$ are $\T_n$ and $\T_m$ respectively, then the minimum type of $\PP$ is $\T_{\lcm(n, m)}$, as $l_2$ is guaranteed to be safely received only when each cycle is synchronised.
  Thus, for integers $n_1, \dots, n_k$, the process $\PP^{(k)}$, defined by
$\PP^{(1)} = \PP_{n_1}$ and 
$\PP^{(k)} = \cond{\e}{\PP_{n_k}}{\PP^{(k-1)}}$
  has a minimum type of $\T_{\lcm(n_1, \dots, n_k)}$, which in general has size exponential in $\sum_{i \in \{1, \dots, k\}} n_i$.
\end{proof}

\section{Complexity Analysis of Bottom-Up MPST} \label{section:modelchecking}
\subsection{Bottom-Up Multiparty Session Type System} \label{section:checking_safety_properties}
We now start the analysis of \emph{bottom-up} MPST.
Recall a \emph{typing context} $\Delta$ defined in \Sec\ref{subsec:topdown}. 
We often write $\Delta$ as 
$\prod_{i \in \I} \pp_i {:} \T_i$ as shorthand for 
$\{\pp_1 {:} \T_1, \dots, \pp_n {:} \T_n\}$ where 
$\prod_{i \in \emptyset} \pp_i {:} \T_i = \emptyset$. 

\begin{definition}[Size of a typing context]
  Define $|\pp {:} \T| = |\T|$ and $|\Delta, \Delta'| = |\Delta| + |\Delta'|$.
\end{definition}

The typing system for process $\PP$ does not change.
For typing a multiparty session, we use \RULE{B-Sess} instead of 
\RULE{T-Sess} where $\safe(\Delta)$ means typing context $\Delta$ satisfies 
a \emph{safety} property. To define this property, we first recall 
the transition relation of typing contexts 
from~\cite[Definition 2.8]{Scalas2019}. 
The transition relation is similar to the LTS of the type graph 
defined in Definition~\ref{def:typegraph}. 

\begin{definition}[Labelled transition relation of typing contexts]\label{def:composition_reduction}
  Define labels
  ($\cxell,\cxell'\cdots$) by 
$\cxell::= \ \trin{\pp\pq}{\S}
\SEP \trout{\pp\pq}{\S} \SEP \trbra{\pp\pq}{l} \SEP
\trsel{\pp\pq}{l} \SEP  \SelBra{\pp}{\pq}{l} \SEP \OutIn{\pp}{\pq}$.
Then the \emph{labelled transition relation} (LTS) over typing contests, 
$\Delta\trans{\cxell}\Delta'$, is defined by:

  \begin{center}
  \small
      \begin{spacing}{0}
      \newcommand{\treesep}{.5em}%
      \begin{prooftree}
      \hypo{\dagger\in \set{!,?}}
      \infer1[\RULE{E-IO}]{\pp : \trdag\pq{\S};\T \redDagger{\pp}{\pq}{\S} \pp : \T}
      \end{prooftree}\hspace*{\treesep}
      \begin{prooftree}
     \hypo{\dagger\in \set{\oplus,\&} \quad k \in \I}
      \infer1[\RULE{E-BS}]{\pp : \trbradag\pq{\set{l_i: \T_i}_{i \in \I}} 
\redDagger{\pp}{\pq}{l_k} \pp : \T_k}
    \end{prooftree}\hspace*{\treesep}
      \begin{prooftree}
        \hypo{\Delta \trans{\cxell} \Delta'}
        \infer1[\RULE{Comp}]{\Delta, \pp : \T \trans{\cxell} \Delta', \pp : \T}
      \end{prooftree}

  \vspace{2mm}

      \begin{prooftree}
        \hypo{\Delta_1 \redOut{\pp}{\pq}{\S} \Delta'_1}
        \hypo{\Delta_2 \redIn{\pq}{\pp}{\S} \Delta'_2}
        \infer2[\RULE{Msg}] {\Delta_1, \Delta_2 \redOutIn{\pp}{\pq} \Delta'_1, \Delta'_2}
      \end{prooftree}\hspace*{\treesep}
      \begin{prooftree}
        \hypo{\Delta_1 \redSel{\pp}{\pq}{l} \Delta'_1}
        \hypo{\Delta_2 \redBra{\pq}{\pp}{l} \Delta'_2}
        \infer2[\RULE{Bra}] {\Delta_1, \Delta_2 \redSelBra{\pp}{\pq}{l} \Delta'_1, \Delta'_2}
      \end{prooftree}\hspace*{\treesep}
      \begin{prooftree}
        \hypo{\pp : \T[\mu \ty. \T/\ty] \trans{\cxell} \Delta'}
        \infer1[\RULE{$\mu$}]{\pp : \mu \ty. \T \trans{\cxell} \Delta'}
      \end{prooftree}
      \end{spacing}
  \end{center}
Define $\Delta \rightarrow \Delta'$ if $\Delta \redOutIn{\pp}{\pq} \Delta'$ or $\Delta \redSelBra{\pp}{\pq}{l} \Delta'$
and $\Delta \rightarrow$ (resp. $\Delta \trans{\cxell}$) 
for 
$\exists \Delta'. \Delta\rightarrow \Delta'$ (resp.~$\Delta \trans{\cxell}\Delta'$). 
\end{definition}

\begin{definition}[Safety, Definition 4.1 in \cite{Scalas2019}]
\label{def:safety}
\label{def:safe_states}
\begin{itemize}[leftmargin=*]
\item 
$\Delta$ is a \emph{safe state} if 
(1) 
$\Delta \redOut{\pp}{\pq}{\S}$ and $\Delta \redIn{\pq}{\pp}{\S'}$ imply $\Delta \redOutIn{\pp}{\pq}$; and 
(2)
  $\Delta \redSel{\pp}{\pq}{l}$ and $\Delta \redBra{\pq}{\pp}{l'}$ imply $\Delta \redSelBra{\pp}{\pq}{l}$.
\item 
$\phi$ is a \emph{safety property} if, for all $\Delta$ such that
$\phi(\Delta)$, we have: 
(1) 
$\Delta$ is a safe state; 
(2) 
$\Delta = \Delta', \pp : \mu\ty.\T$ implies $\phi(\Delta', \pp : \T[\mu\ty.\T/\ty])$; and 
(3) 
$\Delta \rightarrow \Delta'$ implies $\phi(\Delta')$.
\end{itemize}
\noindent We call 
$\Delta$ \emph{safe} if $\phi(\Delta)$ holds for some safety property $\phi$ and write $\safe(\Delta)$ if $\Delta$ is safe. 
\end{definition}

\begin{remark}\label{rem:asymmetric_safe_live}
In Definition~\ref{def:safety}, 
the transition for selection and branching (2) is not symmetric; 
in a safe state, all labels of 
a selection $l$ must be reducible, while \emph{some} 
label of a branching needs to be reducible. 
See \cite{Scalas2019} for more explanations.
\end{remark}
\begin{restatable}{lemma}{safetyfromreductions}\label{thm:safety_from_reductions}\proofreference{Appendix~\ref{app:proof_of_safety_from_reductions}}
   $\Delta$ is safe if and only if $\Delta'$ is a safe state for all $\Delta \rightarrow^* \Delta'$.
\end{restatable}

\begin{theorem}[Communication safety, Corollary~4.7 in \cite{Scalas2019}] 
Suppose $\provesbot \M \tri \Delta$ by \RULE{B-Sess} in 
Figure~\ref{fig:processtyping}. 
Then $\M$ is communication safe. 
\end{theorem}
Unfortunately, under the bottom-up system, 
typable $\M$ does not guarantee neither deadlock-freedom nor 
liveness properties. To guarantee those properties for $\M$, 
one needs to \emph{additionally} check $\Delta$ satisfies  
the corresponding property $\varphi(\Delta)$. 
Below we recall the definitions from \cite{Scalas2019,BSYZ2022,YH2024}.

\begin{definition}[Deadlock-freedom, Figure 5(2) in \cite{Scalas2019}]\label{def:deadlock_freedom}
  $\Delta$ is deadlock-free if $\Delta \rightarrow^* \Delta' \not\rightarrow$ implies $\unfold{\T_\pp} = \tend$ for all $\pp:\T_\pp \in \Delta'$.
\end{definition}
\begin{definition}[Fair and Live Paths, Definition~17 in \cite{BSYZ2022}]
\label{def:fair_live_paths}
A \emph{path} is a sequence of contexts $\mathcal{P} ::=
(\Delta_i)_{i \in N}$, where $N = \{0, 1, 2, \dots\}$ is a (finite or
infinite) set of consecutive natural numbers, and
$\Delta_i \rightarrow \Delta_{i+1}$ for all $i, i+1 \in N$.  Then a path
is \emph{fair} if, for all $i \in N$, $\Delta_i \redOutIn{\pp}{\pq}$
implies $\exists j \geq
i.\:\Delta_j \redOutIn{\pp}{\pq} \Delta_{j+1}$, and
$\Delta_i \redSelBra{\pp}{\pq}{l}$ implies $\exists l', j$ such that
$j \geq i$ and $\Delta_j \redSelBra{\pp}{\pq}{l'} \Delta_{j+1}$.

A path is \emph{live} if, for all $i \in N$:
{\rm (1)} $\Delta_i \redOut{\pp}{\pq}{S}$ implies $\exists j \geq i.\:\Delta_j \redOutIn{\pp}{\pq} \Delta_{j+1}$; 
{\rm (2)} $\Delta_i \redIn{\pp}{\pq}{S}$ implies $\exists j \geq i.\:\Delta_j \redOutIn{\pq}{\pp} \Delta_{j+1}$; 
{\rm (3)} $\Delta_i \redSel{\pp}{\pq}{l}$ implies $\exists l', j$ such that $j \geq i$ and $\Delta_j \redSelBra{\pp}{\pq}{l'} \Delta_{j+1}$; and 
{\rm (4)} $\Delta_i \redBra{\pp}{\pq}{l}$ implies $\exists l', j$ such that $j \geq i$ and $\Delta_j \redSelBra{\pq}{\pp}{l'} \Delta_{j+1}$.
\end{definition}
\begin{definition}[Liveness]\label{def:liveness}
$\Delta$ is \emph{live}, denoted $\live(\Delta)$, 
if $\Delta \rightarrow^* \Delta'$ implies that all paths starting with $\Delta'$ that are fair are also live.
\end{definition}

\begin{remark}[Liveness]\label{rem:liveness}
We require fair paths to define liveness of typing contexts 
(Definition~\ref{def:liveness}), while liveness of session $\M$ 
(Definition~\ref{def:properties}) 
does not require the fairness assumption. 
This is because a typing context without the fair path requirement
can type \emph{non-live sessions}. See \cite[Example~5.14]{Scalas2019} and \cite[Example 9]{YH2024}. 
More explanations with examples  
can be found in
\iffull{\ Appendix~\ref{app:rem:liveness}}\else{\ \cite{UY2024}}\fi.
\end{remark}

\begin{theorem}[Deadlock-freedom and liveness, Theorem~5.15 in \cite{Scalas2019}]
Suppose $\provesbot \M \tri \Delta$ with $\varphi \in \{\df, \live\}$, 
and $\varphi(\Delta)$. Then $\M$ satisfies $\varphi$. 
\end{theorem}




Note that for all $\Delta$, $\live(\Delta)$ implies 
$\df(\Delta)$ but 
$\live(\Delta)$ and $\safe(\Delta)$ are incomparable, and  
$\df(\Delta)$ and $\safe(\Delta)$ are incomparable~\cite[Theorem~4]{YH2024}.  
See Figure~\ref{fig:venn} with appropriate examples. 

\paragraph{\bf Complexity of bottom-up system}
The complexity of checking whether $\M$ satisfies 
$\varphi$ (Figure~\ref{fig:bottom_up})
is reduced to (i) the minimum type inference (given 
in Theorem~\ref{thm:typeinfer:complexity})  
and (ii) checking whether inferred $\Delta$ satisfies $\varphi$.
For the rest of this section, we investigate the complexity of (ii). 

\subsection{Checking Safety, Deadlock-Freedom and Liveness is in \PSPACE}
\label{subsec:pspacecomplete}
Unfortunately, given typing context $\Delta$, 
checking the three aforementioned properties is \emph{\PSPACE-complete}. 
We start proving that checking these properties 
are in \PSPACE. 
First, we bound the number of reachable contexts $\Delta'$ such that $\Delta \rightarrow^* \Delta'$ by the product of the sizes of the types.

\begin{restatable}{lemma}{numberofreachablecompositionstates}\label{thm:number_of_reachable_composition_states} \proofreference{Appendix~\ref{app:pspace_completeness_of_safety_and_deadlock_freedom}}
  Let $\Delta = \prod_{i \in \I} \pp_i : \T_i$. Then $|\{\Delta' \mid \Delta \rightarrow^* \Delta'\}| \leq \prod_{i \in \I}\size{\T_i} = 2^{\bigO{|\Delta|}}$ where 
  $\prod_{i \in \I}n_i$ denotes the product of $n_i$.
\end{restatable}
\noindent Thus the number of reachable contexts is at most exponential in the
size of the input ($\sum_{i \in \I} |\T_i|$). This gives rise to a
nondeterministic polynomial-space algorithm for checking the complement
of safety, where all states at a distance up to this exponential bound
are nondeterministically visited and checked for states that are not
safe. 

\begin{theorem}\label{thm:safety_liveness_pspace}
\proofreference{Appendix~\ref{app:pspace_completeness_of_safety_and_deadlock_freedom}}  
Checking for safety and deadlock-freedom is in \PSPACE.
\end{theorem}
\begin{proof}
We construct a nondeterministic polynomial-space algorithm that accepts if and only if the input is not safe.
Given an initial context $\Delta$, the algorithm nondeterministically traverses a path of $\prod_{i \in \I} |\T_i|$ transitions or fewer,
finding a state $\Delta'$ such that $\Delta \rightarrow^* \Delta'$. Then, it accepts iff $\Delta'$ is not a safe state.
Thus, there exists some accepting path iff there exists a reachable unsafe state by Theorem~\ref{thm:number_of_reachable_composition_states}.
Note that the path does not need to be explicitly stored; rather, only the current context needs to be stored (which can be done in polynomial space).

We use the fact that $\coNPSPACE = \PSPACE$. Deadlock-freedom is similarly proved by checking a deadlocked state. 
\iffull{\ The algorithm (Algorithm~\ref{alg:pspace_safety}) is found in
  Appendix~\ref{app:pspace_completeness_of_safety_and_deadlock_freedom}}\else{\ The
  algorithm is found in \cite{UY2024}}\fi.
\end{proof}


The proof that checking for liveness is in $\PSPACE$ is more involved:
first, we define an equivalent characterisation of liveness. Live paths are characterised by \emph{barbs}, which are the offered transitions in some context, and \emph{observations}, which are the transitions that are matched in some reduction. 

\begin{figure}
{
   \centering
   \small
    \begin{spacing}{3}
    \newcommand{\treesep}{.5em}
    \begin{prooftree}
      \hypo {\Delta \redDagger{\pp}{\pq}{[\S]} \Delta'
            \quad \dagger\in \set{!,?}}
      \infer1[\RULE{Barb-IO}]{\barbDagger{\pp}{\pq} \in \barbs(\Delta)}
    \end{prooftree}\hspace{\treesep}%
    \begin{prooftree}
      \hypo {\Delta \redDagger{\pp}{\pq}{l} \Delta' 
             \quad \dagger\in \set{\oplus,\&}}
      \infer1[\RULE{Barb-BS}] {\barbBra{\pp}{\pq} \in \barbs(\Delta)}
    \end{prooftree}\hspace{\treesep}%

    \begin{prooftree}
      \hypo {\Delta \redOutIn{\pp}{\pq} \Delta'}
      \infer1[\RULE{Obs-O}] {\barbOut{\pp}{\pq} \in \observations(\Delta, \Delta')}
    \end{prooftree}
\begin{prooftree}
      \hypo {\Delta \redOutIn{\pp}{\pq} \Delta'}
      \infer1[\RULE{Obs-I}] {\barbIn{\pq}{\pp} \in \observations(\Delta, \Delta')}
    \end{prooftree}\hspace{\treesep}%
    \begin{prooftree}
      \hypo {\Delta \redSelBra{\pp}{\pq}{l} \Delta'}
      \infer1[\RULE{Obs-B}] {\barbBra{\pq}{\pp} \in \observations(\Delta, \Delta')}
    \end{prooftree}\hspace{\treesep}%
    \begin{prooftree}
      \hypo {\Delta \redSelBra{\pp}{\pq}{l} \Delta'}
      \infer1[\RULE{Obs-S}] {\barbSel{\pp}{\pq} \in \observations(\Delta, \Delta')}
    \end{prooftree}
    \end{spacing}
  }
\vspace{-1cm}
\caption{Barbs $\barbs(\Delta)$ and observations $\observations(\Delta,\Delta')$.}
\label{fig:barbs_obs}
\end{figure}

\begin{definition}[Barbs and observations]\label{def:barbs_obs}
Recall $\cxell$ defined in Definition~\ref{def:composition_reduction} 
and abbreviate 
$\barbOut{\pp}{\pq}$ etc. if the payload is unimportant. 
The sets of \emph{barbs} of $\Delta$ and \emph{observations} of $\Delta,\Delta'$ are defined 
in Figure~\ref{fig:barbs_obs}.
For paths $\mathcal{P} = (\Delta_i)_{i \in N}$, 
define $\barbs(\mathcal{P}) = \bigcup_{i \in N} \barbs(\Delta_i)$ 
and $\observations(\mathcal{P}) = \bigcup_{i, i+1 \in N} \observations(\Delta_i, \Delta_{i+1})$.
\end{definition}
The set of barbs contains all actions offered by some type in $\Delta$, and the set of observations contains all barbs that become satisfied in the definition of liveness from some reduction (Definition \ref{def:liveness}).
Note that $\barbs$ and $\observations$ are checkable in linear time.
%
%
Then we define \emph{counterwitness}, which is a path that is fair but not live. 
\begin{definition}\label{def:counterwitness}
  A path $(\Delta_i)_{i \in N}$ is a counterwitness of $\Delta$ if $\Delta = \Delta_0$ and $\Delta_i \trans{\cxell_i} \Delta_{i+1}$ and:
  \begin{itemize}
    \item There exists $k \in N$ and $a \in \barbs(\Delta_k)$ such
      that $a \not\in \observations((\Delta_i)_{i \geq k})$; and 
    \item For all $k \in N$, $\{\cxell_i \mid i \geq k\} = \{\cxell \mid \exists i \geq k.\:\Delta_i \trans{\cxell}, \cxell \in \{\OutIn{\pp}{\pq}, \SelBra{\pp}{\pq}{l}\}\}$.
  \end{itemize}
\end{definition}


\iffull{The proofs in the following lemmas are found in 
Appendix~\ref{app:pspace_completness_of_liveness}.}\else{}\fi
\begin{restatable}{lemma}{counterwitnessiffnotlive}\label{thm:counterwitness_iff_not_live}
  A counterwitness exists for $\Delta$ if and only if $\Delta$ is not live.
\end{restatable}

Counterwitnesses are paths which have a barb that is never observed.

\begin{lemma}\label{lem:live_iff_barbs_subset_obs} A path $(\Delta_i)_{i \in N}$ is live iff $\barbs(\Delta_k) \subseteq \observations((\Delta_i)_{i \geq k})$ for all $k$. \label{thm:barbs_obs_live_path}
\end{lemma}
\begin{restatable}{lemma}{thmnumberofgammatransitions}\label{thm:number_of_gamma_transitions}
  Let $\Delta = \prod_{i \in \I} \pp_i : \T_i$. Let $n = \sum_{i \in \I} |\T_i|$. Then there are at most $2n$ possible reductions $\cxell \in \{\OutIn{\pp}{\pq}, \SelBra{\pp}{\pq}{l}\}$ such that $\Delta \rightarrow^* \trans{\cxell}$.
\end{restatable}
Then, we give a sufficient bound on the length and periodicity of a counterwitness.
\begin{restatable}{lemma}{witnesslengthforliveness}\label{thm:witness_length_for_liveness}\proofreference{Appendix~\ref{app:pspace_completness_of_liveness}} 
  Let $\Delta = \prod_{i \in \I} \pp_i : \T_i$. Let $m = (2n+2) \cdot \prod_{i \in \I}|\T_i|$.
  If $\Delta$ is not live, then it has a finite counterwitness of size $k$ 
such that 
$k\leq m$ or an infinite counterwitness of the form $\mathcal{P}_1\mathcal{P}_2^\omega$ with $\size{\mathcal{P}_1}\leq m$ and  
$\size{\mathcal{P}_2}\leq m$. 
\end{restatable}
\begin{restatable}{theorem}{livenesspspace}\label{thm:liveness_pspace}
\proofreference{Appendix~\ref{app:pspace_completness_of_liveness}}
  Checking for liveness is in \PSPACE.
\end{restatable}
\begin{proof}
Using the above bound on the length of a counterwitness, we can give a nondeterministic algorithm that checks for the complement of liveness by checking all paths in the form specified by Lemma~\ref{thm:witness_length_for_liveness}.
We construct a nondeterministic polynomial-space 
algorithm that accepts if and only if the input is not live.
The algorithm nondeterministically chooses a candidate witness of the form $\mathcal{P}_1 \mathcal{P}_2^\omega$, where $\size{\mathcal{P}_1}\leq m$ 
and $\size{\mathcal{P}_2}\leq m$ such that $m= (2n+2) \cdot \prod_{i \in \I}|\T_i|$.
Along the path, some $k \leq 2m$ is nondeterministically chosen and the conditions in Definition~\ref{def:counterwitness} are checked.
A similar check is performed for finite witnesses of the form $\mathcal{P}_1$ of length $j$ such that $j\leq m$.
 Along the path, the algorithm involves storing a finite number of sets of barbs, and there are only polynomially many barbs. Therefore, the algorithm is polynomial-space.
\iffull{\ The algorithm (Algorithm~\ref{alg:pspace_liveness}) is 
  found in Appendix~\ref{app:pspace_completness_of_liveness}}\else{\ See
  \cite{UY2024} for the algorithm}\fi.
\end{proof}

\begin{remark}\label{rem:graph_search}
In practice, the number of reachable states is often much smaller than
the bound given in
Lemma~\ref{thm:number_of_reachable_composition_states}.
Hence even if the space complexity is not polynomial, 
we can build a deterministic algorithm that might be more
practical with an explicit breadth-first search over all reachable
states. 
\begin{restatable}{theorem}{safedfcomplexityefficient}
\label{thm:safedf_complexity}
 Checking for safety (or deadlock-freedom) of a typing context $\Delta = \prod_{i \in \I} \pp_i : \T_i$ can be done in $\bigO{|\Delta| \cdot \U}$ time
 where $\U = \prod_{i \in \I} |\T_i|$.
\end{restatable}
We can also use a similar construction to the proof of Savitch's theorem
\cite[Theorem 8.5]{Sipser2012} 
to create a deterministic algorithm for checking liveness.
\begin{restatable}{theorem}{livecomplexityefficient}
\proofreference{Appendix~\ref{subsec:live:efficient}}
\label{thm:live_complexity}
Let $\Delta = \prod_{i \in \I} \pp_i : \T_i$ and $\U = \prod_{i \in \I} |\T_i|$.
Then there exists a deterministic algorithm that checks if $\Delta$ is live in $\bigO{\U^3 \cdot 2^{3|\Delta|} + \U^3 \log \U \cdot 2^{2|\Delta|}}$ time.
\end{restatable}
\end{remark}

\begin{remark}
If we restrict the number of participants to some constant $c$, then
Lemma~\ref{thm:number_of_reachable_composition_states} shows that the number of reachable configurations is $\bigO{\prod_{i \in \I}\size{\T_i}} = \bigO{|\Delta|^c}$.
This implies that the algorithms described in Remark~\ref{rem:graph_search} are polynomial-time for a fixed number of participants,
and suggests that they are feasible
for larger typing contexts with a small number of participants.
\end{remark}

\subsection{PSPACE-Hardness of Checking Safety, Deadlock-Freedom and Liveness} 
\label{subsec:pspacehard}
We turn our attention to \PSPACE-hardness. 
We reduce it from the \emph{quantified Boolean formula (QBF) 
problem} which is known to be \PSPACE-complete~\cite[Theorem 8.9]{Sipser2012}.
Given a formula in QBF $\bqbf = \cQ_1 x_1 \dots \cQ_n x_n \sqbf,\  \cQ_i \in \{\exists, \forall\}$, with $\sqbf$ in conjunctive normal form: $\sqbf := \bigwedge_{i=1}^{m} C_i,\ C_i = (L_{i1} \vee L_{i2} \vee L_{i3})$ where $L$ is a literal which is either a variable $x_i$ or its negation $\neg x_i$, 
the QBF problem asks whether $\bqbf$ is true.  
The full construction for the reduction from the QBF problem is very large and non-trivial:  
we leave the detailed definitions and proofs in
\iffull{\ Appendix~\ref{app:reduction_from_qbf}}\else{\cite{UY2024}}\fi, 
but give a summary. 

Our construction is a typing context $\Deltainit$ that computes the truth value of $\bqbf$. It contains one participant for each variable ($\pp_i$) and clause ($\pr_i$), as well as the ``controller'' ($\ps$). Participant $\pp_i$ finds the truth value of the formula $\cQ_i x_i \dots \cQ_n x_n \sqbf$, and $\pr_i$ finds the truth value of the clauses $C_i$ onwards. The controller $\ps$ queries $\pp_1$ and enters an undesirable state $\Tbad$ if $\bqbf$ is false, otherwise the process repeats indefinitely.
We construct $\Deltainit$ such that its reduction paths are \emph{deterministic} (i.e., $\Delta \rightarrow \Delta_1$ and $\Delta \rightarrow \Delta_2$ 
imply $\Delta_1 = \Delta_2$) and \emph{safe}  
(i.e., $\Delta \rightarrow \Delta'$ is safe if $\Delta$ is a safe state). 
This enables that all possible behaviours of the context are 
analysed by the reductions from $\Deltainit$, and those 
reductions must not contain any violations of a property.

As safety and deadlock-freedom are reachability properties, we only need to exhibit a $\Tbad$ that violates the properties from $\Deltainit$.    
We prove that, if $\bqbf$ is true, then $\Deltainit$ reduces safely
and deterministically back to itself by the multistep reductions.
Otherwise $\Deltainit$ reduces to a state where $\ps$ has
the undesirable type $\Tbad=\tbra{\pp_1}{\kf{unsafe}: \tend}$, 
which represents an unsafe state. 
For deadlock-freedom, we set $\Tbad=\tend$ which creates a stuck state.  

For the liveness, we first identify an equivalent formulation 
of the live path (Definition~\ref{def:fair_live_paths}).
We then prove that if $\bqbf$ is satisfied, a deterministic infinite 
live path reduces to the same state; 
otherwise it leads to $\Tbad=\tend$ showing the path is not live. 
We can use the same bad type as deadlock-freedom since 
if $\Delta$ is live, it is deadlock-free.

\begin{theorem}\label{thm:safety_df_liveness_pspace_hard}
\proofreference{Appendices~\ref{sec:safety_pspace_hard}, \ref{sec:deadlock_freedom_pspace_hard}, \ref{sec:liveness_pspace_hard}.}
  Checking for safety, deadlock-freedom and liveness is \PSPACE-hard.
\end{theorem}
\begin{corollary}
Checking for safety, deadlock-freedom and liveness is $\PSPACE$-complete.
\end{corollary}

 \section{Summary of Complexity Results and Related Work}
\label{sec:related}
\begin{table}[t]
  \centering
\small
  \begin{tabular}{|c|c|c|c|c|c|c|}
    \hline
    \multicolumn{4}{|c|}{\textbf{Binary}} & \multicolumn{3}{c|}{\textbf{Multiparty}} \\
    \hline
    Algorithm & \cite{Lange2016} & \multicolumn{2}{|c|}{\cite{Udomsrirungruang2024a}} & Algorithm & \multicolumn{2}{|c|}{\texthighlight{This paper}} \\
    \hline
    Inductive \cite{Gay2005} & $\bigO{n^{2^n}}$ & $\bigO{n^{n^3}}$ &
    $\bigOmega{(\sqrt{n})!}$ & Inductive \cite{Ghilezan2019} &
    \texthighlight{$\bigO{n^{n^3}}$}
    & \texthighlight{$\bigOmega{(\sqrt{n})!}$} \\
    \hline
    Coinductive \cite{Udomsrirungruang2024a} & $\bigO{n^2}$ & \multicolumn{2}{|c|}{$\bigO{n^2}$} & \texthighlight{Coinductive} & \texthighlight{$\bigO{n^2}$} & \texthighlight{$\bigOmega{n^2}$} \\
    \hline
  \end{tabular}
  \vspace{1mm}%
  \caption{Summary of the complexity results for subtyping ($n=\size{\T}$).}
  \label{tab:subtyping_complexity_summary}
\vspace{-5mm}
\end{table}

\begin{table}[t]
\small
\begin{tabular}{|l|l||l|c|l|}
\hline
\textbf{Projection} & Defs. & \highlight{\textbf{Complexity}} & \textbf{Output} & \highlight{\textbf{Top-down total complexity}} \\
    \hline
    Inductive plain & 
\ref{def:inductive_projection},~\ref{def:inductive_plain_merging}
&
\texthighlight{$\bigTheta{\size{\G} \log{\size{\G}}}$}
& $\size{\G}$ & \texthighlight{$\bigO{2^\size{\M} \cdot (\size{\M} + \size{\G}) + \size{\G} \log{\size{\G}}}$} \\
    Inductive full & 
\ref{def:inductive_projection},~\ref{def:inductive_full_merging}
& \texthighlight{$\bigTheta{\size{\G} \log^2{\size{\G}}}$} &
$\size{\G}$ & \texthighlight{$\bigO{2^\size{\M} \cdot (\size{\M} + \size{\G}) + \size{\G} \log^2{\size{\G}}}$} \\
    \TBC~\cite{Tirore2023} (\Sec\ref{section:coinductive_projection}) & 
\ref{def:coinductiveproj},~\ref{def:coinductive_plain_merge}
& \texthighlight{$\bigTheta{\size{\G}^2}$} &
$\size{\G}$ &
\texthighlight{$\bigO{2^\size{\M} \cdot (\size{\M} + \size{\G}) + \size{\G}^2}$} \\
    \texthighlight{Coinductive full} & 
\ref{def:coinductiveproj},~\ref{def:coinductive_full_merge}
& \texthighlight{$\bigO{\size{\G} \cdot 2^\size{\G}}$},
\texthighlight{$\bigOmega{2^{|\G|^{\left(\frac{1}{2} - \varepsilon\right)}}}$} & $2^\size{\G}$ & \texthighlight{$\bigO{\max(\size{\M}, \size{\G}) \cdot 2^{\max(\size{\M}, \size{\G})}}$} \\
    \hline
  \end{tabular}
\caption{Complexity of the four projections. The output size is used as the input size for checking $\Delta\sqsubseteq \G$. \label{table:projection}}
\vspace{-5mm}
\end{table}

\begin{table}[t]
\small
\begin{tabular}{|l|l||l|l|}
  \hline
\textbf{Property} & Definitions & \highlight{\textbf{Complexity}} & \highlight{\textbf{Bottom-up total complexity}} \\
  \hline
  Safety, deadlock-freedom & Def.~\ref{def:safety}, \ref{def:deadlock_freedom} & \texthighlight{$\bigO{|\Delta| \cdot \U}$}
  & \texthighlight{$\bigO{2^{2\size{\M}}}$} \\
  \hline
  Liveness & Def.~\ref{def:liveness} & \texthighlight{$\bigO{\U^3 \cdot 2^{3|\Delta|} + \U^3 \log \U \cdot
    2^{2|\Delta|}}$} 
  & \texthighlight{$\bigO{\size{M} \cdot 2^{6\size{\M}}}$} \\
  \hline
\end{tabular}
\caption{Complexity of the bottom-up approach ($\U = \prod_{i
    \in \I} |\Tmin_i|$) \label{tab:bot_complexity}}
\vspace{-5mm}
\end{table}

\myparagraph{Summary.\ }
We summarise the complexity results 
\highlight{highlighting} the algorithms and results we have constructed or proven
in this paper. 
Table~\ref{tab:subtyping_complexity_summary} lists the complexity of
the subtyping checking, including the complexity results from
\cite{Lange2016,Udomsrirungruang2024a} (which will 
be discussed in \textbf{Subtyping} in the related work). 
Figure~\ref{fig:venn} lists 
the relationship between the sets of local types 
obtained by the projections and safety, deadlock-freedom and liveness 
properties, and 
Table~\ref{table:projection} lists the complexity of the top-down approach for each projection. 
In practice, 
Figure~\ref{fig:venn} and Table~\ref{table:projection} together give 
us a stand-alone guidance of tractability in relation to 
expressiveness 
(i.e., a set of global types which the algorithm can project) since
local types are directly usable for implementing applications,
such as API generations and runtime-monitoring \cite{Yoshida24,BETTYTOOLBOOK}.
The inductive projections are lower complexity than the corresponding 
coinductive projections. 
The \TBC~algorithm is more efficient (quadratic) than 
the coinductive projection with full merging (in \PSPACE). 
However, \TBC\ \emph{completely} rejects 
inductively projectable global types with full merging which are not inductively 
projectable by plain merging (Proposition~\ref{pro:ifcp}). 

The full complexity of type-checking with global type $\G$ is given below.  

\begin{restatable}[Top-down full complexity]{theorem}{topdowncomplexity}
\label{thm:fullcomplexity}
Given $\G$ and $\M$, checking whether 
$\provestop \M:\Delta$ such that $\Delta\associated \G$ is 
exponential in $\size{\G}$ and $\size{\M}$.
\end{restatable}
\begin{proof}
By Theorem~\ref{thm:typeinfer:complexity}, 
inferring the minimum type from $\M$ takes time $\size{\M} \cdot 2^\size{\M}$
and generates a minimum type graph of size $2^\size{\M}$. The complexities 
of all projections in \Sec\ref{sec:projections} are given in 
Table~\ref{table:projection}.
In the case of the coinductive projection 
with full merging, we check $\G$ is balanced, which costs  
polynomial. Hence it does not 
affect to the total complexity.  
By summing up, the total complexity of type checking by all projection methods 
are exponential in $\size{\G} + \size{\M}$.
\end{proof}

\begin{restatable}[Bottom-up full complexity]{theorem}{bottomupcomplexity}
\label{thm:bottomupfullcomplexity}
Given $\M$, checking whether there exits 
$\provesbot \M:\Delta$ where $\Delta$ is safe, 
deadlock or live has complexity exponential in $\size{M}$. 
\end{restatable}
\begin{proof}
By Theorem~\ref{thm:typeinfer:complexity}, inferring the minimum type of $\PP$ 
takes time $\size{\PP} \cdot 2^\size{\PP}$ 
and generates a minimum type graph of size $2^\size{\PP}$.
Table~\ref{tab:bot_complexity} summarises the complexity of checking safety,
deadlock-freedom and liveness of the inferred typing context, respectively.
For a multiparty session $\M = \prod_{i \in \I} \PP_i$, 
we have $\U = \prod_{i \in \I} |\Tmin_i| \leq 2^{\size{\sum_{i \in
      \I}\PP_i}} = \bigO{2^\size{\M}}$. 
Furthermore, $\size{\Delta} = \bigO{2^{\size{\M}}}$.
Although the type inference and property-checking are individually exponential,
the total complexity in Table~\ref{tab:bot_complexity} is not doubly exponential, as the product of the minimum types $\U$ is exponential in $\size{\M}$.
\end{proof}

\noindent While both approaches are exponential,
the bottom-up requires a more complex approach
during property-checking, and thus the algorithms
have a higher complexity. Among them, liveness is the highest. 
In practice, it is unlikely that type inference yields types of
exponential size. 
Ignoring the inference step,
but focusing on the \emph{type-level},
the top-down inductive plain and full merging and coinductive plain
merging approaches 
are only \emph{polynomial} in the size of the local types, while
the top-down coinductive full merging and 
the bottom-up approaches 
are \emph{exponential}.
These suggest that the top-down approaches with polynomial complexity
are more efficient, albeit limited.

The top-down takes time linear in the size of the type graph of the projection. In practice,
the size of this type graph is much smaller than the theoretical bound of $2^{\size{\G}}$.
Even though examples exist where the projection has an exponential size (Theorem~\ref{thm:coinductive:size}),
these examples remain contrived and unnatural.
In the bottom-up approach, the complexity is
linear in the number of reachable typing contexts.
In practice,
this number can easily become exponential in the size of the session,
for example by composing many copies of independently interacting sessions.
This is the main source of the high complexity of the bottom-up approach.
For processes with a large number of participants and nondeterministic behaviour,
the bottom-up may be more expensive in practice.

\subsection*{Related and Future Work}
\label{subsec:related}
\myparagraph{Subtyping.}
Gay and Hole~\cite{Gay1999} introduced
the binary synchronous session subtyping system and proposed 
a sound and complete inductive subtyping algorithm in~\cite{Gay2005}.
Lange and Yoshida~\cite{Lange2016} reduce subtyping of synchronous binary session types to checking a modal $\mu$-calculus formula, and benchmark 
their algorithm against \cite{Gay2005} and a coinductive algorithm.
Ghilezan \etal~\cite{Ghilezan2019} prove operational and
denotational \emph{preciseness} of synchronous multiparty session
subtyping
for the synchronous multiparty session calculus.  
An inductive subtyping algorithm is given in~\cite{Ghilezan2019}, but
they did not measure its complexity. 
In this paper, we
have proven exponential in Theorem~\ref{thm:complexity_ghilezan_subtyping}.
We proposed a subtyping algorithm for the multiparty setting based on 
a simulation extending the binary algorithm given in 
\cite[Theorem~4.6]{Udomsrirungruang2024a}.
We proved the tight bound of our algorithm is
quadratic in Theorem~\ref{thm:quadratic_subtyping}.
Similar to type simulations (Definition~\ref{def:type_simulation}),
Silva \etal\ \cite{Silva2023} introduce the
$\mathit{XYZW}$-simulation for formulating subtyping of context free
session types inspired by 
$\mathit{XY}$-simulation by Aarts and Vaandrager \cite{Aarts2010}.
A semi-decision procedure for this problem is introduced along with an empirical analysis.
Padovani~\cite{Padovani2019} shows that subtyping for context-free session types is undecidable.

In asynchronous semantics, Ghilezhen \etal~\cite{GPPSY2023} proposed
the \emph{precise multiparty asynchronous subtyping} ($\subt_{a}$)
which is operationally and denotationally sound and complete 
with respect to the liveness property. The synchronous subtyping
relation $\subt$ studied in this paper 
is a strict subset of $\subt_{a}$.
Checking $\subt_{a}$ is known to be
\emph{undecidable} even in the binary case \cite{LY2017}. 
Li \etal~\cite{Li2024} studied the problem of subtyping in
communicating finite state machines (CFSM) \cite{BZ1983}
under \emph{subprotocol fidelity},
where the CFSMs are refined 
with respect to a context of a global type $\G$.
They showed $\PSPACE$-hardness to decide
\emph{monolithic refinement} (whether a set of CFSMs are valid
refinement of a given global type), and
further proposed more tractable (polynomial) \emph{protocol refinement}
where only one local type is changed at a time.
Cutner~\etal~\cite{CYV2022} proposed a sound but not complete 
approximation of $\subt_a$ 
which uses the bounded rewriting technique 
inspired by the IO-tree decomposition theorem in \cite{GPPSY2023}. 
Their relation is integrated into
the top-down approach, and is applied to asynchronous message optimisation
in Rust. 

\myparagraph{End-Point Projections.}
The end-point projection for session types
was first introduced by Honda \etal~\cite{Honda2008},
using the inductive plain merge
(Definition~\ref{def:inductive_plain_merging}), 
with linearity conditions on channels 
to check the well-formedness of global types. 
Later, Bettini \etal~\cite{BettiniCDLDY08,CDYP13} proposed simplified global types
with plain merging, adopted by most works, including ours. 

The inductive full merge
(Definition~\ref{def:inductive_full_merging})
is introduced for enlarging projectable sets of global types, see
\cite[\Sec 3.1,\Sec 8.1]{Scalas2019} for a survey. 
A recent paper \cite{YH2024} proves the subject reduction theorem, 
communication safety, 
deadlock-freedom and liveness 
under the top-down approach which uses 
the inductive full merge (Definition~\ref{def:inductive_full_merging}). 
The typing system incorporates the subtyping relation with the projection 
($\Delta\associated\G$ in \RULE{Sess-T} in Figure~\ref{fig:processtyping}). 
This work has resolved misunderstanding by researchers
(caused by a misinterpretation of the claim in \cite{Scalas2019}) 
that (full inductive) merging leads to type-unsafety. 

Tirore \etal~\cite{Tirore2023} defined
a projection on inductive types that is sound and complete
with respect to the coinductive projection 
with plain merging. 
They proved its correctness in Coq.
Their algorithm relies on plain merging which 
is only defined when all branches have the same projected type.
Therefore, all branches except one can be
discarded to form a candidate projection $\ptrans(\G)$ (\Sec\ref{section:coinductive_projection}). This enables its complexity to stay 
quadratic (Theorem~\ref{thm:coinductiveprojection}). 
They have left the
extension to full merging as their future work. 
Ghilezan \etal~\cite{Ghilezan2019} defined full merging on
regular coinductive types (which can be viewed as infinite trees). 
We reformulate the projection in \cite{Ghilezan2019} using the unfolding 
function, extend to the full coinductive merging and show 
its complexity is in \PSPACE, which is  
the highest among the four projections,
but produces the largest 
set of local types (see Figure~\ref{fig:venn}).




Majumdar \etal~\cite{DBLP:conf/concur/MajumdarMSZ21}
studied a projection of 
more expressive global types where a sender
can distribute messages to different participants via choice (sender-driven global types), and
proved its soundness against CFSMs.  
They did not analyse the complexity of their projection. 
Later, Stutz~\cite{Stutz2023} studied 
\emph{implementability} of 
those extended global types, i.e.~a projection is defined if there is any combination of local specifications that respect deadlock-freedom and protocol refinement.
An $\EXPSPACE$ bound for checking 
implementability is achieved for $0$-reachable, sender-driven global types. 
Li \etal~\cite{Li2023} later improved their implementability algorithm 
to be sound and complete with respect to 
protocol fidelity and deadlock-freedom, and proved it is in 
$\PSPACE$. Their conference paper \cite{Li2023_old} claimed 
their algorithm is \PSPACE-complete, which was corrected in their full version 
\cite{Li2023}, by showing   
a global type 
for which the construction of an implementation 
requires exponential time. 
The main differences from the projections we analysed are:
(1) their global types are \emph{more expressive}
(sender-driven and asynchronous); 
and (2) our projections guarantee \emph{liveness} of local types, while 
theirs guarantees \emph{deadlock-freedom}. For example, 
their algorithm implements 
$\G'$ 
in \cite[Example~3.12]{Ghilezan2019} (cf.~\Sec\ref{subsec:coinductivefull}), 
which generates non-live local types, 
and types non-live session $\M$ (see Remark~\ref{rem:liveness}).
In both their and our work, 
the hardness of the implementability (i.e., projectability) remains opened. 

\myparagraph{Type Inference.}
Vasconcelos and Honda~\cite{Vasconcelos1993} 
introduce a principal typing scheme 
of the polyadic pi-calculus, 
where types are simple tuples of types without subtyping.
Our minimum type inference system uses 
constraint derivations inspired by \cite[\Sec 22.3]{Pierce2002}
but needs to handle \emph{subtyping} over selection and branching
choices, together with a combination with recursion. 
These special features of session types 
require graph-based formulations of the minimum types.  
%
Scalas and Yoshida~\cite{Scalas2019} propose a general typing 
system, which we call the \emph{bottom-up} approach.  
The complexity of checking properties is proven \PSPACE-complete.
Notice that in asynchronous semantics, 
checking properties of typing contexts (bottom-up) is generally undecidable. 
\cite{Scalas2019} did not provide 
neither complexity nor a type inference system from processes. 
To our best knowledge, 
we have proposed the first type inference system for MPST, 
and shown that its worst case complexity is exponential; 
and proved the first \PSPACE-hardness results 
for checking typing context properties. 

\myparagraph{Future Work}
includes an accurate comparison with implementability and subtyping relations 
of asynchronous CFSMs \cite{Li2023,Li2024,BocchiKM24,BravettiCLYZ21}.
Our complexity result reveals that  
the coinductive projection with full merging has essentially 
similar complexity to the algorithm in \cite{Li2023},
which produces  
a larger set of local types preserving deadlock-freedom. 
On the other hand, the four projections studied in this paper 
yield a sound strict subset of \emph{live} local types 
under asynchronous semantics as proven in 
\cite[Appendix~N]{Scalas2019}. 
This projected set can be enlarged 
soundly using $\subt$ (since: $\subt\,\subseteq \,\subt_a$)
or decidable subsets of $\subt_a$
\cite{CYV2022,BocchiKM24,BravettiCLYZ21}
preserving liveness under the asynchronous 
semantics \cite[Theorem~6.3]{GPPSY2023}.
More exactly, the top-down methodology 
replacing $\subt$ by the appropriate (asynchronous) subtyping
in Figure~\ref{fig:top_downtwo} will enlarge the typability of
MPST processes under asynchrony.
Complexity analysis of the binary asynchronous sound session
subtyping in \cite{BocchiKM24,BravettiCLYZ21}
and their extensions to multiparty, as well as  
comparison with \cite{Li2023,Li2024} 
in terms of expressiveness and complexity are interesting future work. 
Our type inference system uses general graph formats 
for representing local types, which would be adaptable 
for the type inference from 
a flexible form of the calculus that directly uses CFSMs 
as types \cite[Chapter 6]{Stutz2024} or the calculus
with mixed choice \cite{PY2024}. 
Extensions to calculi with session delegations 
and hiding \cite{Scalas2019,HYC2016} are challenging 
future work. 


\myparagraph{Acknowledgement}
We thank the reviewers for detailed and helpful comments.  
This work is partially supported by EPSRC EP/T006544/2,
EP/N027833/2, EP/T014709/2, EP/Y005244/1, EP/V000462/1,
EP/X015955/1 and Horizon EU TaRDIS 101093006.


\bibliography{SessionTypes.bib}

\appendix 
\section{Appendix for the MPST Calculus}
\label{app:calculus}
\begin{table}{}
\[
\begin{array}[t]{@{}c@{}}
\eval{\fsucc\valn}(\valn +1)
\qquad
\eval{\fsqrt\valr}(-\valr)
\qquad
\eval{\neg\true}\false \quad \eval{\neg\false}\true 
\qquad \eval\val\val\\\\
\inferrule[]{\eval{\e_1}\true}
{\eval{\e_1\vee\e_2}{\true}}
\qquad
\inferrule[]{\eval{\e_2}\true}
{\eval{\e_1\vee\e_2}{\true}}
\qquad
\inferrule[]{\eval{\e_1}\false \quad \eval{\e_2}\false}
{\eval{\e_1\vee\e_2}{\false}}\\\\
\inferrule[]{\eval{\e_1}\val_1 \quad \eval{\e_2}\val_2 \quad \val_1+\val_1=\val}
{\eval{\e_1+\e_2}{\val}}
\qquad
\inferrule[]{\eval{\e_1}\val}
{\eval{\e_1\oplus\e_2}{\val}}
\qquad
\inferrule[]{\eval{\e_2}\val}
{\eval{\e_1\oplus\e_2}{\val}}
\qquad
\inferrule[]{\eval{\e}{\val}\quad\eval{\Econtext(\val)}{\val'}}{\eval{\Econtext(\e)}{\val'}}
\end{array}
\]
\caption{\label{tab:evaluation}  Expression evaluation.}
\vspace{10mm}%
\end{table}

An \emph{evaluation context} $\Econtext$ 
is an expression with exactly one hole. Table~\ref{tab:evaluation} defines the evaluation rules for expressions.

\section{Appendix for Complexity of Multiparty Subtyping Checking}
\label{app:subtyping}

\begin{lemma} \label{thm:type_graph_nodes_are_subformulas} If $\T'$ is reachable from $\T$ in the type graph
$\GG(\T)$ then $\T' \in \Sub(\T)$.
\end{lemma}
\begin{proof}
By induction on the length of the path from $\T'$ to $\T$, and matching this with the appropriate rule in Definition~\ref{def:subformulas_of_local_type}.
\end{proof}

\begin{lemma}
\label{thm:number_of_local_subterms_is_linear} $|\Sub(\T)| =
\bigO{|\T|}$.
\end{lemma}
\begin{proof}
By induction on the structure of $\T$. We will show that $|\Sub(\T)| \leq |\T|$.

    \case $\T = \tend$ or $\T = \ty$: trivial.

    \case $\T = \mu \ty. \T'$: $|\Sub(\T)| = |\{\mu \ty. \T'\} \cup \{\T''[\mu \ty.\T' / \ty] \mid \T'' \in \Sub(\T')\}| \leq 1 + |\Sub(\T')| = |\Sub(\T)|$.

    \case $\T = \tout\pp{\S}\T'$: $|\Sub(\T)| = |\{\tout\pp{\S}\T'\} \cup \Sub(\T')| \leq 1 + |\Sub(\T')| = |\Sub(\T)|$.

    \case $\T = \tsel\pp{\T_1, \dots, \T_n}$: $\Sub(\T) = |\{\tsel\pp{\T_1, \dots, \T_n}\} \cup \bigcup_{i=1}^n \Sub(\T_i)| \leq 1 + \sum_{i=1}^n |\Sub(\T_i)| \leq 1 + \sum_{i=1}^n |\T_i| = |\Sub(\T)|$.

Other cases: similar.
\end{proof}

\begin{lemma}\label{thm:number_of_edges_in_subtyping_lts_is_linear} The number of edges in the type graph for $\T$ is $O(|\T|)$.
  \begin{proof}
    Define the following function, which is the out-degree of an unfolded type: $\od(\ty) = 0$; $\od(\tend) = 1$; $\od(\tout\pp{S}\T) = \od(\tin\pp{\S}\T) = 1$; and $\od(\tselsub\pp{l_i: \T_i}{i \in \I}) = \od(\tbrasub\pp{l_i: \T_i}{i \in \I}) = |\I|$.
  
    Then, the number of edges in the LTS is
    $\sum_{U \text{ reachable from } T} \od(\unfold{U}) \leq \sum_{U \in \Sub(T)} \od(U) $.

    Let $f(\T) = \sum_{\T' \in \Sub(\T)} \od(\T')$.

    We prove that $f(T) \leq 2|\T| - 1$, by structural induction on $\T$.

    \case $T = \tend$, or $T = \ty$. Then $\Sub(\T) = \{\T\}$, so $f(\T) \leq 2|\T| - 1$.

    \case $\T = \mu \ty. \T'$. Then $f(\T) = \od(\T) + f(\T')$. By the inductive hypothesis, $f(\T') \leq 2|\T'| - 1$. By definition, $\od(\T) = 0$, so $\sum_{\T'' \in \Sub(\T)} \od(\T'') \leq 2|\T'| - 1 \leq 2|\T| - 1$.

    \case $\T = \tout\pp\S\T'$, $T = \tin\pp\S\T'$. Then $f(\T) = 1 + f(\T')$. By the inductive hypothesis, $f(\T') \leq 2|\T'| - 1$. By definition, $\od(\T) = 1$, so $f(\T) \leq 2|\T| - 1$.
  
    \case $T = \tselsub\pp{l_i: \T_i}{i \in \I}$, $T = \tbrasub\pp{l_i: \T_i}{i \in \I}$. Then $f(\T) = \od(\T) + \sum_{i \in \I} f(\T_i)$. By the inductive hypothesis, $f(\T_i) \leq 2|\T_i| - 1$. By definition, $\od(\T) = |\I|$, so $f(\T) \leq |\I| + \sum_{i \in \I} (2|\T_i| - 1) = 2\sum_{i \in \I} |T_i| - 1 \leq 2|\T| - 1$.
  \end{proof}
\end{lemma}

\begin{corollary} \label{thm:size_of_type_graph_is_linear}
The number of nodes in the type graph of $\T$ is $\bigO{|\T|}$.
\end{corollary}

\subtypingtypesimulation*
\label{app:subtypingtypesimulation}

\begin{proof}
  Note that $\subt$ is the largest type simulation, as all type simulations are consistent with the rules in Definition~\ref{def:subtyping_local_types}. Thus the forward direction is immediate. For the backward direction, if $\T_1 \cR \T_2$ then $\cR$ is consistent with the subtyping rules, so $\T_1 \subt \T_2$.
\end{proof}

\subsection{Proof of Theorem~\ref{thm:complexity_ghilezan_subtyping}}\label{section:proof_complexity_ghilezan_subtyping}

To help with the proofs of the complexity results in this section, we first define \emph{nesting depth}:

\begin{definition}[Nesting depth]
  The nesting depth of a type is defined inductively:
  \begin{align*}
    \nd(\tend) = \nd(\ty) &= 1\\
    \nd(\mu \ty. \T) &= \nd(\T) + 1\\
    \nd(\tin\pp\S\T) = \nd(\tout\pp\S\T) &= \nd(\T) + 1\\
    \nd(\tselsub\pp{l_i: \T_i}{i \in \I}) = \nd(\tbrasub\pp{l_i: \T_i}{i \in \I}) &= \max (\{\nd(\T_i) \mid i \in \I\}) + 1
  \end{align*}
\end{definition}

We first give an algorithm which is worst-case exponential, taken from Ghilezan \etal~\cite{Ghilezan2019}, which is the algorithm in \cite{Gay2005} adapted to the multiparty setting. We adapt the syntax to more closely match ours: in particular, we have separate constructs for sending and selection, as well as receiving and branching. The algorithm involves rules in Figure \ref{fig:ghilezhan_subtyping_rules}. These rules prove judgements of the form $\Theta \vdash \T \leq \T'$, which intuitively means ``assuming the relations in $\Theta$, we may conclude that $\T$ is a subtype of $\T'$''.

\begin{figure}
  \centering
  \small
  \setstretch{3}
  \begin{prooftree}
    \hypo{(\T, \T') \in \Theta}
    \infer1[\RULE{Alg-Assump}]{ \Theta \vdash \T \leq \T' }
  \end{prooftree}\hspace{2em}%
  \begin{prooftree}
    \infer0[\RULE{Alg-End}]{ \Theta \vdash \tend \leq \tend }
  \end{prooftree}

  \begin{prooftree}
    \hypo{\Theta \cup \{(\mu \ty.\T, \T')\} \vdash \T [\mu \ty.\T/\ty] \leq \T'}
    \infer1[\RULE{Alg-RecL}]{\Theta \vdash \mu \ty.\T \leq \T'}
  \end{prooftree}\hspace{2em}%
  \begin{prooftree}
    \hypo{\Theta \cup \{(\T, \mu \ty.\T')\} \vdash \T \leq \T'[\mu \ty. \T'/\ty]}
    \infer1[\RULE{Alg-RecR}]{\Theta \vdash \T \leq \mu \ty.\T'}
  \end{prooftree}

  \begin{prooftree}
    \hypo{\Theta \vdash \T_1 \leq \T'_1}
    \infer1[\RULE{Alg-In}]{\Theta \vdash \tin\pp\S{\T_1} \leq \tin\pp{\S}{\T'_1}}
  \end{prooftree}\hspace{2em}%
  \begin{prooftree}
  \hypo{\Theta \vdash \T_1 \leq \T'_1}
  \infer1[\RULE{Alg-Out}]{\Theta \vdash \tout\pp\S{\T_1} \leq \tout\pp{\S}{\T'_1}}
  \end{prooftree}

  \begin{prooftree}
    \hypo{\I \subseteq \J}
    \hypo{\forall i \in \I.\:\Theta \vdash \T_i \leq \T'_i}
    \infer2[\RULE{Alg-Bra}]{\Theta \vdash \tbrasub\pp{l_j: \T_j}{j \in \J} \leq \tbrasub\pp{l_i: \T_i}{i \in \I}}
  \end{prooftree}\hspace{2em}%
  \begin{prooftree}
    \hypo{\I \subseteq \J}
    \hypo{\forall i \in \I.\:\Theta \vdash \T_i \leq \T'_i}
    \infer2[\RULE{Alg-Sel}]{\Theta \vdash \tselsub\pp{l_i: \T_i}{i \in \I} \leq \tselsub\pp{l_j: \T_j}{j \in \J}}
  \end{prooftree}
  \vspace{5pt}
  \caption{Algorithmic rules for subtyping, from~\cite[Table 6]{Ghilezan2019} (syntax slightly modified).}
  \label{fig:ghilezhan_subtyping_rules}
\end{figure}

Ghilezan \etal\ \cite[Theorem 3.26]{Ghilezan2019} prove soundness and completeness of this type system:

\begin{theorem}
  $\emptyset \vdash \T \leq \T'$ iff $\T \leq \T'$.
\end{theorem}

Using this theorem, we can give an algorithm for subtyping of session types. The algorithm tries to build the proof tree bottom-up, using rule \RULE{Alg-Assump} if possible. Note that the rule taken at each step is deterministic (if we give \RULE{Alg-RecL} higher priority that \RULE{Alg-RecR}). By the above theorem, this algorithm is correct.

We adapt the results in \cite{Udomsrirungruang2024a} to prove the following complexity results of this algorithm.
We take $n$ = $|\T| + |\T|'$ where the input is $\T, \T'$:

\subsubsection{Upper bound}

\begin{theorem}\label{thm:worst_case_exponential_subtyping:upper} The upper bound of the worst case complexity of the subtyping algorithm in \cite{Ghilezan2019} is $n^{\bigO{n^3}}$.
\end{theorem}
\begin{proof}
  We proceed similarly to \cite[Theorem 3.6]{Udomsrirungruang2024a}. First, note that all types appearing in the proof tree of $\emptyset \vdash \T_1 \leq \T_2$ are members of $\Sub(\T_1) \cup \Sub(\T_2)$.
  Furthermore, if $\Theta'' \vdash \T_1'' \leq \T_2''$ appears above $\Theta' \vdash \T_1' \leq \T_2'$, then either $|\Theta''| > |\Theta'|$ or $\nd(\T_1'') < \nd(\T_1')$; this can be proved by induction on the proof tree. Thus, in any vertical path in the proof tree, the pairs $(\Theta', \T_1')$ are distinct.
  Therefore, in any judgement $\Theta \vdash \T_1' \leq \T_2'$ in the proof tree, there are only $\bigO{n^2}$ possible pairs that can be in $\Theta$. Furthermore, there are only $\bigO{n}$ possible values of $\T_1'$, and thus $\bigO{n}$ possible values of $\nd(\T_1')$. Therefore, the height of the proof tree is bounded by $\bigO{n^3}$.
  As the branching factor is $\bigO{n}$, we conclude that the total number of nodes in the proof tree is $n^{\bigO{n^3}}$. Each node requires $\bigO{1}$ time to process.
\end{proof}

\subsubsection{Lower bound}

We exhibit an input to Ghilezan \etal's subtyping algorithm where the runtime is exponential.

\newcommand{\Ta}{\T^\kf{a}}
\newcommand{\Tb}{\T^\kf{b}}
\newcommand{\Taf}{\T^\kf{af}}
\newcommand{\Tbf}{\T^\kf{bf}}
\newcommand{\Tc}{\T^\kf{c}}
\newcommand{\cC}{\mathcal{C}}
\newcommand{\infs}{\mathcal{S}}

\begin{definition}\label{def:exponential_example}
  Define the (not necessarily closed) type fragments:
  \begin{align*}
    \Taf_0 &= \ty & \\
    \Taf_r &= \tsel\pp{l_1: \Taf_{r-1}, l_2: \mu\ty'_{r-1}. \Tbf_{r-1}} & (r > 0) \\
    \Tbf_0 &= \ty & \\
    \Tbf_r &= \tsel\pp{l_1: \Tbf_{r-1}, l_2: \Tc} & (r > 0) \\
    \Tc &= \mu\ty''. \tsel\pp{l_1: \ty'', l_2: \ty''} &
  \end{align*}%
  Let $\T_k = \mu\ty. \Ta_k$. Define the (closed) types: $\Ta_{r, k} = \Taf_r[\T_k/\ty]$, $\Tb_{r, k} = \Tbf_r[\T_k/\ty]$.
\end{definition}

Note that in the above definition, there are no bound occurrences of $\ty'_k$, even though we bind it in a $\mu$.
This ensures that all nodes in the statement of Lemma \ref{lemma:exponential_nodes_in_proof_tree} are unique.

\begin{lemma}\label{lemma:exponential_nodes_in_proof_tree}
  For every sequence $\alpha_1, \dots, \alpha_l (0 \leq l < k)$ such that $0 \leq \alpha_i < k-i$:
  \begin{align}
\small    \left.
    \begin{aligned}
      &\{\T_k \leq \T_{k+1}, \Ta_{k, k} \leq T_{k+1}\}\\
      \cup& \{\mu \ty_{\alpha_i}. \Tb_{{\alpha_i}, k} \leq \mu \ty_{\alpha_i+i}. \Tb_{{\alpha_i+i},{k+1}} \mid 1 \leq i \leq l\}\\
      \cup& \{\Tb_{{\alpha_i},k} \leq \mu \ty_{\alpha_i+i}. \Tb_{{\alpha_i+i},{k+1}} \mid 1 \leq i \leq l\}\\
      \cup& \{\T_k \leq \Tb_{i,{k+1}} \mid 1 \leq i \leq l\} \\
      \cup& \{\Ta_{{k-i},k} \leq \T_{k+1} \mid 1 \leq i \leq l\}
    \end{aligned}
    \right\}
    &\vdash \Ta_{{k-l},k} \leq \Ta_{{k+1},{k+1}} \label{eq:exp_counterexample}
  \end{align}
  is derivable from $\emptyset \vdash T_k \leq T_{k+1}$ as the root.
\end{lemma}
\begin{proof}
  By induction on $l$. For some $\alpha_1, \dots, \alpha_l$, let $\infs_l$ be the set of inferences on the left side of (\ref{eq:exp_counterexample}).

  The base case $l=0$ is simple: $\infs_0 = \{\T_k \leq \T_{k+1}, \Ta_{k,k} \leq \T_{k+1}\}$. The corresponding proof tree is\\[1mm]
  \makebox[\textwidth]{
    \centering
    \scriptsize
    \begin{prooftree}
      \hypo{\infs_0 \vdash \Ta_{k, k} \leq \Ta_{{k+1},{k+1}}}
      \infer1[\RULE{Alg-RecR}]{\T_k \leq \T_{k+1} \vdash \Ta_{k,k} \leq \T_{k+1}}
      \infer1[\RULE{Alg-RecL}]{\emptyset \vdash \T_k \leq \T_{k+1}}
    \end{prooftree}
  }\\

  For the inductive step, we will build the tree starting from
  $\infs_{l-1} \vdash \Ta_{{k-(l-1)},k} \leq \Ta{{k+1},k}$, for $l > 0$,
  using the inductive hypothesis.
  The following proof tree works, taking\\
  $\cC_1 = \mu \ty_{\alpha_l}. \Tb_{{\alpha_l},k} \leq \mu \ty_{\alpha_l+l}. \Tb_{{\alpha_l+l},{k+1}},\:
  \cC_2 = \Tb_{{\alpha_l},k} \leq \mu \ty_{\alpha_l+l}. \Tb_{{\alpha_l+l},{k+1}},$\\
  $\cC_3 = \T_k \leq \Tb_{l,{k+1}},\:
  \cC_4 = \Ta_{{k-l},k} \leq \T_{k+1}$:\\

\smallskip 

  \makebox[\textwidth]{
    \centering
    \scriptsize
    \begin{prooftree}
      \hypo{\dots}
      \hypo{\infs_{l-1} \cup \{\cC_1, \cC_2, \cC_3, \cC_4\} \vdash \Ta_{{k-l},k} \leq \Ta_{{k+1},{k+1}}}
      \infer1[\RULE{Alg-RecR}]{\infs_{l-1} \cup \{\cC_1, \cC_2, \cC_3\} \vdash \Ta_{{k-l},k} \leq \T_{k+1} = \Tb_{0,{k+1}}}
      \hypo{\dots}
      \infer2[\RULE{Alg-Sel}]{\vdots}
      \infer1[\RULE{Alg-Sel}]{\infs_{l-1} \cup \{\cC_1, \cC_2, \cC_3\} \vdash \Ta_{k,k} \leq \Tb_{l,{k+1}}}
      \infer1[\RULE{Alg-RecL}]{\infs_{l-1} \cup \{\cC_1, \cC_2\} \vdash \Tb_{0,k} = \T_k \leq \Tb_{l,{k+1}}}
      \hypo{\dots}
      \infer2[\RULE{Alg-Sel}]{\vdots}
      \infer1[\RULE{Alg-Sel}]{\infs_{l-1} \cup \{\cC_1, \cC_2\} \vdash \Tb_{{\alpha_l},k} \leq \Tb_{{\alpha_l + l},{k+1}}}
      \infer1[\RULE{Alg-RecR}]{\infs_{l-1} \cup \{\cC_1\} \vdash \Tb_{{\alpha_l},k} \leq \mu \ty_{\alpha_l}. \Tb_{{\alpha_l+1},{k+1}}}
      \infer1[\RULE{Alg-RecL}]{\infs_{l-1} \vdash \mu \ty_{\alpha_l}. \Tb_{{\alpha_l},k} \leq \mu \ty_{\alpha_l+l}. \Tb_{{\alpha_l+l},{k+1}}}
      \infer2[\RULE{Alg-Sel}]{\infs_{l-1} \vdash \Ta_{{\alpha_l+1},k} \leq \Ta_{{\alpha_l+l+1},{k+1}}}
      \hypo{\dots}
      \infer2[\RULE{Alg-Sel}]{\vdots}
      \infer1[\RULE{Alg-Sel}]{\infs_{l-1} \vdash \Ta_{{k-(l-1)},k} \leq \Ta_{{k+1},{k+1}}}
    \end{prooftree}
  }\\

  Observing that $\infs_l = \infs_{l-1} \cup \{\cC_1, \cC_2, \cC_3, \cC_4\}$ finishes the proof.
\end{proof}

\begin{lemma} Using Definition \ref{def:exponential_example}, $\T_k \subteq \Tc$ for all $k \geq 0$.
\end{lemma}
\begin{proof}
  All subformulas of $\T_k$ have shape $\tsel\pp{l_1: \T^{(1)}, l_2: \T^{(2)}}$,
  for some $\T^{(1)}, \T^{(2)} \in \Sub(\T_k)$.
  Therefore, $\{(\T, \Tc) \mid \T \in \Sub(\T_k)\}$ and $\{(\Tc, \T) \mid \T \in \Sub(\T_k)\}$ are type simulations.
\end{proof}

\begin{corollary}\label{cor:tk_subtsim} $T_k \subtsim T_{k+1}$.
\end{corollary}
\begin{proof}
  Follows from transitivity of $\subteq$. 
\end{proof}

\begin{theorem}\label{thm:worst_case_exponential_subtyping:lower} The lower bound of the worst case complexity of the subtyping algorithm in \cite{Ghilezan2019} is $\bigOmega{(\sqrt{n})!}$.
\end{theorem}
\begin{proof}
  Consider the input $(\T_k, \T_{k+1})$. By Corollary \ref{cor:tk_subtsim},
  $\T_k \subt \T_{k+1}$, and so the algorithm in \cite{Ghilezan2019} will build
  a full proof tree of $\emptyset \vdash \T_k \subt \T_{k+1}$.
  By Lemma~\ref{lemma:exponential_nodes_in_proof_tree}, there is at least one (distinct) node in the proof tree for each
  sequence $\alpha_1, \dots, \alpha_l$ ($0 \leq l < k$, $0 \leq \alpha_i < k - i$), of which there are $k!$.
  We have $n = \bigO{k^2}$, so the number of nodes is at least $\bigOmega{(\sqrt n)!}$. 
\end{proof}


Combining the upper and lower bounds yields the stated result.

\complexityghilezansubtyping*

\subsection{Proof of Theorem~\ref{thm:quadratic_subtyping}}\label{section:quadratic_subtyping_proof}

Using this notion of product graphs and inconsistent nodes, we can
state the quadratic algorithm for subtyping as follows:

\begin{figure}[t]
  \begin{algorithmic}[1]
    \Function{Subtype}{$\T_1, \T_2$}
      \State $\mathcal{V} \gets \emptyset$ \Comment{Set of visited nodes}
      \State $\mathcal{Q} \gets \{(\T_1, \T_2)\}$ \Comment{Set of nodes to visit}
      \While{$\mathcal{Q} \neq \emptyset$}
        \State \text{remove} $(\T_1', \T_2')$ from $\mathcal{Q}$
        \If{$(\T_1', \T_2') \in \mathcal{V}$}
          \State \textbf{continue}
        \EndIf
        \State $\mathcal{V} \gets \mathcal{V} \cup \{(\T_1', \T_2')\}$
        \If{$(\T_1', \T_2')$ is inconsistent}\Comment{Definition~\ref{def:inconsistent_node}}
          \State \textbf{return} \textbf{false}
        \EndIf
        \For{each $(\T_1', \T_2') \trans{\gell} (\T_1'', \T_2'')$}
        \Comment{By definition~of~product subtyping graphs}
          \State $\mathcal{Q} \gets \mathcal{Q} \cup \{(\T_1'', \T_2'')\}$
        \EndFor
      \EndWhile
      \State \textbf{return} \textbf{true}
    \EndFunction
  \end{algorithmic}
  \caption{Algorithm for checking the subtyping relation.}
  \label{fig:subtyping_algorithm}
\end{figure}

\section{Appendix for Complexity of Projection Algorithms}
\label{app:projection}
\subsection{Appendix for Complexity of Inductive Projection with Plain Merging}
\label{app:plainprojection}

\begin{remark}[Merging and type safety]
\label{rem:merging}
Scalas and Yoshida \cite{Scalas2019} have discovered that the proofs of
type safety in the literature which use the full merging 
is flawed (while the proof using the plain merging
\cite{Honda2008,HYC2016,CDYP13} is \emph{not} flawed). 
After this paper, researchers wrongly believed
that the end-point projection is unsound. A recent paper 
by Yoshida and Hou \cite{YH2024} 
has corrected this misunderstanding, proposing a new general proof technique
for type soundness of multiparty session $\pi$-calculus, which uses an
association relation between a global type and its end-point 
projection. 
\end{remark}


\begin{lemma}\label{thm:number_of_global_subterms_is_linear} $\size{\Sub(\G)} = \bigO{|\G|}$.
\end{lemma}
 \begin{proof} By structural induction on $\G$, similarly to Lemma~\ref{thm:number_of_local_subterms_is_linear}.
 \end{proof}

\begin{lemma}\label{thm:plain_proj_is_smaller} Under plain merging, $|\proj{\G}{\pp}| \leq |\G|$.
\end{lemma}
\begin{proof}
Straightforward; by induction on the structure of $\G$.
\end{proof}

The following theorem completes the proof of Theorem~\ref{thm:inductive_plain_merge}.
\begin{theorem}
  There exists an algorithm for inductive projection with plain merging
  with worst-case time complexity $\bigO{\size{\G}\log\size{\G}}$. 
\end{theorem}
\begin{proof}
First, note that by keeping track of the free variables, we may check every syntactic subterm of $\G$ for closedness in $\bigO{n}$ time. We perform this as a first pass. Then, our algorithm will perform the projection in the same recursive sense as in Definition~\ref{def:inductive_projection}, except for the aforementioned checking of closedness. In particular, computing $\bigmergep_{i \in \I} \proj{\G_i}\pp$ will be a syntactic equality check. First, we prove the result for global types where the projection is defined.

\begin{claim}\label{claim:plaindefined} The algorithm computes $\proj\G\pp$ in $2|\G|\log|\G|$ time, if $\proj\G\pp$ is defined.
\end{claim}
\begin{proof}
We prove that the algorithm computes $\proj\G\pp$ in $2|\G|(\log|\G|+1)$ time, if $\proj\G\pp$ is defined.
We use base-$2$ logarithms for convenience: 

  \case $\G = \Gvt{\pq}{\pr}{S} \G'$. We have $|\G| = |\G'| + 1$; by the inductive hypothesis the time taken, including a constant time overhead, is $2|\G'|(\log|\G'|+1) + 1 \leq 2|\G|(\log|\G|+1)$.

  \case $\G = \GvtPair{\pq}{\pr}{l_i: \G_i}{i \in \I}$, $\pp \notin \{\pq, \pr\}$. If $\I = \{i\}$ then by the same logic as above, we compute $\proj\G\pp$ in time $2|\G'|(\log|\G'|+1) + 1 \leq 2|\G|(\log|\G|+1)$.
  
  Otherwise, $|\I| \geq 2$. Let $n_i = |\G_i|$ and $\T_i = \proj{\G_i}\pp$. Then the projections of $\G_i$ take time $2n_i \log n_i$; the merging is a syntactic equality check which takes time $\sum_{i \in \I} |\T_i|$. We are guaranteed that $\T_i = \T$ because of the plain merge.

  Let $x \in \I$ be such that $n_x$ is maximum. Thus, for $i \neq x$, $n_i \leq \frac{n}{2}$ so $\log n_i + 1 \leq \log n$. 
  The time taken for recursive projection is $\sum_{i \in \I}\left(2 n_i (\log n_i + 1)\right)$ and merging is $\sum_{i \in \I}|\T_i|$. Hence we have:
  \begin{align*}
    &\:\underbrace{\sum_{i \in \I}\left(2 n_i (\log n_i+1)\right)}_{\small\text{recursive projection}} 
    + \underbrace{\sum_{i \in \I}|\T_i|}_{\small\text{merging}}\\
    =&\: \sum_{i \in \I}\left(2 n_i (\log n_i+1)\right) + |\I| \cdot |\T|\\
    \leq&\: 2n_x(\log{n_x}+1) + \sum_{i \in \I \minus \{x\}}\left(2 n_i (\log n_i+1) + 2|\T|\right) && \text{(because $|\I| \geq 2$)}\\
    \leq&\: 2n_x(\log{n_x}+1) + \sum_{i \in \I \minus \{x\}}\left(2 n_i (\log n_i+1) + 2 n_i\right) && \text{(by Lemma~\ref{thm:plain_proj_is_smaller})}\\
    \leq&\: 2n_x(\log n+1) + \sum_{i \in \I \minus \{x\}}\left(2 n_i (\log n+1)\right)\\
    =&\: 2(\log n+1) \sum_{i \in \I} n_i\\
    \leq&\: 2 n (\log n+1)
  \end{align*}

  \case $\G = \GvtPair{\pq}{\pr}{l_i: \G_i}{i \in \I}$, $\pp \in \{\pq, \pr\}$. Then the time taken for the recursive projections (using the same notation as the above case) is $\sum_{i \in \I}\left(2 n_i (\log n_i+1)\right) \leq 2 n (\log n+1)$.

  \case $\G = \mu \ty. \G'$. Let $\T' = \proj{\G'}\pp$. We have $|\T| \leq |\T'| + 1$. Checking for closedness was precomputed, thus the total time is $2|\T'|(\log|\T'|+1) + 1 \leq 2|\T|(\log|\T|+1) \leq 2 |\G| (\log |\G| + 1)$.

  \case $\G=\tend$ and $\G=\ty$. Both are projectable in constant time, and so it is bounded by $\G (\log |\G| + 1) = 1$.
\end{proof}

\begin{claim}
\label{claim:plainundefined}
The algorithm computes 
$\proj\G\pp$ in 
$2|\G|\log|\G| + |\G|$, if $\proj\G\pp$ is undefined.
\end{claim}
\begin{proof}
Most reasoning is identical to the above proof, except for the case where an unprojectability error is found, i.e. $\G = \GvtPair{\pq}{\pr}{l_i: \G_i}{i \in \I}$, $\pp \notin \{\pq, \pr\}$, $|\I| \geq 2$, $\proj{\G_i}\pp$ is defined for all $i$, and $\bigmergep_{i \in \I}\proj{\G_i}\pp$ is undefined.
In this case we cannot rely on the fact that $|\T_j| < n_i$ for all $i$, $j$, which was necessary for the above proof.
Instead we prove the following bound: $\sum_{i \in \I}\left(2 n_i \log n_i\right) + \sum_{i \in \I}|\T_i| \leq 2\log n \cdot \sum_{i \in \I}n_i + \sum_{i \in \I}|\proj{\G_i}\pp| \leq 2n \log n + \sum_{i \in \I}|\proj{\G_i}\pp| \leq 2n \log n + n$.
\end{proof}
This concludes the proofs. 
\end{proof}



\subsection{Appendix for Complexity of Inductive Projection with Full Merging}
\label{app:fullprojection}

\begin{restatable}{lemma}{sizefullmerge}
 \label{thm:merge_is_smaller} 
If $\T_1 \mergef \T_2 = \T$, then $|\T| < |\T_1| + |\T_2|$.
\end{restatable}
\begin{proof}
  
By structural induction on $\T_1$.

  \case $\T_1 = \ty$ or $\T_1 = \tend$. Then $\T_2 = \T$, so $|\T| = 1 < 2 = |\T_1| + |\T_2|$.

  \case $\T_1 = \tin\pp{S}{\T_1'}$. Then $\T_2 = \tin\pp{S}{\T_2'}$, so $|\T| = 1 + |\T_1' \mergef \T_2'| < 1 + |\T_1'| + |\T_2'| < |\T_1| + |\T_2|$.

  \case $\T_1 = \tout\pp{S}{\T_1'}$. Similar to the above.
  
  \case $\T_1 = \tselsub\pp{l_i: \T_1^{(i)}}{i \in \I}$. Then $\T_2 = \tselsub\pp{l_i: \T_2^{(i)}}{i \in \I}$, so $|\T| = 1 + \sum_{i \in \I}|\T_1^{(i)} \mergef \T_2^{(i)}| < 1 + \sum_{i \in \I}(|\T_1^{(i)}| + |\T_2^{(i)}|) < 1 + |\T_1| + |\T_2|$.

  \case $\T_1 = \tbrasub\pp{l_i: \T_1^{(i)}}{i \in \I}$. Then $\T_2 = \tselsub\pp{l_i: \T_2^{(j)}}{j \in \J}$, so $|\T| = 1 + \sum_{i \in \I \cap \J}\left(\T_1^{(i)} \mergef \T_2^{(i)}\right) + \sum_{i \in \I \minus \J}\left(\T_1^{(i)}\right) + \sum_{i \in \J \minus \I}\left(\T_2^{(i)}\right) < 1 + \sum_{i \in \I}(\T_1^{(i)}) + \sum_{i \in \J}(\T_2^{(i)}) < |\T_1| + |\T_2|$.
\end{proof}

By repeated application of the above lemma, we have:

\begin{corollary}\label{thm:multi_merge_is_smaller} If $\bigmergef_{i \in \I} \T_i = \T$ and $\T$ is defined then $|\T| < \sum_{i \in \I} |\T_i|$.
\end{corollary}

\begin{lemma}
\label{thm:proj_is_smaller} Under full merging, $|\proj{\G}{\pp}| \leq |\G|$.
\end{lemma}
\begin{proof}
  By structural induction on $\G$.

  \case $\G = \left(\GvtPair{\pq}{\pr}{l_i: \G_i}{i \in \I}\right)$ where $\pp \notin \{\pq, \pr\}$.
  Then, using Corollary \ref{thm:multi_merge_is_smaller} and the inductive hypothesis,
  $|\proj\G\pp| = 1 + \bigmergef_{i \in \I} \proj{\G_i}{\pp} < 1 + \sum_{i \in \I} |\proj{\G_i}{\pp}|
  < 1 + \sum_{i \in \I}|\G_i| = |\G|$.
  
\noindent All other cases are trivial.
\end{proof}

\complexityfullmerge*

\begin{proof}
To allow for the complexity logarithmic in $|\T_1|$, in our algorithm (and other algorithms that use this lemma), we will represent local types as follows:
    \begin{itemize}
      \item $\ty$ and $\tend$ are primitive.
      \item $\tin\pp{\S}{\T}$ and $\tout\pp{\S}{\T}$ contain a pointer to $\T$.
      \item $\tselsub\pp{l_i: \T_i}{i \in \I}$ is represented as a (sorted) list of $\langle l_i, \text{pointer to }\T_i\rangle$ pairs.
      \item $\tbrasub\pp{l_i: \T_i}{i \in \I}$ is represented as a balanced binary search tree (BST) with keys $l_i$ and values as pointers to $\T_i$. This is done by first fixing a total ordering on labels consistently across the algorithm.
    \end{itemize}
    This differs from the na\"ive representation only in the representation of branching: instead of a list of branches, we will use a BST.

    Global types will be represented in the usual way, as syntactic trees.

    By structural induction on $\T_1$. We will show that, up to constant factors, the merge can be computed in time at most $2(|\T_1| + |\T_2| - |\T|) - 1 + (|\T_2| - 1) \cdot \log |\T_1|$. Note that this is positive by Lemma~\ref{thm:merge_is_smaller}.

    \case $\T_1 = \ty$ or $\T_1 = \tend$. Then $\T_2 = \T$. This can be checked in time $1 = 2(|\T_1| + |\T_2| - |\T|) - 1 + (| \T_2| - 1) \cdot \log |\T_1|$.
    
    \case $\T_1 = \tin\pp{S}{\T_1'}$, $\T_2 = \tin\pp{S}{\T_2'}$. Then it suffices to find $\T_1' \mergef \T_2'$. By the inductive hypothesis this takes 
    $2(|\T_1'| + |\T_2'| - |\T_1' \mergef \T_2'|) - 1 + (|\T_2'| - 1) \cdot \log |\T_1'|$ time.
    Adding a constant overhead of $1$, the total time is
    $2(|\T_1'| + |\T_2'| - |\T_1' \mergef \T_2'|) + (|\T_2'| - 1) \cdot \log |\T_1'|
    = 2(|\T_1| + |\T_2| - |\T|) - 2 + (|\T_2'| - 1) \cdot \log |\T_1'| 
    < 2(|\T_1| + |\T_2| - |\T|) - 1 + (|\T_2| - 1) \cdot \log |\T_1|$.

    \case $\T_1 = \tout\pp{S}{\T_1'}$, $\T_2 = \tout\pp{S}{\T_2'}$. Similar to the above.

    \case $\T_1 = \tselsub\pp{l_i: \T_1^{(i)}}{i \in \I}$, $\T_2 = \tselsub\pp{l_i: \T_2^{(i)}}{i \in \I}$. Then it suffices to find $\T_1^{(i)} \mergef \T_2^{(i)}$ for each $i \in \I$. By the inductive hypothesis this takes time:
    \begin{align*}
       &\:\sum_{i \in \I} \left(2(|\T_1^{(i)}| + |\T_2^{(i)}| - |\T_1^{(i)} \mergef \T_2^{(i)}|) - 1 + (|\T_2^{(i)}| - 1) \cdot \log |\T_1^{(i)}|\right)\\
      =&\:2\left(\sum_{i \in \I}|\T_1^{(i)}| + \sum_{i \in \I}|\T_2^{(i)}| - \sum_{i \in \I}|\T_1^{(i)} \mergef \T_2^{(i)}|\right) - |\I| + \sum_{i \in \I}(|\T_2^{(i)}| - 1) \cdot \log |\T_1^{(i)}|\\
      \leq&\:2(|\T_1| + |\T_2| - |\T|) - |\I| - 2 + \sum_{i \in \I}(|\T_2^{(i)}| - 1) \cdot \log |\T_1|\\
      <&\:2(|\T_1| + |\T_2| - |\T|) - |\I| - 1 + (|\T_2| - 1) \cdot \log |\T_1|\\
    \end{align*}
    Adding a linear overhead in the number of branches, $|\I|$, yields the result.

    \case $\T_1 = \tbrasub\pp{l_i: \T_1^{(i)}}{i \in \I}$, $\T_2 = \tbrasub\pp{l_i: \T_2^{(j)}}{j \in \J}$. We do the following:
    \begin{itemize}
      \item Start with the BST for $\T_1$, containing $\{l_i \mapsto \T_1^{(i)} \mid i \in \I\}$. Call this tree $\ts{T}$.
      \item Iterate over the BST for $\T_2$. For each item $l_j \mapsto \T_2^{(j)}$:
      \begin{itemize}
        \item If the key $l_j$ is present in $\ts{T}$, with value $\T_1^{(j)}$, then overwrite the value with $\T_1^{(j)} \mergef \T_2^{(j)}$.
        \item Otherwise, add a new key-value pair $l_j \mapsto \T_2^{(j)}$ to 
$\ts{T}$.
      \end{itemize}
    \end{itemize}
    It follows that the tree $\ts{T}$ at the end of this process is the tree representing the type $\T = \T_1 \mergef \T_2$. The time taken to perform this operation is the sum of the time taken for the merging of branches in $\{l_i \mid i \in \I \cap \J\}$, as well as the overhead of the $|\J|$ tree operations:
{
      \small
      \begin{align*}
        &\:\underbrace{
          \sum_{i \in \I \cap \J}\left(2(|\T_1^{(i)}| + |\T_2^{(i)}| - |\T_1^{(i)} \mergef \T_2^{(i)}|) - 1 + (|\T_2^{(i)}| - 1) \cdot \log|\T_1^{(i)}|\right)
        }_{\text{merging}} + 
        \underbrace{|\J| \log|\T_1|}_{\text{overhead}}\\
        \leq&\:\sum_{i \in \I \cap \J}\left(2(|\T_1^{(i)}| + |\T_2^{(i)}| - |\T_1^{(i)} \mergef \T_2^{(i)}|) - 1 + (|\T_2^{(i)}| - 1) \cdot \log|\T_1|\right) + |\J| \log|\T_1|\\
        =&\:2\left(\sum_{i \in \I \cap \J}|\T_1^{(i)}| + \sum_{i \in \I \cap \J}|\T_2^{(j)}| - \sum_{i \in \I \cap \J}|\T_1^{(i)} \mergef \T_2^{(i)}|\right)\\
& \qquad - |\I \cap \J|
        + \log|\T_1| \cdot \left(\sum_{i \in \I \cap \J} |\T_2^{(i)}| - |\I \cap \J| + |\J|\right)\\
      \end{align*}
    }
    We have
    {
      \small
      \begin{align*}
        &\:\sum_{i \in \I \cap \J}|\T_1^{(i)}| + \sum_{i \in \I \cap \J}|\T_2^{(j)}| - \sum_{i \in \I \cap \J}|\T_1^{(i)} \mergef \T_2^{(i)}|\\
        =&\:\sum_{i \in \I}|\T_1^{(i)}| + \sum_{i \in \J}|\T_2^{(i)}| - \left(\sum_{i \in \I \cap \J}|\T_1^{(i)} \mergef \T_2^{(i)}| + \sum_{i \in \I \minus \J}|\T_1^{(i)}| + \sum_{i \in \J \minus \I}|\T_2^{(i)}|\right)\\
        =&\:|\T_1| + |\T_2| - |\T| - 1
      \end{align*}
    }
    and
    {
      \small
      \begin{align*}
        &\:\sum_{i \in \I \cap \J} |\T_2^{(i)}| - |\I \cap \J| + |\J|\\
        =&\:\sum_{i \in \I \cap \J} |\T_2^{(i)}| + |\J \minus \I|\\
        \leq&\:\sum_{i \in \J} |\T_2^{(i)}| && \text{because } |T_2^{(i)}| \geq 1\\
        \leq&\:|\T_2| - 1
      \end{align*}
    }
    Thus the total time is less than $2(|\T_1| + |\T_2| - |\T|) - 2 - |\I \cap \J|
    + (|\T_2| - 1) \cdot \log|\T_1|$, as desired.
  \end{proof}

\begin{corollary}
  \label{cor:fullmerge}
  If $\bigmergef_{i \in \I} \T_i = \T$ and $\T$ is defined then $\T$ can be computed in time\\
  $\bigO{\sum_{i \in \I} {|\T_i|} - |\T| + \log \left(\sum_{i \in \I} |\T_i|\right) \cdot \sum_{i \in \I \minus \{j\}} |\T_i|}$, for any $j \in \I$.
\end{corollary}  
\begin{proof} Our algorithm proceeds as follows:
  \begin{algorithmic}
  \State $\T \gets \T_j$
  \For{$i \in \I \minus \{j\}$}
    \State $\T \gets \T \mergef \T_i$
  \EndFor
  \State \Return $\T$
\end{algorithmic}
Taking $\I = \{j, a_1, \dots, a_k\}$, label the $m$-th intermediate result $\T'_m = \bigmergef_{j \in \{x, a_1, \dots, a_m\}} \T_i$.
    By Lemma~\ref{thm:proj_is_smaller}, $|\T'_m| \leq \left(\sum_{i \in \I} |\T_i|\right)$.
    Then by the above lemma, calculating $\T'_{m-1} \mergef \T_m$ takes time proportional to\\[1mm]
\centerline{
$|\T'_{m-1}| + |\T_m| - |\T'_m| + |\T_m| \log |\T'_{m-1}|
      \leq\:|\T'_{m-1}| + |\T_m| - |\T'_m| + \log \left(\sum_{i \in \I} |\T_i|\right) \cdot |\T_m|$
}\\[1mm]
    Using the facts that $\T_0' = \T_j$ and $\T_k' = \T$, summing the above for $1 \leq m \leq k$ yields the desired result.
  \end{proof}

The following theorem completes the proof of Theorem~\ref{thm:fullmergeopt}.

\begin{theorem}
  There exists an algorithm for inductive projection with full merging
  with worst-case time complexity
  $\bigO{\size{\G}\log^2\size{\G}}$. 
\end{theorem}
\begin{proof}
  Similarly to the plain merging algorithm, we start by checking every syntactic subterm for closedness in linear time.
  Then, we show inductively that (up to constant factors) we can compute $\proj\G\pp$ in time $|\G| \log^2 |\G| + 2|\G| - |\T|$, where $\T = \proj\G\pp$.

  For convenience, define $|\T| = 0$ if $\T$ is undefined. We use base-$2$ logarithms for ease of explanations.

  \case $\G = \Gvt{\pq}{\pr}{S} \G'$. We have $|\G'| + 1 = |\G|$. Define $\T' = \proj{\G'}\pp$; we have $|\T| \leq |\T'| + 1$. Then the time taken is $1 + |\G'| \log^2 |\G'| + 2|\G'| - |\T'| \leq |\G'| \log^2 |\G'| + 2|\G| - |\T| \leq |\G| \log^2 |\G| + 2|\G| - |\T|$.

  \case $\G = \GvtPair{\pq}{\pr}{l_i: \G_i}{i \in \I}$, $\pp \notin \{\pq, \pr\}$. Denote $n_i = |\G_i|$. Let $m \in \I$ be such that $|G_m|$ is maximised. Let $\T_i = \proj{\G_i}{\pp}$. Then the time taken is
  \begin{align*}
    &\underbrace{\sum_{i \in \I}\left(n_i \log^2 n_i + 2n_i - |\T_i|\right)}_{\small\text{recursive projection}} 
    + \underbrace{\sum_{i \in \I}|\T_i| - |\T| + \log \left(\sum_{i \in \I} |\T_i|\right) \cdot \log |\T_m| \sum_{i \in \I \minus \{m\}} |\T_i|}_{\small\text{merging}}\\
    =&\: \sum_{i \in \I} \left(n_i \log^2 n_i\right) + \log \left(\sum_{i \in \I} |\T_i|\right) \cdot \left(\sum_{i \in \I \minus \{m\}} |\T_i|\right) + 2(|\G| - 1) - |\T|\\
    \leq&\: \log n \cdot \left(\sum_{i \in \I \minus \{m\}} \left(n_i \log n_i + n_i\right) + n_m \log n\right) + 2(|\G| - 1) - |\T|\\
    \leq&\: \log n \cdot \left(\sum_{i \in \I \minus \{m\}} \left(n_i \log n\right) + n_m \log n\right) + 2(|\G| - 1) - |\T|\\
    =&\: \log n \cdot \sum_{i \in \I} \left(n_i \log n\right) + 2(|\G| - 1) - |\T|\\
    =&\: n \log^2 n + 2(|\G| - 1) - |\T|
  \end{align*}
  using the fact that $\frac{n_i}n \leq \frac12$ for $i \neq m$ (therefore $\log n - \log n_i \geq 1$).
  
  \case $\G = \GvtPair{\pq}{\pr}{l_i: \G_i}{i \in \I}$, $\pp \in \{\pq, \pr\}$. Similar to the above case, with no merging. Thus this case can be performed in less time.
  
  \case $\G = \mu \ty. \G'$. Let $\T' = \proj{\G'}\pp$. We have $|\T| \leq |\T'| + 1$. Checking for closedness was precomputed, thus this is a constant-time check and a recursion: $1 + |\G'| \log^2 |\G'| + 2|\G'| - |\T'| \leq |\G| \log^2 |\G| + 2|\G| - |\T|$.

  \case $\G = \ty, \tend$. This takes constant time; $|\G| \log^2 |\G| + 2|\G| - |\T| = 1$ is positive.
\end{proof}




\subsection{Appendix for Complexity of Tirore \etal's Binder-Agnostic Projection \cite{Tirore2023}}
\label{app:coinductiveprojection}

\propifcp*

\begin{proof}
Assume by $\G$ is inductively 
projectable by the full merging but not by plain merging. 
We write $\G'\in \G$ if $\G'$ is syntactically subterm of $\G$.
Then there exits 
$\GvtPair{\pq}{\pr}{l_i: \G_i}{i \in \I} \in \G$ such that 
$\proj{G_k}{\pp} \not\equiv \proj{G_h}{\pp}$ with $k,h\in I$. 
By definition of plain merging (Definition~\ref{def:coinductive_plain_merge}), 
if it is coinductively projectable with the plain merging, $\mergef$ 
for the branching types should be restricted as:
\[
\tbrasub\pp{l_i: \T_i}{i \in \I} \mergef \tbrasub\pp{l_i: \T_i'}{i \in \I}
= \tbraset\pp\{l_k: (\T_k \mergef \T_k')\}_\{k \in \I\}
\]
with the rest of the rules of $\mergef$ in Definition~\ref{def:inductive_full_merging}, which yields $\mergep=\mergef$. 
\end{proof}

We prove Theorem~\ref{thm:coinductiveprojection}, 
recalling all definitions needed from \cite{Tirore2023}. 

Adapted to our syntax, the function $\ptrans$ is defined identically to inductive projection (Definition~\ref{def:inductive_projection}) with the following merging function: $\T_1 \mergep \T_2 = \T_1$. Clearly, $\ptrans$ is not a sound projection function on its own, so we must rely on the checking of $\G \cproj\pp \ptrans(\G)$.


\begin{definition} The graph of $(\G, \T)$, with respect to $\pp$, has nodes $\Sub(\G) \times \Sub(\T)$, with edges defined by the following rules:

  {
    \centering
    \small
    \begin{spacing}{4}
    \newcommand{\treesep}{.5em}
    \begin{prooftree}
      \hypo{\unfold{\G} = \GMsg \pq \pr \S \G'}
      \hypo{\unfold{\T} = \tin \pp \S \T' \vee \unfold{\T} = \tout \pp \S \T'}
      \hypo{\pp \in \{\pq, \pr\}}
      \infer3{(\G, \T) \rightarrow (\G', \T')}
    \end{prooftree}

    \begin{prooftree}
      \stackedhypos1{
        \hypo{\unfold{\T} = \tselsub \pp {l_i: \T_i}{i \in \J} \vee \unfold{\T} = \tbrasub \pp {l_i: \T_i}{i \in \J}}
      }3{
        \hypo{\unfold{\G} = \GvtPair \pq \pr {l_i: \G_i}{i \in \I}}
        \hypo{\pp \in \{\pq, \pr\}}
        \hypo{k \in \I \cap \J}
      }
      \infer1{(\G, \T) \rightarrow (\G_k, \T_k)}
    \end{prooftree}
    
    \begin{prooftree}
      \hypo{\G \trans{} \G'}
      \hypo{\unfold{\G} = \GMsg \pq \pr \S \G' \vee \unfold{\G} = \GvtPair \pq \pr {l_i : \G_i}{i \in \I}}
      \hypo{\pp \notin \{\pq, \pr\}}
      \infer3{(\G, \T) \rightarrow (\G', \T)}
    \end{prooftree}
    \end{spacing}
  }
\end{definition}

\begin{definition} For a graph $G = (V, E)$ and predicate $P$, $\sat_P(\{\}, v)$ is a predicate that holds if $P$ holds for all $w$ such that $w$ is reachable from $v$:

  \begin{equation*}
    \sat_P(V, v) = \begin{cases}
      1 & \text{if } v \in V\\
      P(v) \wedge \bigwedge_{v \rightarrow u \in E} \sat_P(V, u) & \text{otherwise}
    \end{cases}
  \end{equation*}
\end{definition}

\begin{definition} The predicate $\ProjPredp(\G, \T)$ is defined as follows:\\[2mm]
$\ProjPredp(\G, \T) =$ \\
  \begin{equation*}
\begin{cases}
      \kf{PL}_\pp(\kf{LG}(\G)) = \kf{LT}(\T) \wedge d_g(\G) = d_l(\T) & \text{if }\kf{PL}_\pp(\kf{LG}(\G))\text{ is defined}\\
      0 < d_g(\G) &\text{if }\kf{partOf}_\pp(\G)\text{ and }\kf{guarded}_\pp(\G)\\
      \sat_{\kf{UnravelPred}}(\{\}, \G) \wedge \neg\kf{partOf}_\pp(\G) \wedge \unfold{\T} = \tend & \text{otherwise}
    \end{cases}
  \end{equation*}
\end{definition}

All terms on the right-hand side of $\ProjPredp$, except for $\sat_{\kf{UnravelPred}}(\{\}, \G)$, are simple constant-time checks (possibly with some linear-time precomputation), which we omit. Also, $\kf{UnravelPred}$ itself is a constant-time check. The following theorem is proven in \cite{Tirore2023}:

\begin{theorem}\label{thm:projpred_is_proj} $\proj{\G}{\pp}={\T}$ if $\sat_{\ProjPredp}(\{\}, (\G, \T))$.
\end{theorem}

Therefore checking $\sat_\ProjPredp(\{\}, (\G, \ptrans(\G)))$ is enough to decide projection.

\subsubsection{Complexity}

$\ptrans$ is easy to compute: merging takes constant time thus it can be computed in~$\bigO{|\G|}$. Thus the main complexity of the algorithm comes from checking $\sat_{\ProjPredp}(\{\}, (\G, \T))$. The two functions to consider are $\sat_\ProjPredp$ and $\sat_{\kf{UnravelPred}}$. We have the following observations:

\begin{observation}\label{thm:graph_sat_in_linear_time} For any graph $G = (V, E)$, $\sat_P(\{\}, v)$ can be computed in $\bigO{|V| \cdot T + |E|}$ time, where $T$ is the time taken to compute $P(v)$.
\end{observation}
\begin{proof} Obvious; any graph search algorithm works.
\end{proof}

\begin{observation} The size of the graph of $\G$ is $\bigO{|\G|}$; the size of the graph of $(\G, \T)$ is $\bigO{|\G| \cdot |\T|}$.
\end{observation}
\begin{proof} The number of nodes is bound by $|\Sub(\G)|$ and $|\Sub(\G)| \cdot |\Sub(\T)|$ respectively; the number of edges for each graph can be bounded using structural induction.
\end{proof}

Thus we can show that checking the projection takes quadratic time:

\begin{lemma}\label{thm:projpred_quadratic} Checking $\sat_{\ProjPredp}(\{\}, (\G, \T))$ can be done in $\bigO{|\G| \cdot |\T|}$ time.
\end{lemma}
\begin{proof}
  First, computing $\sat_{\kf{UnravelPred}}(\{\}, \G')$ for all~$\G' \in \Sub(\G)$ takes~$\bigO{|\G|}$ time as $\kf{UnravelPred}$ is a constant-time predicate (then using Observation~\ref{thm:graph_sat_in_linear_time}.)

  After this precomputation, $\ProjPredp$ takes constant time to compute. So, checking $\sat_{\ProjPredp}(\{\}, (\G, \T))$ takes $\bigO{|\Sub(\G) \times \Sub(\T)|} \leq \bigO{|\G| \cdot |\T|}$ time.
\end{proof} 

\begin{restatable}{theorem}{theconductiveprojection}
  \label{thm:coinductiveprojection:upper}
  The projection in~\cite{Tirore2023} takes $\bigO{|\G|^2}$ time.
\end{restatable}
\begin{proof}
    We have shown above that checking $\G\cproj\pp\T$ takes $\bigO{|\G| \cdot |\T|}$ time (Lemma~\ref{thm:projpred_quadratic}, Theorem~\ref{thm:projpred_is_proj}). Additionally, checking $\ptrans$ takes $\bigO{|\G|}$ time. Using the fact that $|\T| \leq |\G|$ (analogue of Lemma~\ref{thm:proj_is_smaller} for $\ptrans$), we have the result.
\end{proof}

\newcommand{\Grep}{\G^{\kf{rep}}}
\newcommand{\Trep}{\T^{\kf{rep}}}
\begin{restatable}{theorem}{theconductiveprojectionlower}
  \label{thm:coinductiveprojection:lower}
  The projection in~\cite{Tirore2023} takes $\bigOmega{|\G|^2}$ time in the worst case.
\end{restatable}
\begin{proof}
  Define the (not necessarily closed) global type fragments:
  \begin{align*}
    \Grep_0 &= \ty \\
    \Grep_k &= \Gvt\pp\pq\tint \Grep_{k-1} & (k > 0) \\
    \Trep_0 &= \ty \\
    \Trep_k &= \tout\pq\tint \Trep_{k-1} & (k > 0)
  \end{align*}
  We can construct a global type $\G$ such that the the graph of $(\G, \ptrans(\G))$ has $\bigTheta{|\G|^2}$ nodes:
  \begin{align*}
    \G &= \mu\ty. \GvtPairs\pq\pr{l_1: \mu\ty. \Grep_n, l_2: \mu\ty. \Grep_{n+1}}
  \end{align*}
  We have that $\ptrans{\G} = \mu\ty. \Trep_n$.
  We have that $(\G, \ptrans(\G)) = (\G, \mu\ty. \Trep_n) \rightarrow (\mu\ty. \Grep_{n+1}, \mu\ty. \Trep_n)$.
  From here, all nodes $(\Grep_{i}\subst{\Grep_{n+1}}{\ty}, \Trep_{j}\subst{\Trep_n}{\ty})$, $0 \leq i \leq n, 0 \leq j \leq n-1$ are reachable.
  Using the fact that $|\G| = \bigTheta{n}$, the graph of $(\G, \ptrans(\G))$ has $\bigTheta{|\G|^2}$ nodes, thus the algorithm takes $\bigOmega{|\G|^2}$ time.
\end{proof}

\subsection{Appendix for Section \ref{subsec:coinductivefull}}
\label{app:coinductivefull}

\thmbalancedcomplexity*
\begin{proof}
Note that $\G$ is \emph{not} balanced if there exists a cycle from which some participant $\pp$ is reachable,
but the cycle itself uses only nodes where $\pp$ is not involved.
For each participant $\pp$, we can find all nodes that don't involve $\pp$ but $\pp$ is reachable in $\bigO{|\G|}$ time. We can then find cycles involving these nodes also in $\bigO{|\G|}$ time. Thus, we can check if $\G$ is balanced in $\bigO{|\G|^2}$ time, as there are $\bigO{|\G|}$ participants.
\end{proof}

\begin{example}[Subset construction]
\label{ex:subsetcontruction}
Consider the global type
\begin{align*}
\small
\G = \GvtPairs{\pp}{\pq}{&l_1: \Gvt{\pq}{\pr}{\tint} \GvtPairs{\pq}{\pr}{l_3: \Gvt{\pr}{\pp}{\tint} \tend},\\
                           &l_2: \Gvt{\pq}{\pr}{\tint} \GvtPairs{\pq}{\pr}{l_4: \Gvt{\pr}{\pp}{\tbool} \tend}}
  \end{align*}
The subset construction for projection onto $\pr$, $\gproj{\pr}(\G)$, has
initial node
$\gcl{\pr}(\{\G\})
 = \{\G, \Gvt{\pq}{\pr}{\tint} \GvtPairs{\pq}{\pr}{l_3: \Gvt{\pr}{\pp}{\tint} \tend},
    \Gvt{\pq}{\pr}{\tint} \GvtPairs{\pq}{\pr}{l_4: \Gvt{\pr}{\pp}{\tbool} \tend}\}$. Then:\\
$\begin{array}{rcl}
\gcl{\pr}(\{\G\})
  & \trans{\trin{\pq}{\tint}} & 
  \{\GvtPairs{\pq}{\pr}{l_3: \Gvt{\pr}{\pp}{\tint} \tend}, \GvtPairs{\pq}{\pr}{l_4: \Gvt{\pr}{\pp}{\tbool} \tend}\}\ = \ \cG\\
 \cG & \trans{\trbra{\pq}{l_3}} & 
  \{\Gvt{\pr}{\pp}{\tint} \tend\}
  \trans{\trout{\pp}{\tint}}
  \{\tend\}
  \trans{\trend}
  \Skip\quad\text{and}\\
\cG 
& \trans{\trbra{\pq}{l_4}}& 
  \{\Gvt{\pr}{\pp}{\tbool} \tend\}
  \trans{\trout{\pp}{\tbool}}
  \{\tend\}
  \trans{\trend}
  \Skip
\end{array}
$\\[1mm]
The corresponding local type is $\tin{\pq}{\tint} \tbra{\pq}{l_3: \tout{\pp}{\tint} \tend, l_4: \tout{\pp}{\tbool} \tend}$.
\end{example}

\begin{lemma}[Monotonicity and closure of $\gcl{\pp}$]
\label{lem:monotonicity_p_closure}
\label{lem:subset_closure_of_p_closure}
{\em (1)} $\cG \subseteq \cG'$ implies 
$\gcl{\pp}(\cG) \subseteq \gcl{\pp}(\cG')$; and 
{\em (2)} $\cG \subseteq \gcl{\pp}(\cG')$ implies 
$\gcl{\pp}{(\cG)} \subseteq \gcl{\pp}{(\cG')}$.
\end{lemma}
\begin{proof}
(1) By definition. (2) By (1) and  $\gcl{\pp}(\gcl{\pp}(\cG'))=\gcl{\pp}(\cG')$.
\end{proof}



\begin{restatable}{lemma}{lemsubsetconstruction}\label{lem:subset_construction_is_merging}
Let $\G$ be a balanced global type. If the projection graph $\gproj{\pp}(\G)$ is defined, then $\GG(\cG) = \bigmergecf_{\G' \in \cG}\GG(\{\G'\})$ for all nodes $\cG$ in $\gproj{\pp}(\G)$.
\end{restatable}


\begin{proof}
We show that the set of equations $\GG(\cG) = \bigmergecf_{\G' \in \cG}(\GG(\{\G'\}))$ is consistent with the coinductive definition of merging.
  
  \case $\cG \trans{\trend} \Skip$. Then for all $\G' \in \cG$, we have $\pp \notin \pt{\G'}$. Thus, $\GG(\{\G'\}) = \tend$, so the result follows.

  \case $\cG \trans{\trout{\pq}{\S}} \cG'$. Then for all $\G_i \in \cG$ involving $\pp$, we have $\unfold{\G_i} = \Gvt{\pp}{\pq}{\S}\G_i'$,
  and $\cG' = \bigcup_{i \in \I} \gcl{\pp}(\{\G_i'\})$.
  We have $\GG(\cG') = \bigmergecf_{\G' \in \cG'} \GG(\{\G'\})$. Thus, $\GG(\cG) = \bigmergecf_{\G' \in \cG} \GG(\{\G'\})$ is consistent with the definition.

  \case $\cG \trans{\trin{\pq}{\S}} \cG'$. Symmetric.

  \case $\cG \trans{\trsel{\pq}{l_j}} \cG'_j$. Then for all $\G_i \in \cG$ involving $\pp$,
  we have $\unfold{\G_i} = \GvtPair{\pp}{\pq}{l_i: \G_{ij}}{j \in \J}$.
  For all $j \in \J$, $\cG'_j = \bigcup_{i \in \I} \gcl{\pp}(\{\G_{ij}\})$.
  We have $\GG(\cG'_j) = \bigmergecf_{\G' \in \cG'_j} \GG(\{\G'\})$.
  Thus, $\GG(\cG) = \bigmergecf_{\G' \in \cG} \GG(\{\G'\})$ is consistent with the definition.

  \case $\cG \trans{\trbra{\pq}{l_j}} \cG'_j$. Then for all $\G_i \in \cG$ involving $\pp$,
  we have $\unfold{\G_i} = \GvtPair{\pq}{\pp}{l_j: \G_{ij}}{j \in \J_i}$.
  For all $j \in \bigcup_{i \in \I} \J_i$,
  $\cG'_j = \bigcup_{i \in \I, j \in \J_i} \gcl{\pp}(\{\G_{ij}\})$.
  We have $\GG(\cG'_j) = \bigmergecf_{\G' \in \cG'_j} \GG(\{\G'\})$.
  Thus, $\GG(\cG) = \bigmergecf_{\G' \in \cG} \GG(\{\G'\})$ is consistent with the definition.
\end{proof}

\thmcoprojsound* 
\begin{proof}
We show that $\cR = \{(\G', \GG(\gcl{\pp}(\{\G'\})))\ | \ \G' \in \Sub(\G)\}$ is a subset of the projection relation. 

Consider some pair $(\G', \GG(\cG))$, where $\cG = \gcl{\pp}(\{\G'\})$.
  
  \case If $\pp \notin \pt{\G'}$, then we have that $\cG \trans{\trend} \Skip$. Thus by definition, $\G' \fproj{\pp} \GG(\cG)$.
  
  \noindent Otherwise, proceed by case analysis on the shape of $\G'$.

  \case If $\unfold{\G'} = \Gvt{\pp}{\pq}{\S} \G''$, then $\cG = \{\G'\} \trans{\trout{\pq}{\S}} \gcl{\pp}(\{\G''\})$.
  We have that $\G'' \cR \GG(\gcl{\pp}(\{\G''\}))$ so $\G' \cR \GG(\cG)$ is consistent with the coinductive rules.

  \case If $\unfold{\G'} = \Gvt{\pq}{\pp}{\S} \G''$, symmetric.

  \case If $\unfold{\G'} = \Gvt{\pq}{\pr}{\S} \G''$, for $\pp \notin \{\pq, \pr\}$,
  then $\G'' \cR \GG(\gcl{\pp}(\{\G''\}))$. Note that $\gcl{\pp}(\{\G''\}) \cup \G' = \cG$,
  but $\G'$ does not involve $\pp$ so $\GG(\cG) = \GG(\gcl{\pp}(\{\G''\}))$.
  Therefore $\G' \cR \GG(\cG)$ is consistent with the coinductive rules.

  \case If $\unfold{\G'} = \GvtPair{\pp}{\pq}{l_i: \G_i}{i \in \I}$, then,
  for all $i \in \I$,
  $\cG = \{\G'\} \trans{\trsel{\pq}{l_i}} \gcl{\pp}(\{\G_i\})$.
  We have that $\G_i \cR \GG(\gcl{\pp}(\{\G_i\}))$ so $\G' \cR \GG(\cG)$ is consistent with the coinductive rules.

  \case If $\unfold{\G'} = \GvtPair{\pq}{\pp}{l_i: \G_i}{i \in \I}$, then,
  for all $i \in \I$,
  $\cG = \{\G'\} \trans{\trbra{\pq}{l_i}} \gcl{\pp}(\{\G_i\})$.
  We have that $\G_i \cR \GG(\gcl{\pp}(\{\G_i\}))$ so $\G' \cR \GG(\cG)$ is consistent with the coinductive rules.

  \case If $\unfold{\G'} = \GvtPair{\pq}{\pr}{l_i: \G_i}{i \in \I}$ for $\pp \notin \{\pq, \pr\}$,
  then, $\G_i \cR \GG(\gcl{\pp}(\{\G_i\}))$ for all $i \in \I$.
  Note that $\bigcup_{i \in \I} \gcl{\pp}(\{\G_i\}) \cup \G' = \cG$,
  but $\G'$ does not involve $\pp$ so $\GG(\cG) = \GG(\bigcup_{i \in \I} \gcl{\pp}(\{\G_i\})) = \bigmergecf_{i \in \I} \GG(\gcl{\pp})(\{\G_i\})$,
  by Lemma~\ref{lem:subset_construction_is_merging}.
  Therefore $\G' \cR \GG(\cG)$ is consistent with the coinductive rules.

  
  Thus, $\G \fproj{\pp} \GG(\gcl{\pp}(\{\G\}))$, as required.
\end{proof}

\begin{restatable}{lemma}{lemmcoprojcomplete}\label{lem:projection_and_merging_of_closure}
Let $\G$ be a balanced global type such 
that $\G \fproj{\pp} \T$ and $\gcl{\pp}(\G) = \{\G_i\}_{i \in \I}$.
Then there exists $\G_i \fproj{\pp} \T_i$ such that 
(1) $\T = \bigmergecf_{i \in \I} \T_i$; and 
(2) $\T = \bigmergecf_{k \in K} \T_k$ where 
$K=\set{i \ | \ \involve{\G_i}{\pp}, \ i \in I}$.
\end{restatable}

\begin{proof}
By the definition of a  balanced property and the definition of $\G \fproj{\pp} \T$,
  all global types in $\gcl{\pp}(\G)$ that do not involve $\pp$ has a projection which is the merging of
  other projections of $\gcl{\pp}(\G)$. Furthermore, the dependencies of the projections do not have a cycle.
  Thus, we may find the projections of the global types in $\gcl{\pp}(\G)$ that involve $\pp$, and merge
  them together to obtain all other projections. Thus (2) follows.
  (1) follows from the fact that merging is idempotent.
\end{proof}

\thmprojcomplete*

\begin{proof}
For contradiction, assume that $\G \fproj{\pp} \T$ but some local type $\cG \in \gproj{\pp}(\G)$ is invalid
  (i.e. its outgoing edges do not match any of the rules in the subset construction).
  Then, there exists a path in $\gproj{\pp}(\G)$ from $\gcl{\pp} (\G)$ to $\cG$.

  Without loss of generality $\cG$ is the first invalid node along the path.
  We show by induction that for all nodes $\cG' = \{\G_i\}_{i \in \I}$ along the path,
  there exists $\G_i \fproj{\pp} \T_i$ and $\T'$ such that $\T' = \bigmergecf_{i \in \I} \T_i$.

  First, the initial node $\gcl{\pp}(\G)$ contains $\G$, and $\G \fproj{\pp} \T$ by assumption.

  Now, consider some edge $\cG' \trans{\cxell} \cG''$, where $\cG' = \{\G_i'\}_{i \in \I}$,
  and assume there exists $\G_i' \fproj{\pp} \T_i'$ and $\T' \in \Sub(\T)$ such that $\T' = \bigmergecf_{i \in \I} \T_i'$.

  We analyse the possible cases for $\cxell$. Let $\I^* = \{i \in \I \:|\: \G_i' \text{ involves } \pp\}$.

  \case $\cxell = \trend$. Then, $\cG'' = \Skip$. Contradiction, as $\Skip$ is always valid.

  \case $\cxell = \trout{\pq}{\S}$. Then, $\G_i' = \Gvt{\pp}{\pq}{\S}\G_i''$ for all $i \in \I^*$
  For all such $i$, $\T_i' = \tout{\pq}{\S}\T_i''$. Thus, $\G_i'' \fproj{\pp} \T_i''$,
  and $\cG'' = \gcl{\pp}(\{\G_i''\}_{i \in \I^*})$.
  By Lemma~\ref{lem:projection_and_merging_of_closure}, $\T' = \bigmergecf_{i \in \I^*} \T_i'$.
  Thus, $\T' = \tout{\pq}{\S}\T''$ for $\T'' = \bigmergecf_{i \in \I^*} \T_i''$.
  Using Lemma~\ref{lem:projection_and_merging_of_closure} again, we get the desired result.

  \case $\cxell = \trin{\pq}{\S}$. Symmetric.

  \case $\cxell = \trsel{\pq}{l_k}$. Then, $\G_i' = \GvtPair{\pp}{\pq}{l_j: \G_{ij}}{j \in \J}$ for all $i \in \I^*$.
  For all such $i$, $\T_i' = \tsel{\pq}{l_k}\T_i''$,
  and $\cG'' = \gcl{\pp}(\{\G_{ik}\}_{i \in \I^*})$.
  By Lemma~\ref{lem:projection_and_merging_of_closure}, $\T' = \bigmergecf_{i \in \I^*} \T_i'$.
  Thus, $\T' = \tsel{\pq}{l_k}\T''$ for $\T'' = \bigmergecf_{i \in \I^*} \T_i''$.
  Using Lemma~\ref{lem:projection_and_merging_of_closure} again, we get the desired result.

  \case $\cxell = \trbra{\pq}{l_k}$. Then, $\G_i' = \GvtPair{\pp}{\pq}{l_j: \G_{ij}}{j \in \J_i}$ for all $i \in \I^*$.
  For all such $i$, $\T_i' = \tsel{\pq}{l_k}\T_i''$,
  and $\cG'' = \gcl{\pp}(\{\G_{ik}\}_{k \in \J_i, i \in \I^*})$.
  By Lemma~\ref{lem:projection_and_merging_of_closure}, $\T' = \bigmergecf_{k \in \J_i, i \in \I^*} \T_i'$.
  Thus, $\T' = \tsel{\pq}{l_k}\T''$ for $\T'' = \bigmergecf_{k \in \J_i, i \in \I^*} \T_i''$.
  Using Lemma~\ref{lem:projection_and_merging_of_closure} again, we get the desired result.

  Therefore, for $\cG = \{\G_i\}_{i \in \I}$, there exists $\G_i \fproj{\pp} \T_i$ and $\T'$ such that $\T' = \bigmergecf_{i \in \I} \T_i$.
  But there are incompatible outgoing edges from $\cG$, so the projections of $\G_i$ must have unmergeable first messages.
\end{proof}




\section{Appendix for Minimum Typing System}
\label{app:minimum}
\subsection{Typing Multiparty Sessions}
\label{app:typing} 
We recall the typing system of expressions from \cite{Ghilezan2019}. 

\begin{definition} \label{def:typingexpressions}
The typing rules for expressions are given below:
\begin{center}
  \small
\begin{spacing}{3}
\newcommand{\treesep}{.5em}





  \begin{prooftree}
    \hypo {x: \S \in \Gamma}
    \infer1[\RULE{T-SortVar}] {\Gamma \vdash x: \S}
  \end{prooftree}\hspace{\treesep}%
  \begin{prooftree}
    \infer0[\RULE{T-True}] {\Gamma \vdash \true: \tbool}
  \end{prooftree}\hspace{\treesep}%
  \begin{prooftree}
    \infer0[\RULE{T-False}] {\Gamma \vdash \false: \tbool}
  \end{prooftree}\hspace{\treesep}%
  \begin{prooftree}
    \infer0[\RULE{T-Nat}] {\Gamma \vdash \valn: \tnat}
  \end{prooftree}\hspace{\treesep}%

  \begin{prooftree}
    \infer0[\RULE{T-Int}] {\Gamma \vdash \valr: \tint}
  \end{prooftree}\hspace{\treesep}%
  \begin{prooftree}
    \hypo {\Gamma \vdash \e:\tint}
    \infer1[\RULE{T-Neg}] {\Gamma \vdash \fsqrt\e: \tint}
  \end{prooftree}
  \begin{prooftree}
    \hypo {\Gamma \vdash \e: \tbool}
    \infer1[\RULE{T-Not}] {\Gamma \vdash \neg \e: \tbool}
  \end{prooftree}\hspace{\treesep}%

  \begin{prooftree}
    \hypo {\Gamma \vdash \e_1: \tbool}
    \hypo {\Gamma \vdash \e_2: \tbool}
    \infer2[\RULE{T-Or}] {\Gamma \vdash \e_1 \vee \e_2: \tbool}
  \end{prooftree}
  \begin{prooftree}
    \hypo {\Gamma \vdash \e_1: S}
    \hypo {\Gamma \vdash \e_2: S}
    \infer2[\RULE{T-Nondet}] {\Gamma \vdash \e_1 \oplus \e_2: \S}
  \end{prooftree}
  \begin{prooftree}
    \hypo {\Gamma \vdash \e_1: \tint}
    \hypo {\Gamma \vdash \e_2: \tint}
    \infer2[\RULE{T-Add}] {\Gamma \vdash \e_1 + \e_2: \tint}
  \end{prooftree}\hspace{\treesep}%

  \end{spacing}
\end{center}

\end{definition}

\subsection{Constraint Inference Rules for Expressions}
\label{app:constraint_rules}

We explain the rules for the expressions in Figure~\ref{fig:mininfer}. 
\RULE{C-SortVar} states that a sort variable $x$ must have a matching sort with the typing context.
Rules \RULE{C-True}, \RULE{C-False}, \RULE{C-Nat}, and \RULE{C-Int} state that the expressions $\true$, $\false$, $\valn$, and $\valr$ must have the corresponding sorts.
Rule \RULE{C-Neg} states that the expression $\fsqrt\e$ and $\e$ must both have sort $\tint$.
Rule \RULE{C-Not} states that the expression $\neg \e$ and $\e$ must both have sort $\tbool$.
Similarly, \RULE{C-Or} states that all relevant expressions must have sort $\tbool$.
Rule \RULE{C-Nondet} states that the expressions $\e_1$, $\e_2$ and $\e_1 \oplus \e_2$ must have the same sort.
Rule \RULE{C-Add} states that the expressions $\e_1$, $\e_2$ and $\e_1 + \e_2$ must have sort $\tint$.


\subsection{Proofs for Soundness and Completeness of Constraint Derivations}
\label{app:subsec:soundcomplete}

\uniqueconstraintset* 
\begin{proof}
We prove by induction on the proof trees of $\Gamma \vdash \PP: \xi \mid_\cX \C$ and $\Gamma' \vdash \PP: \xi' \mid_{\cX'} \C'$.

  \case \RULE{C-End}. Then $\C = \C' = \emptyset$. Take $\sigma = \emptyset$, $\pi = \emptyset$.


  \case \RULE{C-Bra}. Then $\PP = \procbra \pp{l_i: \PP_i}_{i \in \I}$.

  Then we have proof trees for $\Gamma \vdash \PP_i: \xi_i \mid_{\cX_i} \C_i$, $\Gamma' \vdash \PP_i: \xi_i' \mid_{\cX_i'} \C_i'$, for all $i \in \I$. By the inductive hypothesis, there exist one-to-one functions $\sigma_i, \pi_i$ such that $\C_i\pi_i\sigma_i = \C_i'$. Let $\sigma = \bigcup_{i \in \I} \sigma_i \cup \{\xi \mapsto \xi'\}$ and $\pi = \bigcup_{i \in \I} \pi_i$. Then, as $\cX_i, \xi, \xi'$ are disjoint, $\sigma$ and $\pi$ are one-to-one. Furthermore, 
$\C\pi\sigma = \C'$.

  Other cases: Similar.
\end{proof}

\begin{lemma}\label{thm:exists_derivation_with_equal_domain_gamma}
\label{thm:unique_domain_gamma}
Let $\Gamma \vdash \PP: \T$ be derivable. Then there exists a derivation of $\Gamma' \vdash \PP: \xi \mid_\cX \C$ such that $\dom{\Gamma} = \dom{\Gamma'}$.
\end{lemma}
\begin{proof}
By matching each typing rule with the corresponding constraint rule.  
\end{proof}

\soundminimum*

\begin{proof}
We show that if $\Gamma \vdash \PP: \xi \mid_\cX \C$ is derivable and $(\sigma, \pi)$ is a solution to $\C$,
then there is a typing derivation for $\Gamma\sigma\pi \vdash \PP: \xi\sigma\pi$
(where we define $\Gamma\sigma\pi$ to apply $\sigma\pi$ to each type in $\Gamma$).
Also, we prove this for sorts: if $\Gamma \vdash \e: \alpha \mid_\cX \C$
is derivable then $\Gamma\sigma\pi \vdash \e: \alpha\sigma\pi$.

We prove this by induction on the proof tree of the constraint derivations.

    \case \RULE{C-End}. Then we can use \RULE{T-Inact} followed by \RULE{T-Sub}.

    \case \RULE{C-Var}. Then we can use \RULE{T-Var}.
    
    \case \RULE{C-Rec}. Then by the inductive hypothesis there is a typing derivation for $\Gamma\sigma\pi, \ty: \xi\sigma\pi \vdash \PP: \xi\sigma\pi$. We can use \RULE{T-Rec} to get $\Gamma\sigma\pi \vdash \mu \ty. \PP: \xi\sigma\pi$.

    \case \RULE{C-In}. If $\sigma\pi$ is a solution to $\C'$ then it is also a solution of $\C$. By the inductive hypothesis, there is a typing derivation for $\Gamma\sigma\pi, x: \alpha\pi \vdash \PP: \psi\sigma\pi$. We also have $\tin \pp {\pi \alpha} \psi\sigma\pi \subt \xi\sigma\pi$. We can then do the following to create the required type derivation:\\[2mm]
\centerline{
    \begin{prooftree}
      \hypo {\Gamma\sigma\pi, x: \alpha\pi \vdash \PP: \psi\sigma\pi}
      \infer1[\RULE{T-In}] {\Gamma\sigma\pi \vdash \procin \pp{x} \PP: \tin \pp{S\sigma\pi } \psi\sigma\pi}
      \hypo {\tin \pp{\alpha\pi} \psi\sigma\pi \subt \xi\sigma\pi}
      \infer2[\RULE{T-Sub}] {\Gamma\sigma\pi \vdash \procin \pp{x} \PP: \xi\sigma\pi}
    \end{prooftree}
}\\[2mm]
    \case \RULE{C-Out}. We have that $\sigma\pi$ is a solution of $\C_1$ and $\C_2$. By the inductive hypothesis, there are typing derivations for $\Gamma\sigma\pi \vdash \PP: \psi\sigma\pi$ and $\Gamma\sigma\pi \vdash \e: \alpha\pi$. As $(\sigma, \pi)$ is a solution to $\C'$, $\Gamma\sigma\pi \vdash \procout \pp{\e} \PP: \xi\sigma\pi$. Proceed with the following derivation:\\[2mm]
\centerline{
    \begin{prooftree}
      \hypo {\Gamma\sigma\pi \vdash \PP: \xi\sigma\pi}
      \hypo {\Gamma\sigma\pi \vdash \e: \alpha\pi}
      \infer2[\RULE{T-Out}] {\Gamma\sigma\pi \vdash \procout \pp{\e} \PP: \tout \pp{\alpha\pi} \psi\sigma\pi}
      \hypo {\tout \pp{\pi\alpha} \psi\sigma\pi \subt \xi\sigma\pi}
      \infer2[\RULE{T-Sub}] {\Gamma\sigma\pi \vdash \procout \pp{\e} \PP: \xi\sigma\pi}
    \end{prooftree}
}\\[2mm]

    \case \RULE{C-Bra}, \RULE{C-Sel}. Similar to the above: use rules \RULE{T-Bra} then \RULE{T-Sub}, and \RULE{T-Sel} then \RULE{T-Sub} respectively.

    \case \RULE{C-Cond}. We have that $(\sigma, \pi)$ is a solution of $\C_1$, $\C_2$ and $\C_3$. By the inductive hypothesis, there are typing derivations for $\Gamma\sigma\pi \vdash \PP: \psi_1\sigma\pi$, $\Gamma\sigma\pi \vdash \PP': \psi_2\sigma\pi$ and $\Gamma\sigma\pi \vdash \e: \alpha$, with $\alpha = \tbool \in \C_3$. Thus $\alpha\pi = \tbool$.\\[2mm]
    {\small
      \begin{prooftree}
        \hypo {\Gamma\sigma\pi \vdash \PP: \psi_1\sigma\pi}
        \hypo {\psi_1\sigma\pi \subt \xi\sigma\pi}
        \infer2 {\Gamma\sigma\pi \vdash \PP: \psi\sigma\pi}
        \hypo {\Gamma\sigma\pi \vdash \PP': \psi_2\sigma\pi}
        \hypo {\psi_2\sigma\pi \subt \xi\sigma\pi}
        \infer2 {\Gamma\sigma\pi \vdash \PP': \xi\sigma\pi}
        \hypo {\Gamma\sigma\pi \vdash \e: \tbool}
        \infer3[\RULE{T-Cond}] {\Gamma\sigma\pi \vdash \cond{\e}{\PP}{\PP'}: \sigma\pi\xi}
      \end{prooftree}
    }

    \smallskip

\noindent(where \RULE{T-Sub} is used in the top subtrees).

    \case \RULE{C-SortVar}. We have $x: \alpha \in \Gamma$. Therefore we can use \RULE{T-SortVar}:\\[2mm]
    \begin{prooftree}
      \hypo {x: \alpha\pi \in \Gamma\sigma\pi}
      \infer1[\RULE{T-SortVar}] {\Gamma\sigma\pi \vdash x: \alpha\pi}
    \end{prooftree}

    \smallskip
    
    \case \RULE{C-True}. We have $\tbool = \alpha \in \C$. Then we can use \RULE{T-True}:\\[2mm]
    \begin{prooftree}
      \infer0[\RULE{T-True}] {\Gamma\sigma\pi \vdash \true: \alpha\pi}
    \end{prooftree}

    \smallskip

    \case \RULE{C-Nat}. Then we can use \RULE{T-Nat}.

    \case \RULE{C-Int}. Then we can use \RULE{T-Int}.

    \case \RULE{C-Not}. By the inductive hypothesis, there is a typing derivation for $\Gamma\sigma\pi \vdash \e: \tbool$. We can use \RULE{T-Not}, noting that $\alpha\pi = \tbool$:\\[2mm]
    \begin{prooftree}
      \hypo {\Gamma\sigma\pi \vdash \e: \tbool}
      \infer1[\RULE{T-Not}] {\Gamma\sigma\pi \vdash \neg \e: \tbool}
    \end{prooftree}\\

    \case \RULE{C-Or}, \RULE{C-Add}, \RULE{C-Neg}. Similar to the above.

    \case \RULE{C-Nondet}. We have that $(\sigma, \pi)$ is a solution of $\C$, so $\alpha_1\pi = \alpha_2\pi = \beta\pi$. By the inductive hypothesis, there are typing derivations for $\Gamma\sigma\pi \vdash \e_1: \alpha_1\pi$ and 
$\Gamma\sigma\pi \vdash \e_2: \alpha_2\pi$. We can use the following derivation:

\smallskip
    
    \begin{prooftree}
      \hypo {\Gamma\sigma\pi \vdash \e_1: \alpha_1\pi}
      \hypo {\Gamma\sigma\pi \vdash \e_2: \alpha_2\pi}
      \infer2[\RULE{T-Nondet}]{\Gamma\sigma\pi \vdash \e_1 \oplus \e_2: \beta\pi}
    \end{prooftree}
  \end{proof}

\completeminimum* 

\begin{proof}
We show that if $\Gamma \vdash \PP: \xi \mid_\cX \C$, and $\Gamma\sigma\pi \vdash \PP: \T_0$, such that $\dom{\sigma} \cap \cX = \emptyset$ and $\dom{\pi} \cap \cX = \emptyset$,
then there exists a solution $(\sigma', \pi')$ of $\C$ such that $\xi\sigma'\pi' = \T_0$, $\sigma' \minus \cX = \sigma$, and $\pi' \minus \cX = \pi$.

We also show the analogous statement for sorts,
under the same conditions as the above statement for types: if $\Gamma \vdash \e: \alpha \mid_\cX \C$, and $\Gamma\sigma\pi \vdash \e: \S$,
then there exists a solution $(\sigma', \pi')$ of $\C$ such that $\alpha\pi' = \S$, $\sigma' \minus \cX = \sigma$, and $\pi' \minus \cX = \pi$.

We do this by induction on the proof tree of the type derivation. For the following sort cases, we will consider the first rule in the derivation.

  \case \RULE{T-SortVar}. Then $\e = x$, $x: \alpha \in \Gamma\sigma\pi$. By inversion we have that \RULE{C-SortVar} was used, and so $\C = \emptyset$, $\cX = \emptyset$. Take $\sigma' = \sigma$ and $\pi' = \pi$. Then $(\sigma', \pi')$ is a solution of $\C$ because $\alpha\pi = \S$.

  \case \RULE{T-True}. Then $\e = \true$, $S = \tbool$. By inversion we have that \RULE{C-True} was used, and so $\C = \{\tbool = \alpha\}$, $\cX = \{\alpha\}$. Take $\sigma' = \sigma, \pi' = \pi \cup \{\alpha \mapsto \tbool\}$. Then $(\sigma', \pi')$ is a solution of $\C$ because $\alpha\pi = \tbool$.

  \case \RULE{T-Nondet}. Then $\e = \e_1 \oplus \e_2$, $\Gamma\sigma\pi \vdash \e_1: \S$, $\Gamma \vdash \e_2: \S$. By inversion we have that \RULE{C-Nondet} was used, and so $\Gamma \vdash \e_1: \alpha_1 \mid_{\cX_1} \C_1$, $\Gamma \vdash \e_2: \alpha_2 \mid_{\cX_2} \C_2$, $\cX = \cX_1 \dcup \cX_2 \dcup \{\beta\}$, $\C = \C_1 \cup \C_2 \cup \{\alpha_1 = \beta, \alpha_2 = \beta\}$.

  By the inductive hypothesis, there exist solutions $(\sigma_1, \pi_1)$ of $\C_1$ and $(\sigma_2, \pi_2)$ of $\C_2$ such that $\sigma_1 \minus \cX_1 = \sigma$, $\sigma_2 \minus \cX_2 = \sigma$, $\pi_1 \minus \cX_1 = \pi_1$, $\pi_2 \minus \cX_2 = \pi_2$, $\alpha_1\pi_1  = \S$ and $\pi_2 \alpha_2 = \S$. Take $\sigma' = \sigma_1 \cup \sigma_2, \pi' = \pi_1 \cup \pi_2 \cup \cup \{\beta \mapsto \S\}$. We have that $\pi'\beta = \S$ as desired.

  \case \RULE{T-False}, \RULE{T-Int}, \RULE{T-Nat}, \RULE{T-Not}, \RULE{T-Or}, \RULE{T-Add}, \RULE{T-Neg}. Similar to the above.

\medskip

For the type cases, we will consider the first non-\RULE{T-Sub} rule used in the derivation, so the tree covered is of shape:\\[2mm]
\centerline{
    \begin{prooftree}
      \hypo{...}
      \infer1[\RULE{R}] {\Gamma\sigma\pi \vdash \PP: \T_k}
      \hypo {\T_k \subt \T_{k-1}}
      \infer2 {\Gamma\sigma\pi \vdash \PP: \T_{k-1}}
      \hypo {\T_{k-1} \subt \T_{k-2}}
      \infer2{\vdots}
      \hypo {\T_1 \subt \T_0}
      \infer2{\Gamma\sigma\pi \vdash \PP: \T_0}
    \end{prooftree}
}\\[2mm]
    Thus we have that $\Gamma\sigma\pi \vdash \PP: \T_k$ is derivable and $\T_k \subt \T_0$ by transitivity of subtyping.

    We will now consider the rule $\RULE{R}$ used to derive $\Gamma\sigma\pi \vdash \PP: \T_k$.

    \case \RULE{T-Inact}. Then $\PP = \inact$, $\T_k = \tend$. By inversion we have that \RULE{C-End} was used, and so $\C = \{\tend \subt \xi\}$, $\cX = \{\xi\}$. Take $\sigma' = \sigma \cup \{\xi \mapsto \T_0\}, \pi' = \pi$. Then $(\sigma', \pi')$ is a solution of $\C$ because $\tend = \T_k \subt \T_0$.

    \case \RULE{T-Var}. Then $\PP = X$, $X: \T_k \in \Gamma\sigma\pi$. By inversion we have that \RULE{C-Var} was used, and so $\C = \emptyset$, $X: \xi \in \Gamma$.

    Therefore $\xi\sigma\pi = \T_k$. Take $\sigma' = \sigma \cup \{\xi \mapsto \T_0\}, \pi' = \pi$. Then $(\sigma', \pi')$ is a solution of $\C$ because $\T_k \subt \T_0$. Furthermore, $\pi' \sigma' \xi = \T_0$ as desired.

    \case \RULE{T-Rec}. We can assume that no \RULE{T-Sub} rule was used directly below \RULE{T-Rec}, as we could have done the subtyping directly above instead. So $\T_k = \T_0$.
    
    Then $\PP = \mu X. \PP'$, $\Gamma\sigma\pi, X: \T_k \vdash \PP': \T_k$. By inversion we have that \RULE{C-Rec} was used, and so $\Gamma, \ty: \xi \vdash \PP': \xi \mid_\cX \C$.

    By the inductive hypothesis there exists a solution $(\sigma_1, \pi_1)$ of $\C$ such that $\sigma_1 \minus \cX = \sigma$, $\pi_1 \minus \cX = \pi$, 
$\sigma_1\xi\pi_1 = \T_k$. Taking $\sigma' = \sigma_1$, $\pi' = \pi_1$ works because $\T_k = \T_0$.

    \case \RULE{T-In}. Then $\PP = \procin \pp{x} \PP'$, $\T_k = \tin \pp{S} \T'$, $\Gamma\sigma\pi, x: S \vdash \PP': \T'$. By inversion we have that \RULE{C-In} was used: $\Gamma, x: \alpha \vdash \PP': \psi \mid_{\cX'} \C'$, $\cX = \cX' \dcup \{\xi, \alpha\}$, $\C = \C' \cup \{\tin \pp{\alpha} \psi \subt \xi\}$, $\alpha, \psi$ fresh.

    Let $\sigma'' = \sigma$ and $\pi'' = \pi \dcup \{\alpha \mapsto \S\}$ and $\Gamma' = \Gamma, x: \alpha$. Then $\Gamma'\pi''\sigma'' \vdash \PP': \T'$. By the inductive hypothesis, there exists a solution $(\sigma_1, \pi_1)$ of $\C'$ such that $\sigma_1 \minus \cX' = \sigma''$, $\pi_1 \minus \cX' = \pi''$, 
$\psi\pi_1\sigma_1 = \T'$. Take $\sigma' = \sigma_1 \cup \{\xi \mapsto \T_0\}, \pi' = \pi_1$. Then $\xi\sigma'\pi' = \T_0$ and $\pi' \alpha = S$ and $\tin \pp{S} \T' = \T_k \subt \T_0$. Therefore $(\sigma', \pi')$ is a solution of $\C$ as desired.

    Also, $\sigma' \minus \cX = (\sigma_1 \cup \{\xi \mapsto \T_0\}) \minus (\cX' \dcup \{\xi, \alpha\}) = \sigma_1 \minus \cX = \sigma$ and $\pi' \minus \cX = (\pi_1 \cup \{\alpha \mapsto \S\}) \minus (\cX' \dcup \{\xi, S'\}) = \pi'' \minus \cX = \pi$, as desired.

    \case \RULE{T-Out}. Then $\PP = \procout \pp{\e} \PP'$, $\T_k = \tout \pp{S} \T'$, $\Gamma\sigma\pi \vdash \PP': \T''$ and $\Gamma\sigma\pi \vdash \e: S$. By inversion we have that \RULE{C-Out} was used, and so $\Gamma \vdash \PP': \psi \mid_{\cX_1} \C_1$, $\Gamma \vdash \e: \alpha \mid_{\cX_2} \C_2$, $\cX = \cX_1 \dcup \cX_2 \dcup \{\xi\}$, $\C = \C_1 \cup \C_2 \cup \{\tout{\pp}{\alpha} \psi \subt \xi\}$.

    By the inductive hypothesis, there exist solutions $(\sigma_1, \pi_1)$ of $\C_1$ and $(\sigma_2, \pi_2)$ of $\C_2$ such that $\sigma_1 \minus \cX_1 = \sigma$, $\sigma_2 \minus \cX_2 = \sigma$, $\sigma_1\psi\pi_1 = \T'$, $\pi_2 \alpha = S$. Take $\sigma' = \sigma_1 \cup \sigma_2 \cup \{\xi \mapsto \T_0\}$ and $\pi' = \pi$. $\sigma'$ is a valid mapping because $\sigma_1$ and $\sigma_2$ only differ on $\cX_1$ and $\cX_2$ (which are disjoint due to freshness), and $\xi$ is fresh.

    We have that $\xi\sigma'\pi' = \T_0$, and $\tout \pp{S} \T' = \T_k \subt \T_0$. Therefore $(\sigma', \pi')$ is a solution of $\C$ as desired. Also, $\sigma' \minus \cX = (\sigma_1 \minus \cX_1) \cup (\sigma_2 \minus \cX_2) = \sigma$ and $\pi' \minus \cX = (\pi_1 \minus \cX_1) \cup (\pi_2 \minus \cX_2) = \pi$ by freshness of $\xi, \cX_1, \cX_2$.

    \case \RULE{T-Bra}, \RULE{T-Sel}, \RULE{T-Cond}. Similar to the above.
\end{proof}




\subsection{Appendix for Solving the Minimum Type}
\label{app:solveminimum}



To more succinctly represent the rules for minimum types, we define the following notation, which represents all constraints with type variable $\xi$ on the right hand side:

\begin{definition}\label{def:dependencies}\
Define the \emph{dependencies} of $\set{\xi_i}_{i\in I}$ in $\C$
as $\dep_\C(\set{\xi_i}_{i\in I})=
\{\T_i \mid \T_i \subt \xi_i \in \C, \ i\in I\}$. 
\end{definition}

\begin{definition}\label{def:mapping}\
Let $\cS,\cS_1, \dots$ range over a non-empty set of type variables. Assume  
$\Gmin{\C}$ is a minimum type graph. 
We write $\gnode{\cS}_{\Gmin{\C}}$ for a type representation 
of a type graph $\Gmin{\C}(\cS)$. We often denote 
$\gnode{\cS}$ for $\gnode{\cS}_{\Gmin{\C}}$ if ${\Gmin{\C}}$ is clear from the context. 
\end{definition}

\subsetsubtyping*


  \begin{proof}
We restate the lemma using the notation defined in Definition~\ref{def:mapping}.

\begin{claim}
Let $\Gmin{\C}$ be a minimum type graph. Let $\pi$ be a solution of $\Csort$. If $\cS_1 \subseteq \cS_2 \subseteq \cS_3$ 
where $\cS_3$ is a node in $\Gmin{\C}$, 
then $\gnode{\cS_1}\pi \subt \gnode{\cS_2}\pi$.
\end{claim}

Each type variable in the set $\cS$ specifies some set of behaviours that must be followed in node $\gnode{\cS}$: thus $\gnode{\cS_1}$ is less constrained.  
    Formally, we show that $\sim\: = \{(\gnode{\cS}\pi, \gnode{\cS'}\pi) \mid \cS \subseteq \cS'\} \cup \{(\kf{Skip}, \kf{Skip})\}$ is a type simulation.

    \begin{itemize}[leftmargin=*]
      \item
      If $\T_i = \tin{\pp}{\beta_i} \xi_i$ for all $\T_i$ in $\dep_\C(\cS_1)$,
      then $\gnode{\cS_1}\pi \trans{\trin{\pp}{\alpha_1\pi}} \gnode{\set{\xi_i \mid \tin{\pp}{\beta_i}\xi_i \in \dep_\C(\cS_1)}}\pi$,
      and $\beta_i = \alpha_1 \in \Csort$.
      Then, as $\dep_\C(\cS_1) \subseteq \dep_\C(\cS_2)$,
      we have that
      $\gnode{\cS_2}\pi \trans{\trin{\pp}{\alpha_2\pi}} \gnode{\set{\xi_j \mid \tin{\pp}{\beta_j}\xi_j \in \dep_\C(\cS_2)}}\pi$,
      and $\beta_j = \alpha_2 \in \Csort$.
      Thus we have $\alpha_1\pi = \alpha_2\pi = \beta_j\pi$ for all $\tin{\pp}{\beta_j}\xi_j \in \dep_\C(\cS_2)$.
      Also, $\gnode{\set{\xi_i \mid \T_i \in \dep_\C(\cS_1)}}\pi \sim \gnode{\set{\xi_j \mid \T_j \in \dep_\C(\cS_2)}}\pi$, so the type simulation rules are followed.
      \item $\T_j = \tout{\pp}{\beta_i} \xi_j$ for all $\T_j$ in $\dep_\C(\cS_1)$ is symmetric.

      \item
      If $\T_k = \tbra{\pp}{l_j: \xi_{kj}}_{j \in \J_k}$ for all 
$\T_k$ in $\dep_\C(\cS_2)$,
      then $\gnode{\cS_2}\pi \trans{\trbra{\pp}{l_j}} 
\gnode{\set{\xi_{kj} \mid  \tbra{\pp}{l_j: \xi_{kj}}_{j \in \J_k} 
       \in \dep_\C(\cS_2)}}\pi$.
      Thus $j \in \bigcap_k \J_k$, so
      for all $\T_i$ in $\dep_\C(\cS_1)$, $j \in \J_i$.
      Hence $\gnode{\cS_1}\pi \trans{\trbra{\pp}{l_j}} 
\gnode{\set{\xi_{ij} \mid \tbra{\pp}{l_j: \xi_{ij}}_{j \in \J_i} \in \dep_\C(\cS_1)}}\pi$.
      Also, $\gnode{\set{\xi_j \mid \T_j \in \dep_\C(\cS_1)}}\pi
 \sim \gnode{\set{\xi_i \mid \T_i \in \dep_\C(\cS_2)}}\pi$.

      \item
      If $\T_i = \tsel{\pp}{l_j: \xi_{ij}}_{j \in \J_i}$ for all $\T_i$ in $\dep_C(\cS_1)$,
      then $\gnode{\cS_1}\pi \trans{\trsel{\pp}{l_j}} \gnode{\set{\xi_{ij} \mid \tsel{\pp}{l_j: \xi_{ij}}_{j \in \J_i} \in \dep_\C(\cS_1)}}\pi$.
      Thus $j \in \bigcup \J_i$, so
      there exists $\T_k$ in $\dep_\C(\cS_2)$ such that $j \in \J_k$.
      Hence $\gnode{\cS_2}\pi \trans{\trsel{\pp}{l_j}} \gnode{\set{\xi_{kj} \mid \T_k \in \dep_\C(\cS_2)}}\pi$.
      Also, $\gnode{\set{\xi_j \mid \T_j \in \dep_\C(\cS_1)}}\pi \sim \gnode{\set{\xi_k \mid \tsel{\pp}{l_j: \xi_{ik}}_{j \in \J_k} \in \dep_\C(\cS_2)}}\pi$.

      \item If $\T_i = \tend$ for all $\T_i$ in $\dep_\C(\cS_1)$, then $\gnode{\cS_1}\pi \trans{\trend} \kf{Skip}$. Then there exists $\T_j = \tend$ in $\dep_\C(\cS_2)$, so $\gnode{\cS_2}\pi \trans{\trend} \kf{Skip}$.
    \end{itemize}

    Thus $\sim$ is a type simulation.
  \end{proof}

\theminsound*

\begin{proof}
    All sort constraints are satisfied because $\pi$ solves $\Csort$ which is a superset of the sort constraints of $\C$. We will show that all type constraints are satisfied.

    \case $\tend \subt \xi$.
    Then $\tend\sigma\pi = \tend$. By Definition~\ref{def:minimum_type_graph_of_process},
    $\gnode{\set{\xi}} \trans{\trend}$ so $\xi\sigma\pi = \gnode{\set{\xi}}\pi = \tend$. Thus $\tend\sigma\pi \subt \xi\sigma\pi$.

    \case $\tin{\pp}{\alpha} \xi \leq \psi$.
    Then $(\tin{\pp}{\alpha} \xi)\sigma\pi = \tin{\pp}{\alpha\pi} \gnode{\set{\xi}}\pi \trans{\trin{\pp}{\alpha\pi}} \gnode{\set{\xi}}\pi$.
    We have $\psi\sigma\pi = \gnode{\set{\psi}}\pi \trans{\trin{\pp}{\beta\pi}} \gnode{\cS}\pi$,
    where $\xi \in \cS$ and $\alpha = \beta \in \Csort$, so $\alpha\pi = \beta\pi$.
    Thus, by Lemma \ref{def:subsets_are_subtypes}, $\gnode{\set{\xi}} \subt \gnode{\cS}$.
    So, by definition of subtyping, $(\tin{\pp}{\alpha} \xi)\sigma\pi \subt \psi\sigma\pi$.

    \case $\tout{\pp}{\alpha} \xi \leq \psi$. Symmetric.

    \case $\tselsub{\pp}{l_i: \xi_i}{i \in \I} \leq \psi$.
    Then $(\tselsub{\pp}{l_i: \xi_i}{i \in \I}) \sigma\pi = \tselsub{\pp}{l_i: \gnode{\set{\xi_i}}\pi}{i \in \I} \trans{\trsel{\pp}{l_j}} \gnode{\set{\xi_j}}\pi$,
    We have $\psi\sigma\pi = \gnode{\set{\psi}}\pi \trans{\trsel{\pp}{l_j}} \gnode{\cS}\pi$,
    where $\xi_j \in \cS$ by Definition~\ref{def:minimum_type_graph_of_process}.
    Thus, by Lemma~\ref{def:subsets_are_subtypes}, $\gnode{\set{\xi}} \subt \gnode{\cS}$.
    So, by definition of subtyping, $(\tselsub{\pp}{l_i: \xi_i}{i \in \I}) \sigma\pi \subt \psi\sigma\pi$.

    \case $\tbrasub{\pp}{l_i: \xi_i}{i \in \I} \leq \psi$.
    We have $\psi\sigma\pi = \gnode{\set{\psi}}\pi \trans{\trsel{\pp}{l_j}} \gnode{\cS}\pi$,
    where $\xi_j \subseteq \cS$ by Definition~\ref{def:minimum_type_graph_of_process}.
    Thus, $(\tbrasub{\pp}{l_i: \xi_i}{i \in \I}) \sigma\pi = \tbrasub{\pp}{l_i: \gnode{\set{\xi_i}}\pi}{i \in \I} \trans{\trbra{\pp}{l_j}} \gnode{\set{\xi_j}}\pi$,
    and by Lemma~\ref{def:subsets_are_subtypes}, $\gnode{\set{\xi}} \subt \gnode{\cS}$.
    So, by definition of subtyping, $(\tbrasub{\pp}{l_i: \xi_i}{i \in \I}) \sigma\pi \subt \psi\sigma\pi$.

\end{proof}

We write $\Gmin{\C}\pi$ to be the graph with $\pi$ applied to all transitions.

\begin{definition}[Matching simulation]\label{def:matching_simulation}
Let $\Gmin{\C}$ be a minimmum type graph and $\T$ be a type.
We define a simulation $\psubtpara{\Gmin{\C}}{\T}$
between a type
represented in $\Gmin{\C}$ and $\Sub(\T)$ as the least relation
satisfying the following rules. \\[1mm]
{\small
\mbox{
      \begin{prooftree}
         \hypo{-}
         \infer1[\RULE{$\psubt$-Init}]{\gnode{\{\xi\}}
            \psubtpara{\Gmin{\C}}{\T}
         \T}
        \end{prooftree}
      \
     \begin{prooftree}
       \hypo{\gnode{\cS_1} \trans{\trdag{\pp}{\beta}} \gnode{\cS_2}
       \quad   
       \T_1 \trans{\trdag{\pp}{S}} \T_2
       \quad   
        \gnode{\cS_1}
        \psubtpara{\Gmin{\C}}{\T}
        \T_1
       \quad \dagger\in \{!,?\}}
     \infer1[\RULE{$\psubt$-IO}]{\gnode{\cS_2}
       \psubt_{\Gmin{\C}, \T}
       \T_2}
     \end{prooftree}
}\\[1mm]
     \mbox{
     \begin{prooftree}
        \hypo{\gnode{\cS} \trans{\trend} \Skip}
        \infer1[\RULE{$\psubt$-End}]{\gnode{\cS}
          \psubtpara{\Gmin{\C}}{\T}
          \tend}
           \end{prooftree}
      \quad         
           \begin{prooftree}
           \hypo{\gnode{\cS_1} \trans{\trbradag{\pp}{l}}
               \gnode{\cS_2}
               \quad 
                 \T_1 \trans{\trbradag{\pp}{l}} \T_2
\quad
                 \gnode{\cS_1}
            \psubtpara{\Gmin{\C}}{\T}
            \T_1
             \quad 
             \dagger\in \{\oplus,\&\}}
             \infer1[\RULE{$\psubt$-BS}]{\gnode{\cS_2}
              \psubtpara{\Gmin{\C}}{\T}
             \T_2}
            \end{prooftree}
      }
}
\end{definition}
Our goal is to prove that $\psubt_{\Gmin{\C}, \T}$ is a type simulation, when a solution for $\Csort$ is substituted to the nodes of the minimum type graph.
Here $\gnode{\cS} \psubt_{\Gmin{\C}, \T} \T'$ means ``$\T'$ must respect the behaviours
of all $\psi \in \cS$''. Hence $\gnode{\set{\xi}} \psubt_{\Gmin{\C}, \T} \T$ means that
$\T$ must respect the behaviour of $\xi$, which is the type variable
corresponding to the type of $P$.

\begin{lemma}\label{lem:simsound:a}
  Assume $\vdash \PP: \xi \mid_\cX \C'$ and $\C = \tr(\C')$.
  Let $(\sigma, \pi)$ be a solution to $\C$.
  Let $\Gmin{\C}$ and
  $\Csort$ be the minimum type graph and sort constraint
  generated
  by Definition \ref{def:minimum_type_graph_of_process}.
  Let $\T = \xi\sigma\pi$.
  For all $\cS \in \Gmin{\C}$ and $\T' \in \Sub(\T)$, if $\gnode{\cS}
  \psubt_{\Gmin{\C}, \T} \T'$ then $\psi\sigma\pi
  \subt \T'$ for all $\psi \in \cS$.
\end{lemma}

\begin{proof}
  By fixing $\Gmin{\C}$ and $\T$, we write $\psubt$ for
  $\psubtpara{\Gmin{\C}}{\T}$. 
  
  We prove by induction on the proof trees of $\gnode{\cS} \psubt \T'$.

  \case \RULE{Sim-Init}. We have $\xi\sigma\pi = \T$ by definition.

  \case \RULE{Sim-IO}. We prove the case for the input. The output
  is identical. By the inductive hypothesis we have $\psi\sigma\pi \subt \T_1$ for all $\psi$ in $\cS_1$.
  Let $\psi_2 \in \cS_2$. By definition of $\Gmin{\C}$, there exists $\psi_1 \in \cS_1$ such that $\tin{\pp}{\alpha}\psi_2 \subt \psi_1 \in \C$.
  So $\psi_1\sigma\pi \subt \T_1$ and $\psi_1\sigma\pi \trans{\trin{\pp}{\alpha\pi}} \psi_2\sigma\pi$.
  By definition of subtyping, we conclude that $\psi_2\sigma\pi \subt \T_2$.

  Furthermore, for all $\tin\pp{\alpha^i}{\psi_2^i} \in \dep_\C(\cS_1)$, we have that $\tin\pp{\alpha^i\pi} \psi_2^i\sigma\pi \subt \psi_1\sigma\pi \subt \T_1$ for some $\psi_1 \in \cS_1$.
  Also, $\alpha^i = \beta \in \Csort$ for all $i$.


  \case \RULE{Sim-SB}. We prove the case for the selection. By the inductive hypothesis we have $\psi\sigma\pi \subt \T_1$ for all $\psi \in \cS_1$.
  Let $\psi_2 \in \cS_2$. By definition of $\Gmin{\C}$, there exists $\psi_1 \in \cS_1$ such that $\tsel{\pp}{l: \psi_2} \subt \psi_1 \in \C$.
  Therefore, $\psi_1\sigma\pi \subt \T_1$ and $\psi_1\sigma\pi \trans{\trsel{\pp}{l}} \psi_2\sigma\pi$.
  By definition of subtyping, $\psi_2\sigma\pi \subt \T_2$.


  \case \RULE{Sim-End}. By definition of $\Gmin{\C}$, for $\psi \in \cS$, $\psi \subt \tend \in \C$. Thus $\psi\sigma\pi = \tend$.
\end{proof}

\begin{lemma}\label{lem:simsound:b}
  Assume $\vdash \PP: \xi \mid_\cX \C'$ and $\C = \tr(\C')$.
  Let $(\sigma, \pi)$ be a solution to $\C$.
  Let $\Gmin{\C}$ and
  $\Csort$ be the minimum type graph and sort constraint
  generated
  by Definition \ref{def:minimum_type_graph_of_process}.
  Let $\T = \xi\sigma\pi$.
  Assume $\pi'$ is the most general solution of
  $\Csort$. 
  Then there exists $\pi''$ such that $\psubts \:= \{(\gnode{\cS}\pi'\pi'', \T') \mid \gnode{\cS} \psubt_{\Gmin{\C}, \T} \T'\}$ is a type simulation.
\end{lemma}
\begin{proof}
  As $\pi'$ is a most general solution on $\Csort$, there exists $\pi''$ such that $\alpha\pi = \alpha\pi'\pi''$ for all $\alpha \in \Sv$.
  Let $\gnode{\cS}\pi'\pi'' \psubts \T'$, so $\gnode{\cS} \psubt \T'$. We check the conditions for type simulation:
  
  \begin{itemize}
    \item If $\gnode{\cS} \trans{\trend}$, choose an arbitrary $\psi \in \cS$. Then by Definition \ref{def:minimum_type_graph_of_process}, $\tend \subt \psi \in \C$.
    By Lemma \ref{lem:simsound:a}, $\psi\sigma\pi \subt \T'$. Therefore, $\tend \subt \T'$, i.e. $\T' \trans{\trend}$.

    \item If $\gnode{\cS} \trans{\trin{\pp}{\beta}} \gnode{\cS_1}$, then by Definition \ref{def:minimum_type_graph_of_process},
    $\tin{\pp}{\alpha_i}\xi_i \subt \psi_i \in \C$ for all $\psi_i \in \cS_1$,
    with $\beta = \alpha_i \in \C$.
    By Lemma \ref{lem:simsound:a}, $\psi_i\sigma\pi \subt \T'$.
    Therefore, $\T' \trans{\trin{\pp}{\S}} \T_1$ for some $\T_1$, where $\S = \alpha_i\pi = \alpha_i\pi'\pi'' = \beta\pi'\pi''$
    (because $\pi'$ solves $\Csort$, which contains all sort constraints of $\C$).
    By definition of $\psubt$, $\gnode{\cS_1}\pi'\pi'' \psubts \T_1$.

    \item If $\gnode{\cS} \trans{\trout{\pp}{\beta}} \gnode{\cS_1}$: symmetric.

    \item If $\T' \trans{\trbra{\pp}{l_k}} \T_1$, then by Lemma \ref{lem:simsound},
    $\psi\sigma\pi \subt \T'$ for all $\psi \in \cS$. Therefore, by definition of subtyping,
    $\psi\sigma\pi \trans{\trbra{\pp}{l_k}}$, so
    $ \gnode{\set{\psi}}\sigma\pi \trans{\trbra{\pp}{l_k}}$. By definition of $\Gmin{\C}$, $\gnode{\cS} \trans{\trbra{\pp}{l_k}} \gnode{\cS_1}$ for some $\cS_1$, therefore $\gnode{\cS_1}\pi'\pi'' \psubts \T_1$.

    \item If $\gnode{\cS} \trans{\trsel{\pp}{l_k}} \gnode{\cS_1}$, then by Definition \ref{def:minimum_type_graph_of_process}, there exists $\psi \in \cS_1$ such that $\tselsub{\pp}{l_j: \xi_{j}}{j \in \J} \subt \psi \in \C$ such that $k \in \J$ and $\xi_k = \psi'$.
    Then by Lemma \ref{lem:simsound:a}, $\psi\sigma\pi \subt \T'$. Therefore, $\T' \trans{\trsel{\pp}{l_k}} \T_1$ for some $\T_1$.
    By definition of $\psubt$, $\gnode{\cS_1}\pi'\pi'' \psubts \T_1$.
  \end{itemize}

  All conditions are valid and so $\psubts$ is a type simulation.
\end{proof}

The above two lemmas are the two statements of
Lemma~\ref{lem:simsound} below:

\begin{restatable}[Simulation soundness]{lemma}{lemsimsound} \label{lem:simsound}
Assume $\vdash \PP: \xi \mid_\cX \C'$ and $\C = \tr(\C')$.
Let $(\sigma, \pi)$ be a solution to $\C$.
Let $\Gmin{\C}$ and
$\Csort$ be the minimum type graph and sort constraint
generated
by Definition \ref{def:minimum_type_graph_of_process}.
Let $\T = \xi\sigma\pi$.
{\rm (a)} For all $\cS \in \Gmin{\C}$ and $\T' \in \Sub(\T)$, if $\gnode{\cS} \psubt_{\Gmin{\C}, \T} \T'$
then $\psi\sigma\pi \subt \T'$ for all $\psi \in \cS$; and 
{\rm (b)} There exists $\pi''$ such that
$\psubts\ = \{(\gnode{\cS}\pi'\pi'', \T') \mid \gnode{\cS} \psubt_{\Gmin{\C}, \T} \T'\}$
is a type simulation.
\end{restatable}
\begin{proof}
  (a) By induction on the proof tree of
  $\gnode{\cS} \psubtpara{\Gmin{\C}}{\T} \T'$. \\
  (b) By verifying that the type simulation rules are satisfied, using (a).
\end{proof}

\themincomplete*

\begin{proof}
 $\gnode{\set{\xi}} \psubtpara{\Gmin{\C}}{\T} \T$ by \RULE{$\psubt$-Init}, so by Lemma~\ref{lem:simsound}(b), $\gnode{\set{\xi}}\pi'\pi'' \psubts \T$.
  We have that $\gnode{\set{\xi}}\pi' = \Tmin$,  and $\psubts$ is a type simulation, thus
  $\Tmin\pi'' \subt \T$.
\end{proof}

\section{Appendix for Checking Safety, Deadlock-Freedom and Liveness}
\label{app:model_checking}
\subsection{Detailed Examples for Remark~\ref{rem:liveness}}
\label{app:rem:liveness}
Let us call 
$\Delta$ \emph{weak live} 
defined by Definition~\ref{def:liveness} \emph{without} fair path assumptions. 
We show the weak-live context can type a \emph{non-live} session $\M$. 
Consider the following context 
with three parties, $\role{a}$lice, $\role{b}$ob and $\role{s}$eller
(adapted from \cite[Example~5.14]{Scalas2019}):\\[1mm]
\centerline{
$
\begin{array}{rcl}
\Delta' =  
\role{a}:
\mu \ty.\tsel{\role{b}}{\mathsf{m}:
\tbra{\role{b}}{\mathsf{yes}:\tselset{\role{s}}{\mathsf{buy}},
\ \mathsf{no}:\ty}}, \ 
\role{b}: 
\mu \ty.\tbra{\role{a}}{\mathsf{m}:\tselset{\role{a}}{\mathsf{no}:\ty},\ 
                        \mathsf{cancel}}, 
\role{s}: 
\tbra{\role{a}}{\mathsf{buy},\ \mathsf{no}},
\end{array}
$
}\\[1mm]
$\Delta'$ is neither live nor weak live. 
This is because that $\role{a}$lice never chooses to send $\mathsf{cancel}$ to $\role{b}$ob, who in turn always answers $\mathsf{no}$
to all $\mathsf{m}$essages. 
The session $\M$ which is typed by $\Delta'$ is not live either. Now consider 
typing environment $\Delta$ such that $\Delta'\subt \Delta$,  
which can be obtained by adding choices 
$\mathsf{cancel}$ at $\role{a}$lice and 
$\mathsf{yes}$ at $\role{b}$ob 
({\color{violet}colored}). We let $\Delta=\Delta'',\Delta'(\role{s}) $
\\[1mm]
\centerline{
$
\begin{array}{rcl}
\Delta''=
\role{a}:
\mu \ty.\tsel{\role{b}}{\mathsf{m}:
\tbra{\role{b}}{\mathsf{yes}:\tselset{\role{s}}{\mathsf{buy}},
\ \mathsf{no}:\ty}, \ {\color{violet}\mathsf{cancel}:
\tselset{\role{s}}{\mathsf{no}}}}, 
\role{b}: 
\mu \ty.\tbra{\role{a}}{\mathsf{m}:\tsel{\role{a}}{{\color{violet}\mathsf{yes}}, \ \mathsf{no}:\ty},\ 
                        \mathsf{cancel}}
\end{array}
$}\\[1mm]
Session $\M$ typed by 
$\Delta'$ is also typable by $\Delta$ 
using the subsumption rule (\RULE{T-Sub}).   
$\Delta$ is weak live but not live since there is a fair path 
$\Delta\redSelBra{\role{a}}{\role{b}}{\mathsf{m}}
\redSelBra{\role{b}}{\role{a}}{\mathsf{no}}
\redSelBra{\role{a}}{\role{b}}{\mathsf{m}}\redSelBra{\role{b}}{\role{a}}{\mathsf{no}}\cdots$ under which $\role{s}$eller  
indefinitely awaits selection from $\role{a}$lice.  
Hence weak live $\Delta$ types non-live $\M$. 

\subsection{Proof of Lemma~\ref{thm:safety_from_reductions}}
\label{app:proof_of_safety_from_reductions}

\safetyfromreductions*
\begin{proof}
($\Longrightarrow$) Fix a safety property $\phi$ such that $\phi(\Delta)$. By definition, $\Delta$ is a safe state.
Let $\Delta \rightarrow^* \Delta'$.
By induction on the length of the reduction sequence, then applying the second condition of safety properties, we have that $\Delta'$ is a safe state.

  ($\Longleftarrow$) Define $\Delta', \mu\ty. \T \trans{\unf} \Delta', \T[\mu\ty. \T/\ty]$ for all $\Delta', \T$.
  Note that, if $\Delta_1 \trans{\unf} \Delta_2$ then $\Delta_1 \trans{\cxell} \Delta' \iff \Delta_2 \trans{\cxell} \Delta'$.
  Let $\Delta'$ be a safe state for all $\Delta \rightarrow^* \Delta'$.
  Let $\phi$ be the minimum relation containing $\Delta$ and closed under $\rightarrow$ and $\trans{\unf}$.

  By definition, the latter two conditions for safety hold. It suffices to show that $\Delta'$ is safe for all $\phi(\Delta')$.
  Consider the shortest sequence of $\rightarrow$ and $\trans{\unf}$ reductions from $\Delta$ to $\Delta'$.
  If $\Delta_1\trans{\unf}^*\rightarrow \Delta_2$ then $\Delta_1 \rightarrow \Delta_2$, so the shortest sequence of reductions has the form $\Delta \rightarrow^* \Delta'' \trans{\unf}^* \Delta'$.
  But $\Delta''$ is a safe state by the claim, and $\Delta'', \Delta'$ have the same outgoing transitions. So $\Delta'$ is a safe state. Therefore $\phi$ is a safety property.
\end{proof}

\subsection{Checking Safety and Deadlock-Freedom is in PSPACE}
\label{app:pspace_completeness_of_safety_and_deadlock_freedom}

\begin{lemma}\label{thm:composition_transition_as_type_graph}
  Let $\Delta = \prod_{i \in \I} \pp_i : \T_i$. Let $\Delta \rightarrow^* \Delta'$, such that $\Delta' = \prod_{i \in \I} \pp_i : \T'_i$. Then $\T'_i \in \Sub(\T_i)$.
\end{lemma}
\begin{proof}
  By inspection of the typing rules of Definition~\ref{def:composition_reduction}, the possible reductions of $\pp_i : \T_i$ are exactly the possible reductions of $\GG(\T_i)$. By Lemma~\ref{thm:type_graph_nodes_are_subformulas}, this is a subset of $\Sub(\T_i)$.
\end{proof}

\numberofreachablecompositionstates*
\begin{proof}
  By the above lemma, applied to each $\pp_i : \T_i$. We have that $\{\Delta' \mid \Delta \rightarrow^* \Delta'\} \subseteq \{\prod_{i \in \I} \pp_i : \T'_i \mid \T'_i \in \Sub(\T_i)\}$, so $|\{\Delta' \mid \Delta \rightarrow^* \Delta'\}| \leq \prod_{i \in \I} |\T_i|$.
\end{proof}

\begin{algorithm}[H]
\footnotesize
\caption{\footnotesize A nondeterministic polynomial-space algorithm for checking if a typing context is not safe. The algorithm accepts if at least one combination of nondeterministic choices results in an accepting state. The input is a context $\Delta = \prod_{i \in \I} \pp_i : \T_i$.}
   \label{alg:pspace_safety}
   \begin{algorithmic}[1]
     \Function{NotSafe}{$\Delta$}
       \State {$m \gets \prod_{i \in \I} |\T_i|$}
       \State {$n \gets 0$}
       \State {$\Delta' \gets \Delta$}
       \Loop
         \If {$\Delta'$ is not a safe state} \label{alg:pspace_safety_if_cond}
           \State {\Return accept}
         \EndIf
         \State {$n \gets n + 1$}
         \If {$n > m$}
           \State {\Return reject}
         \EndIf
         \State {nondeterministically choose $\Delta''$ such that $\Delta' \rightarrow \Delta''$}
         \State {$\Delta' \gets \Delta''$}
       \EndLoop
     \EndFunction
   \end{algorithmic}
\end{algorithm}

\begin{theorem}\label{thm:safety_in_pspace}
  Checking for safety is in \PSPACE.
\end{theorem}
\begin{proof}
We use Algorithm~\ref{alg:pspace_safety}, which is a nondeterministic polynomial-space algorithm that accepts if and only if the input is not safe. The algorithm checks all $\Delta'$ such that $\Delta'$ is reachable from $\Delta$ in $M$ transitions or fewer, where $m$ is given in Lemma~\ref{thm:number_of_reachable_composition_states}, accepting if there exists some $\Delta'$ that is not a safe state. As the shortest path to any state has length at most $m$, the algorithm is correct.
Furthermore, the algorithm needs to store only the current context $\Delta'$, which is representable by type graphs in linear space, and the counter $n$, which has at most $\log{N} = \bigO{\sum_{i \in I}|\T_i|}$ bits, so the algorithm is polynomial-space.



The result follows from $\coNPSPACE = \NPSPACE$ and from Savitch's theorem: $\NPSPACE = \PSPACE$ \cite[Theorem 8.5]{Sipser2012}.
\end{proof}

We can use a similar algorithm for deadlock-freedom.

\begin{theorem}
  Checking for deadlock-freedom is in \PSPACE.
\end{theorem}
\begin{proof}
The proof is identical to the proof for safety, but with the condition in line~\ref{alg:pspace_safety_if_cond} replaced with $\Delta' \not\rightarrow \implies \forall \pp: \T_\pp \in \Delta'. \unfold{\T_\pp} = \tend$, which is also checkable in polynomial time.

\end{proof}

\subsubsection{A More Practical Algorithm for Safety and Deadlock-Freedom}
\label{subsec:safedf:efficient}
Even though there must exist a deterministic polynomial-space algorithm for checking safety 
and deadlock-freedom by the equivalence of $\coNPSPACE$ and $\PSPACE$, this might not be
practical due to the high time complexity of the reduction.
A more practical deterministic algorithm would do an explicit breadth-first search over all reachable states.
Even though this requires exponential memory in the worst case (so this is not a $\PSPACE$-algorithm),
it takes linear time and memory in the number of reachable states.
In practice, this number is often much smaller than the theoretical bound of $\prod_{i \in \I} |\T_i|$.
This can also be adapted for deadlock-freedom in the same way as the other algorithm.

We instead illustrate a
deterministic algorithm that might be more practical, which instead does an explicit breadth-first search
over all reachable states (Algorithm~\ref{alg:safety_more_practical}).

\begin{algorithm}[H]
\footnotesize
\caption{\footnotesize A deterministic algorithm for checking if a typing context is safe. The input is a context $\Delta = \prod_{i \in \I} \pp_i : \T_i$.}
  \label{alg:safety_more_practical}
  \begin{algorithmic}[1]
    \Function{Safe}{$\Delta$}
      \State {$V \gets \{\Delta\}$} \Comment{Set of found states}
      \State {$Q \gets [\Delta]$} \Comment{Queue of states to visit}
      \While {$Q$ is nonempty}
        \State {$\Delta', Q \gets Q$}
         \If {$\Delta'$ is not a safe state}
           \State {\Return reject}
         \EndIf
         \State {$V \gets V \cup \{\Delta'\}$}
         \ForAll {$\Delta''$ such that $\Delta' \rightarrow \Delta''$}
           \If {$\Delta'' \not\in V$}
             \State {$V \gets V \cup \Delta''$}
             \State {$Q \gets Q, \Delta''$}
           \EndIf
         \EndFor
       \EndWhile
       \State {\Return accept}
     \EndFunction
   \end{algorithmic}
\end{algorithm}


\safedfcomplexityefficient*

\begin{proof}
  Lemma~\ref{thm:number_of_reachable_composition_states} shows that there are at most $\U$ reachable states from $\Delta$.
  The algorithm visits each state at most once, and for each state, there are at most $\Delta$ outgoing transitions,
  so it checks if it is safe (or deadlock-free) in $\bigO{|\Delta|}$ time.
\end{proof}

\subsection{Checking Liveness is in PSPACE}
\label{app:pspace_completness_of_liveness}


\counterwitnessiffnotlive*
\begin{proof}
  In Definition~\ref{def:counterwitness}, the first condition holds when the path is not live Lemma~\ref{lem:live_iff_barbs_subset_obs}. The second condition holds when the path is fair.
\end{proof}

\thmnumberofgammatransitions*
\begin{proof}
  By considering the type graphs of $\T_i$.
  \begin{itemize}[leftmargin=*]
    \item If $\cxell = \pp_i\pp_j$, then $\GG(\T_{\pp_i})$ must contain an edge $\trout{\pp_j}{S}$, and $\GG(\T_{\pp_j})$ must contain an edge $\trin{\pp_i}{S}$ for some $S$.
    \item If $\pp_i\pp_j:l$, then $\GG(\T_{\pp_i})$ must contain an edge $\trsel{\pp_j}{l}$, and $\GG(\T_{\pp_j})$ must contain an edge $\trbra{\pp_i}{l}$.
  \end{itemize}
  As each reduction gives rise to unique edges, and there are less than or equal to $2|\T|$ edges in a type graph of $\T$~\cite[Lemma 4.3]{Udomsrirungruang2024a}, we conclude that there are $2n$ possible unique reductions. 
\end{proof}

The following lemma shows that there must exist finite or periodic counterwitnesses to liveness if the context is not live. We write $\mathcal{P}_1\mathcal{P}_2^\omega$ to denote an infinite path that concatenates $\mathcal{P}_1$ and infinitely many copies of $\mathcal{P}_2$.

\witnesslengthforliveness*
\begin{proof}
  Fix some counterwitness $(\Delta_i)_{i \in N}$ of $\Delta$. By definition, there exists $k \in N$ and $a \in \barbs(\Delta_k)$ such that $a \not\in \bigcup_{i \geq k} \observations(\Delta_i)$.
  
  \case $N$ is finite. Then we can \textit{shortcut} some cycles:
  if $\{\cxell_i \mid i \geq j_1\} = \{\cxell_i \mid i \geq v\}$ such that $j_1 < j_2$ and $k \not\in \{j_1, j_1+1, \dots, j_2\}$
  then we may remove $\Delta_{j_1+1}, \dots, \Delta_{j_2}$ from the path.
  As $\{\cxell_i \mid i \geq w\} \subseteq \{\cxell \mid \exists i \geq w.\:\Delta_i \trans{\cxell}\}$ for all $w$, the subsets stay equal,
  and the first condition is not violated because we are removing elements from $\observations$.
  We can continue this process until there are no cycles with the same value of $|\{\cxell_i \mid i \geq j\}|$ throughout.
  As there are only $2n$ possible values of $\cxell_i$ by Lemma~\ref{thm:number_of_gamma_transitions}, and we cannot shortcut cycles containing $k$,
  there are at most $2n+2$ sections, each with no repeating states,
  so the resulting path is of size $m$.

  \case $N$ is infinite. We first need to make the path periodic.
  This can be done by taking the first $u \geq k$ such that $\{\cxell_i \mid k \leq i \leq u\} = \{\cxell_i \mid k \leq i\}$.
  Then, take the first $u' > u$ such that $\Delta_{u'} = \Delta_v$ for some $v < u'$.
  For any $j$, $\{\cxell_i \mid j \leq i \leq u'\} = \{\cxell_i \mid j \leq i\} = \{\cxell \mid \exists j \leq i.\:\Delta_i \trans{\cxell}\} \supseteq \{\cxell \mid \exists j \leq i \leq u'.\:\Delta_i \trans{\cxell}\}$.
  But $\{\cxell_i \mid j \leq i \leq u'\} \subseteq 
  \{\cxell \mid \exists j \leq i \leq u'.\:\Delta_i \trans{\cxell}\}$ by definition, so the two sets are equal.
  Also, $\bigcup_{k \leq i \leq u'} \observations(\Delta_i) \subseteq \bigcup_{k \leq i} \observations(\Delta_i)$.

  Thus, $\Delta' = \mathcal{P}_1\mathcal{P}_2^\omega$ such that $\mathcal{P}_1 = \Delta_0, \dots, \Delta_{v-1}$
  and $\mathcal{P}_2=\Delta_v, \dots, \Delta_{u'-1}$ is a valid counterwitness.
  Then, we will shortcut cycles similarly to the finite case.
  Before $\Delta_v$, the strategy for finite paths still works.
  For the cycle $(\Delta_v, \dots, \Delta_{u'-1})$, we will shortcut $j_1 < j_2$
  such that $\{\cxell_i \mid j_1 \leq i < u'\} = \{\cxell_i \mid j_2 \leq i < u'\}$,
  by removing $\Delta_{j_1+1}, \dots, \Delta_{j_2}$.
  As all the removed transitions in $\{j_1+1, \dots, j_2\}$ are still found in $\{j_2, \dots, u'-1\}$,
  the second condition is not violated.
  Furthermore, the first condition is not violated by the same reason as above.
  
  Therefore, by the same reasoning as in the finite case, the resulting path is of size $m$.
\end{proof}

We can use this bound to create a polynomial-space search for exponential-length counterwitnesses.

\begin{algorithm}
\footnotesize
  \caption{\footnotesize A nondeterministic polynomial-space algorithm for checking if a typing context is not live. The algorithm accepts if at least one combination of nondeterministic choices results in an accepting state. The input is a context $\Delta = \prod_{i \in \I} \pp_i : \T_i$.}
   \label{alg:pspace_liveness}
   \begin{algorithmic}[1]
     \Function{NotLive}{$\Delta, n$}
       \State {$M \gets (2n+2) \cdot \prod_{i \in \I}|\T_i|$}
       \State {$U \gets \emptyset$} \Comment{unfair transitions}
       \State {$T \gets \emptyset$} \Comment{observed transitions in cycle}
       \State {$B \gets \emptyset$} \Comment{barbs of $\Delta_k$}
       \State {$O \gets \emptyset$} \Comment{observed barbs}
       \State $\Delta_0 \gets \Delta$
       \For {$\Delta', \kf{prev}(\Delta')$ in some nondeterministically chosen path $\mathcal{P}_1$ of length $\leq M$ starting at $\Delta_0$}
         \State $U \gets U \cup \{\cxell \mid \Delta' \trans{\cxell}\}$
         \State $U \gets U \minus \{\cxell'\}$ where $\kf{prev}(\Delta') \trans{\cxell'} \Delta'$
         \If {$\Delta_k$ has not been set}
           \State{nondeterministically choose to set $\Delta_k \gets \Delta'$ and $B \gets \barbs(\Delta')$}
         \Else
           \State $O \gets O \cup \observations(\kf{prev}(\Delta'), \Delta')$
         \EndIf
       \EndFor
       \If {$\Delta_k$ is set and $U = \emptyset$ and $B \minus O \neq \emptyset$} 
         \Return accept \Comment{finite case}
       \EndIf
       \For {$\kf{prev}(\Delta'), \Delta'$ in some nondeterministically chosen nonempty path $\mathcal{P}_2$ of length $\leq M$ starting at the last element of $\mathcal{P}_1$}
         \State $T \gets T \cup \{\cxell'\}$ where $\kf{prev}(\Delta') \trans{\cxell'} \Delta'$
         \If {$\Delta_k$ has not been set}
           \State{nondeterministically choose to set $\Delta_k \gets \Delta'$ and $B \gets \barbs(\Delta')$}
         \Else
           \State $O \gets O \cup \observations(\kf{prev}(\Delta'), \Delta')$
         \EndIf
         \State $U \gets U \cup \{\cxell \mid \Delta' \trans{\cxell}\}$
       \EndFor
       \If {$\mathcal{P}_2$ is not a cycle}
         \Return reject
       \EndIf
       \If {$\Delta_k$ is set and $U \minus T = \emptyset$ and $B \minus O \neq \emptyset$}
         \Return accept \Comment{cycle case}
       \Else\@
         \Return reject
       \EndIf
     \EndFunction
   \end{algorithmic}
 \end{algorithm}

\livenesspspace*

  


\begin{proof}
   We use Algorithm~\ref{alg:pspace_liveness}, which is a nondeterministic polynomial-space algorithm that accepts if and only if the input is not live. At a high level, the algorithm nondeterministically generates periodic paths of bounded length (Theorem~\ref{thm:witness_length_for_liveness}), and checks if they are counterwitnesses. $\Delta_k$ is nondeterministically chosen as the paths are generated.

   At the end of the first loop, $U$ is the set of reductions $\cxell$ such that $\Delta_i \trans{\cxell}$ but there is no $\Delta_j \in \mathcal{P}_1$, $j > i$ such that $\Delta_j \trans{\cxell} \Delta_{j+1}$. $B = \barbs(\Delta_k)$ and $O = \observations((\Delta_i)_{i \geq k})$ in $\mathcal{P}_1$.
   Thus we are checking if $\mathcal{P}_1$ is a counterwitness with the given value of $k$.

   At the end of the second loop, we add all reductions $\Delta_i \trans{\cxell}$ to $U$ and all observed transitions to $T$. After the loop, we check that all transitions were observed and not all barbs were matched, so we are checking for a counterwitness for $\mathcal{P}_1\mathcal{P}_2^\omega$. Thus the algorithm is correct.

   Furthermore, the algorithm needs to store only four polynomial-size sets $U, T, B, O$, as well as $\Delta_k$, and pointers to the start and end of each path, which are polynomial in size. In particular, the paths $\mathcal{P}_1$ and $\mathcal{P}_2$ do not need to be stored explicitly.
Thus the algorithm is polynomial-space.
\end{proof}

\subsubsection{A Deterministic Algorithm for Checking Liveness}
\label{subsec:live:efficient}

We can use a similar construction to the proof of Savitch's theorem \cite[Theorem 8.5]{Sipser2012}
to create a deterministic algorithm for checking liveness.

\livecomplexityefficient*

\begin{proof}
  Define the predicate
  $\Reach_t^a(\Delta', \Delta'', T, U)$,
  for barb $a$, $t \in \mathbb{N}$, $\Delta', \Delta''$ contexts,
  to be true iff there exists a path of length at most $2^t$ from $\Delta'$ to $\Delta''$ such that barb $a$ is never observed,
  the set of observed transitions is $T$,
  the set of "unfair transitions" (Algorithm~\ref{alg:pspace_liveness}) is $U$.

  We can compute $\Reach_t^a(\Delta', \Delta'', T, U)$ for all possible $\Delta', \Delta'', T$ and $U$ as follows:

  \begin{itemize}
    \item $\Reach_0^a$ can be computed for each possible transition. If $\Delta' \trans{\cxell} \Delta''$,
    then $\Reach_t^a(\Delta', \Delta'', T, U)$ iff $a \notin \barbs(\Delta)$ and $T = \{\cxell\}$, and $U = \{\cxell' | \Delta' \trans{\cxell'} \text{ and } \cxell' \neq \cxell\}$,
    and $\Reach_t^a(\Delta', \Delta', \emptyset, \emptyset)$.
    This can be done in time $\bigO{\U^2 \cdot 2^{|\Delta|}}$.
    \item $\Reach_{t+1}(\Delta', \Delta'', T_1 \cup T_2, U')$ is true iff there exists some $\Delta^*$ such that $\Reach_t(\Delta', \Delta^*, T_1, U_1)$ and $\Reach_0(\Delta^*, \Delta'', T_2, U_2)$, such that $U' = U_2 \cup (U_1 \minus T_2)$.
    This can be done in time $\bigO{\U^3 \cdot 2^{2|\Delta|}}$.
  \end{itemize}
  Let $m = (2 |\Delta|+2) \cdot \U$ be the longest possible length of a counterwitness, defined in Lemma~\ref{thm:witness_length_for_liveness}.
  Thus, we can compute $\Reach_m(\Delta', \Delta'', T, U)$ in time $\bigO{\U^3 \log \U \cdot 4^{|\Delta|}}$, where $t = \lceil \log m \rceil$.

  Similarly, we can define $\CReach_t^a(\Delta', \Delta'', T, U)$ to be true iff
  there exists a path $\Delta' = \Delta_0 \trans{\cxell_1} \Delta_1 \dots \trans{\cxell_n} \Delta_n = \Delta''$
  of length $n \leq 2^t$ such that barb $a$ is never observed, and
  $T = \{\cxell_1, \dots, \cxell_n\}$, and $U = \{\cxell \mid \Delta_i \trans{\cxell}\}$.
  We can compute $\CReach_m(\Delta', \Delta'', T, U)$ in time $\bigO{\U^3 \log \U \cdot 2^{2 |\Delta|}}$.

  Then, a counterwitness exists if there exists some $a, \Delta_1, \Delta_2, \Delta_3$
  such that $\Delta_1$ is reachable from $\Delta$, $a \in \barbs(\Delta_1)$,
  and $\Reach_t^a(\Delta_1, \Delta_2, T_1, U_1)$ and $\CReach_t^a(\Delta_2, \Delta_3, T_2, U_2)$ such that $U_1 \cup U_2 \subseteq T_1$,
  or $\Reach_t^a(\Delta_1, \Delta_2, T_1, \emptyset)$.

  This can be done in time $\bigO{\U^3 \cdot 2^{3 |\Delta|}}$ time.
\end{proof}

\subsection{Reduction from the Quantified Boolean Formula Problem}
\label{app:reduction_from_qbf}

Our construction is a typing context $\Deltainit$ that effectively computes the truth value of $\bqbf$. It contains one participant for each variable ($\pp_i$) and clause ($\pr_i$), as well as the ``controller'' ($\ps$).
Due to the constraint that selection and branching must involve a
single other participant, we build the context in a way that
allows communication in a linear topology:\\[1mm]
\centerline{
$\role{s} \ \longleftrightarrow \ \pp_1 \ \longleftrightarrow \ \pp_2
  \ \longleftrightarrow \cdots \longleftrightarrow \ \pp_n
  \ \longleftrightarrow \ \role{r}_1
    \ \longleftrightarrow \cdots \longleftrightarrow \ \role{r}_{m+1} 
  $ 
}

\begin{itemize}[leftmargin=*]
\item $\ps$ is the ``controller'' of the session: it queries the other participants to check if the QBF is true; if it is not true, it enters the undesirable state $\Tbad$, which we will define separately for each property we check. Otherwise, the QBF is queried in an infinite loop.
  \item $\pp_i$ handles the truth value of $v_i$, and finds the truth value of $\cQ_i v_i \dots \cQ_n v_n \sqbf$.
  It does this by trying both values of $v_i$:
  first, it lets $v_i = \tfalse$ (represented by being in
  state $F$ in Figure~\ref{fig:pspace_safety_graphs})
  and asks $\pp_{i+1}$ for the truth value of $\cQ_{i+1} v_{i+1} \dots \cQ_n v_n$.
  During this, $\pp_i$ is able to respond to requests with label $\kfquery_{\rolem{t}}$,
  which are querying the truth value of the variable represented by participant $\rolem{t}$. 
  If needed, it then lets $v_i = \ttrue$ (represented by being in state $T$ in Figure~\ref{fig:pspace_safety_graphs})
  and repeats the process.
  The responses from $\pp_{i+1}$ are then combined (depending on $\cQ_i$) into a response sent to $\pp_{i-1}$: either $\kfdoneyes$ or $\kfdoneno$.
  \item Once all $\pp_i$ have their truth values initialised,
  $\pr_i$ handles the truth value of the $i$-th clause of $\sqbf$,
  namely $L_{i1} \vee L_{i2} \vee L_{i3}$.
  First, the participant queries the truth value from $\pr_{i+1}$.
  If this is false, then it immediately
  returns $\kfdoneno$ to the previous participant.
  Otherwise, it queries (using $\kfquery_{\rolem{t}}$) the truth value of
  each of its literals in turn, and returns the appropriate value.
  An exception is $\pr_{m+1}$, which starts the process by returning $\kfdoneno$.
\end{itemize}

\paragraph*{\bf Overview.}
We reduce from the quantified Boolean formula (QBF) problem: let 
\[\bqbf = \cQ_1 v_1 \dots \cQ_n v_n \sqbf,\  \cQ_i \in \{\exists, \forall\}
\]
be a QBF, with $\sqbf$ in conjunctive normal form: $\sqbf := \bigwedge_{i=1}^{m} C_i,\ C_i = (L_{i1} \vee L_{i2} \vee L_{i3})$. Each $L_{ij}$ is a literal, which is either a variable $v_i$ or its negation $\neg v_i$. The problem of checking whether QBFs are true is well-known to be \PSPACE-complete \cite[Theorem 8.9]{Sipser2012}.

Let the set of participants be~$\mathbb{P} = \{\pp_i \mid i \in \{1, \dots, n\}\} \cup \{\pr_i \mid i \in \{1, \dots, m+1\}\} \cup \{\ps\}$. For convenience, define $\pp_0 = \ps$, $\pp_{n+1} = \pr_1$, $\pr_0 = \pp_n$.
Construct the session types in Figure~\ref{fig:pspace_safety}: let their context be~$\Deltainit = \prod_{\pp \in \mathbb{P}} \pp : \T_\pp$. $\Tbad$ is an arbitrary type, which we will set to different types depending on the property we are checking.

We draw the type graphs of each participant in the context in Figure \ref{fig:pspace_safety_graphs}. We label the nodes of the type graphs to more clearly prove the reduction, and we will use the names of graph nodes and their corresponding types interchangeably.

\input{qbf_reduction_fig}

\paragraph*{\bf Notation.}
We use the following notation to refer to the reachable contexts of the participants. Let $(\ps : \T^*_{\ps}, \pp_1 : \T^*_{\pp_1}, \dots, \pp_n : \T^*_{\pp_n}, \pr_1 : \T^*_{\pr_1}, \dots, \pr_{m+1} : \T^*_{\pr_{m+1}})$ refer to the context $\prod_{\pp\in\mathbb{P}} \pp : \T^*_{\pp}$ (this is the same order as the chain of communications: each type communicates with only the types directly left and right of it). We may omit some or all of the prefixes $\pp : $ if it is clear from the context.

For an assignment of truth values $\cA: \{v_1, \dots, v_k\} \rightarrow \{\tfalse, \ttrue\}$ with $k \leq n$, define $V^i = T$ if $\cA(v_i) = \ttrue$, and $V^i = F$ otherwise. Define $(\T^*_{\ps}, \cA, \T_{\pp_{k+1}}, \dots, \T_{\pr_n}) = (\T^*_{\ps}, V^1, \dots, V^k, \T_{\pp_{k+1}}, \dots, \T_{\pr_n})$. Define $[\cA] = (W, \cA, I_1, I, \dots, I)$. $\Deltainit = (I, \dots, I)$ is the starting context.

\subsubsection{Reduction Behaviour of the Composition}

\begin{definition}[Deterministic, safe reduction paths]\label{def:unique_safe_reduction_paths}
The reduction  $\Delta \rightarrow \Delta'$ is \emph{deterministic} 
if $\Delta \rightarrow \Delta_1$ and $\Delta \rightarrow \Delta_2$ 
imply $\Delta_1 = \Delta_2 = \Delta'$. We write 
$\Delta \rightarrowunique  \Delta'$ if 
$\Delta \rightarrow  \Delta'$ is deterministic. 
For safe states $\Delta$, define $\Delta \rightarrowsafeunique \Delta'$ if $\Delta \rightarrowunique \Delta'$.

  We write $\rightarrowuniques$ for the reflexive and transitive closure of $\rightarrowunique$. We write $\Delta \rightarrowsafeuniques \Delta'$ if $\Delta$ is a safe state and $\Delta'$ can be reached from $\Delta$ by a sequence of $\rightarrowsafeunique$ reductions.
\end{definition}

We will now describe the behaviour of the context in Figure~\ref{fig:pspace_safety}, primarily to show that it has a unique and safe reduction path that calculates the truth value of the QBF.

The following lemma states that processes handling clauses can ``query'' the truth value of variables by selecting the label $\kfquery_{\pp_i}$ and will get the truth value in the form of a selection of $\kfyes$ or $\kfno$.

\begin{lemma}[Propagation of $\kfquery_{\pp_i}$]\label{thm:pspace_query_reduction}
  Let $\Delta = (W, \cA, \pr_1 : R, \dots, \pr_{k-1} : R, \pr_k : \T, I, \dots, I)$ for some type $\T$. Let $\T = \tsel{\pr_{k-1}}{\kfquery_{\pp_i}: \T'}$, such that $\T' = \tbra{\pr_{k-1}}{\kfyes: \T_\kfyes, \kfno: \T_\kfno}$. Then:
  \begin{itemize}
    \item If $\cA(\pp_i) = \ttrue$, then $\Delta \rightarrowsafeuniques (W, \cA, \pr_1 : R, \dots, \pr_{k-1} : R, \T_\kfyes, I, \dots, I)$.
    \item If $\cA(\pp_i) = \tfalse$, then $\Delta \rightarrowsafeuniques (W, \cA, \pr_1 : R, \dots, \pr_{k-1} : R, \T_\kfno, I, \dots, I)$.
  \end{itemize}
\end{lemma}
\begin{proof}
 We focus on the case where $\cA(\pp_i) = \ttrue$; the other case is symmetric. Defining $(V^j)^{\participant{t}}_{11} = T^{\participant{t}}_{11}$ if $V^j = T$, and $(V^j)^{\participant{t}}_{11} = F^{\participant{t}}_{11}$ if $V^j = F$, \textit{etc.}, observe that:\\

$
\begin{array}{rll}
\small
\Delta\rightarrowsafeunique &(W, \cA, \pr_1 : R, \pr_2 : R, \dots, \pr_{k-3} : R, \pr_{k-2} : R, \pr_{k-1} : R_1^{\pp_i}, \pr_k : \T', I, \dots, I)\\
    \rightarrowsafeunique &(W, \cA, \pr_1 : R, \pr_2 : R, \dots, \pr_{k-3} : R, \pr_{k-2} : R_1^{\pp_i}, \pr_{k-1} : R_2, \pr_k : \T', I, \dots, I)\\
                          &\dots \\
    \rightarrowsafeunique &(W, \cA, \pr_1 : R_1^{\pp_i}, \pr_2 : R_2, \dots, \pr_{k-3} : R_2, \pr_{k-2} : R_2, \pr_{k-1} : R_2, \pr_k : \T', I, \dots, I)\\
    = &(W, V^1, \dots, V^{i-1}, \pp_i: T, V^{i+1}, V^{i+2}, \dots, V^{n-1}, V^n, R_1^{\pp_i}, R_2, \dots, R_2, \T', I, \dots, I)\\
   \rightarrowsafeunique &(W, V^1, \dots, V^{i-1}, \pp_i : T, V^{i+1}, V^{i+2}, \dots, V^{n-1}, (V^n)^{\pp_i}_{11}, R_2, R_2, \dots, R_2, \T', I, \dots, I)\\
    \rightarrowsafeunique &(W, V^1, \dots, V^{i-1}, \pp_i : T, V^{i+1}, V^{i+2}, \dots, (V^{n-1})^{\pp_i}_{11}, V^n_{12}, R_2, R_2, \dots, R_2, \T', I, \dots, I)\\
                          &\dots \\
\end{array}
$\\

$
\begin{array}{lll}
\rightarrowsafeunique &(W, V^1, \dots, V^{i-1}, \pp_i : T, (V^{i+1})^{\pp_i}_{11}, V^{i+2}_{12}, \dots, V^{n-1}_{12}, V^n_{12}, R_2, R_2, \dots, R_2, \T', I, \dots, I)\\
\rightarrowsafeunique &(W, V^1, \dots, V^{i-1}, \pp_i : T_{21}, V^{i+1}_{12}, V^{i+2}_{12}, \dots, V^{n-1}_{12}, V^n_{12}, R_2, R_2, \dots, R_2, \T', I, \dots, I)\\
    \rightarrowsafeunique &(W, V^1, \dots, V^{i-1}, \pp_i : T, V^{i+1}_{13}, V^{i+2}_{12}, \dots, V^{n-1}_{12}, V^n_{12}, R_2, R_2, \dots, R_2, \T', I, \dots, I)\\
    \rightarrowsafeunique &(W, V^1, \dots, V^{i-1}, \pp_i : T, V^{i+1}, V^{i+2}_{13}, \dots, V^{n-1}_{12}, V^n_{12}, R_2, R_2, \dots, R_2, \T', I, \dots, I)\\
                          &\dots \\
    \rightarrowsafeunique &(W, V^1, \dots, V^{i-1}, \pp_i : T, V^{i+1}, V^{i+2}, \dots, V^{n-1}, V^n_{13}, R_2, R_2, \dots, R_2, \T', I, \dots, I)\\
    \rightarrowsafeunique &(W, V^1, \dots, V^{i-1}, \pp_i : T, V^{i+1}, V^{i+2}, \dots, V^{n-1}, V^n, R_3, R_2, \dots, R_2, \T', I, \dots, I)\\
    = &(W, \cA, \pr_1 : R_3, \pr_2 : R_2, \dots, \pr_{k-2} : R_2, \pr_{k-1} : R_2, \pr_k : \T', I, \dots, I)\\
    \rightarrowsafeunique &(W, \cA, \pr_1 : R, \pr_2 : R_3, \dots, \pr_{k-2} : R_2, \pr_{k-1} : R_2, \pr_k : \T', I, \dots, I)\\
                          &\dots \\
    \rightarrowsafeunique &(W, \cA, \pr_1 : R, \pr_2 : R, \dots, \pr_{k-2} : R, \pr_{k-1} : R_3, \pr_k : \T', I, \dots, I)\\
    \rightarrowsafeunique &(W, \cA, \pr_1 : R, \pr_2 : R, \dots, \pr_{k-2} : R, \pr_{k-1} : R, \pr_k : \T_\kfyes, I, \dots, I)
\end{array}
$\\

Intuitively, the message $\kfquery_{\pp_i}$ is being propagated through the chain of participants, starting at $\pr_k$ and ending at $\pp_i$, and then the ``return value'' of $\kfyes$ is propagated backwards through the same chain. The states of the participants store the necessary information to uniquely specify this reduction path.
\end{proof}

The next lemma details how $\pr_i$ finds the truth value of the conjunction of all clauses starting from $\C_i$.
As the participants $\pr_i$ use the $\kfquery_{\rolem{t}}$ labels to find the truth values of their variables,
we use Lemma~\ref{thm:pspace_query_reduction} to specify the results of this behaviour.
As one might expect, $\pr_i$ enters state $T_r$, ready to send a $\kfdoneyes$ to $\pr_{i-1}$
if the conjunction is true, otherwise it will enter state $F_r$ and instead send $\kfdoneno$.

\begin{lemma}[Reduction of $\pr_i$]\label{thm:pspace_clause_reduction}
  For $\cA: \{v_1, \dots, v_n\} \rightarrow \{\tfalse, \ttrue\}$:
  \begin{itemize}
    \item $(W, \cA, R, \dots, R, \pr_k : I_1, I, \dots, I) \rightarrowsafeuniques (W, \cA, R, \dots, R, \pr_k : T_r, I, \dots, I)$ if $\cA \models \bigwedge_{i=k}^m C_i$.
    \item $(W, \cA, R, \dots, R, \pr_k : I_1, I, \dots, I) \rightarrowsafeuniques (W, \cA, R, \dots, R, \pr_k : F_r, I, \dots, I)$ if $\cA \not\models \bigwedge_{i=k}^m C_i$.
  \end{itemize}
\end{lemma}
\begin{proof}
  Let $\sqbf_k = \bigwedge_{i=k}^m C_i$. The proof is by induction on $k$. For $k = m$, $F$ and $I_1$ are the same state, so the result follows immediately.
  For $k < m$, we have two cases to consider.

  \noindent\case $\cA \models \sqbf_k$. We have that $\cA \models \sqbf_{k+1}$ and at least one of $\cA(L_{k1})$, $\cA(L_{k2})$ and $\cA(L_{k3})$ is $\ttrue$.
  
If $\cA(L_{k1}) = \ttrue$:
\[
  \begin{array}{lll}
    &(W, \cA, R, \dots, R, \pr_k : I_1, I, \dots, I)\\
    \rightarrowsafeunique &(W, \cA, R, \dots, R, \pr_k : R, \pr_{k+1} : I_1, I \dots, I)\\
    \rightarrowsafeuniques &(W, \cA, R, \dots, R, \pr_k : R, \pr_{k+1} : T_r, I, \dots, I) & \text{by inductive hypothesis}\\
    \rightarrowsafeunique &(W, \cA, R, \dots, R, \pr_k : Q_1', I, \dots, I)\\
    \rightarrowsafeuniques &(W, \cA, R, \dots, R, \pr_k : T_r, I, \dots, I) & \text{by Lemma~\ref{thm:pspace_query_reduction}}
  \end{array}
\]
  If $\cA(L_{k1}) = \tfalse$ and $\cA(L_{k2}) = \ttrue$:
\[
  \begin{array}{lll}
    &(W, \cA, R, \dots, R, \pr_k : I_1, I, \dots, I)\\
    \rightarrowsafeunique &(W, \cA, R, \dots, R, \pr_k : R, \pr_{k+1} : I_1, I \dots, I)\\
    \rightarrowsafeuniques &(W, \cA, R, \dots, R, \pr_k : R, \pr_{k+1} : T_r, I, \dots, I) & \text{by inductive hypothesis}\\
    \rightarrowsafeunique &(W, \cA, R, \dots, R, \pr_k : Q_1', I, \dots, I)\\
    \rightarrowsafeuniques &(W, \cA, R, \dots, R, \pr_k : Q_2', I, \dots, I) & \text{by Lemma~\ref{thm:pspace_query_reduction}}\\
    \rightarrowsafeuniques &(W, \cA, R, \dots, R, \pr_k : T_r, I, \dots, I) & \text{by Lemma~\ref{thm:pspace_query_reduction}}\\
\end{array}
\]

  If $\cA(L_{k1}) = \tfalse$ and $\cA(L_{k2}) = \tfalse$ and $\cA(L_{k3}) = \ttrue$:
\[
  \begin{array}{lll}
    &(W, \cA, R, \dots, R, \pr_k : I_1, I, \dots, I)\\
    \rightarrowsafeunique &(W, \cA, R, \dots, R, \pr_k : R, \pr_{k+1} : I_1, I \dots, I)\\
    \rightarrowsafeuniques &(W, \cA, R, \dots, R, \pr_k : R, \pr_{k+1} : T_r, I, \dots, I) & \text{by inductive hypothesis}\\
    \rightarrowsafeunique &(W, \cA, R, \dots, R, \pr_k : Q_1', I, \dots, I)\\
    \rightarrowsafeuniques &(W, \cA, R, \dots, R, \pr_k : Q_2', I, \dots, I) & \text{by Lemma~\ref{thm:pspace_query_reduction}}\\
    \rightarrowsafeuniques &(W, \cA, R, \dots, R, \pr_k : Q_3', I, \dots, I) & \text{by Lemma~\ref{thm:pspace_query_reduction}}\\
    \rightarrowsafeuniques &(W, \cA, R, \dots, R, \pr_k : T_r, I, \dots, I) & \text{by Lemma~\ref{thm:pspace_query_reduction}}\\
\end{array}
\]

  \noindent\case $\cA \not\models \sqbf_k$. We have that $\cA \not\models \sqbf_{k+1}$ or $\cA(L_{k1}) = \cA(L_{k2}) = \cA(L_{k3}) = \tfalse$.
  
  If $\cA \not\models \sqbf_{k+1}$:
\[
  \begin{array}{lll}
    &(W, \cA, R, \dots, R, \pr_k : I_1, I, \dots, I)\\
    \rightarrowsafeunique &(W, \cA, R, \dots, R, \pr_k : R, \pr_{k+1} : I_1, I \dots, I)\\
    \rightarrowsafeuniques &(W, \cA, R, \dots, R, \pr_k : R, \pr_{k+1} : F_r, I, \dots, I) & \text{by inductive hypothesis}\\
    \rightarrowsafeunique &(W, \cA, R, \dots, R, \pr_k : F_r, I, \dots, I)\\
\end{array}
\]
  If $\cA \models \sqbf_{k+1}$ and $\cA(L_{k1}) = \cA(L_{k2}) = \cA(L_{k3}) = \tfalse$:
\[
  \begin{array}{lll}
    &(W, \cA, R, \dots, R, \pr_k : I_1, I, \dots, I)\\
    \rightarrowsafeunique &(W, \cA, R, \dots, R, \pr_k : R, \pr_{k+1} : I_1, I \dots, I)\\
    \rightarrowsafeuniques &(W, \cA, R, \dots, R, \pr_k : R, \pr_{k+1} : T_r, I, \dots, I) & \text{by inductive hypothesis}\\
    \rightarrowsafeunique &(W, \cA, R, \dots, R, \pr_k : Q_1', I, \dots, I)\\
    \rightarrowsafeuniques &(W, \cA, R, \dots, R, \pr_k : Q_2', I, \dots, I) & \text{by Lemma~\ref{thm:pspace_query_reduction}}\\
    \rightarrowsafeuniques &(W, \cA, R, \dots, R, \pr_k : Q_3', I, \dots, I) & \text{by Lemma~\ref{thm:pspace_query_reduction}}\\
    \rightarrowsafeuniques &(W, \cA, R, \dots, R, \pr_k : F_r, I, \dots, I) & \text{by Lemma~\ref{thm:pspace_query_reduction}}\\
\end{array}
\]

  Each case yields the desired result, so the lemma holds.
\end{proof}

Next, we show the behaviour of $\pp_i$, which behaves similarly to $\pr_i$, but instead handles the quantified variables.

\begin{lemma}[Reduction of $\pp_i$]\label{thm:pspace_var_reduction}
  For $\cA: \{v_1, \dots, v_k\} \rightarrow \{\tfalse, \ttrue\}$:
  \begin{itemize}
    \item $[\cA] \rightarrowsafeuniques (W, \cA, T_r, I, \dots, I)$ if $\cA \models \cQ_{k+1} v_{k+1} \dots \cQ_n v_n \sqbf$.
    \item $[\cA] \rightarrowsafeuniques (W, \cA, F_r, I, \dots, I)$ if $\cA \not\models \cQ_{k+1} v_{k+1} \dots \cQ_n v_n \sqbf$.
  \end{itemize}
\end{lemma}
\begin{proof}
  Let $\sqbf_k = \cQ_{k+1} v_{k+1} \dots \cQ_n v_n \sqbf$.
  The proof is by induction on $k$. For $k = 0$, use Lemma~\ref{thm:pspace_clause_reduction}.
  For $k > 0$, we consider the case $\cQ_{k+1} = \forall$; the $\exists$ case is similar.
  
  If $\cA \models \sqbf_k$,
    we have $\cA_{[v_{k+1} \mapsto \tfalse]} \models \sqbf_k$
    and $\cA_{[v_{k+1} \mapsto \ttrue]} \models \sqbf_k$.
\[
  \begin{array}{lll}
    [\cA]\\
    \rightarrowsafeunique &(W, \cA, F, I_1, \dots, I)
      = [\cA_{[v_{k+1} \mapsto \tfalse]}]\\
    \rightarrowsafeuniques &(W, \cA_{[v_{k+1} \mapsto \tfalse]}, T_r, I, \dots, I)
      = (W, \cA, F, T_r, I, \dots, I)\\
    \rightarrowsafeunique &(W, \cA, I_2, I, \dots, I)\\
    \rightarrowsafeunique &(W, \cA, T, I_1, \dots, I)
      = [\cA_{[v_{k+1} \mapsto \ttrue]}]\\
    \rightarrowsafeuniques &(W, \cA_{[v_{k+1} \mapsto \ttrue]}, T_r, I, \dots, I)
      = (W, \cA, T, T_r, I, \dots, I)\\
    \rightarrowsafeunique &(W, \cA, T_r, I, \dots, I)\\
\end{array}
\]
  If $\cA \not\models \sqbf_k$, then there are two cases: first, (i) $\cA_{[v_{k+1} \mapsto \tfalse]} \not\models \sqbf_k$:
\[
  \begin{array}{lll}
    [\cA]\\
    \rightarrowsafeunique &(W, \cA, F, I_1, \dots, I)
      = [\cA_{[v_{k+1} \mapsto \tfalse]}]\\
    \rightarrowsafeuniques &(W, \cA_{[v_{k+1} \mapsto \tfalse]}, F_r, I, \dots, I)
      = (W, \cA, F, F_r, I, \dots, I)\\
    \rightarrowsafeunique &(W, \cA, F_r, I, \dots, I)\\
\end{array}
\]
  Second, (ii) $\cA_{[v_{k+1} \mapsto \tfalse]} \models \sqbf_k$ and $\cA_{[v_{k+1} \mapsto \ttrue]} \not\models \sqbf_k$:
\[
  \begin{array}{lll}
    [\cA]\\
    \rightarrowsafeunique &(W, \cA, F, I_1, \dots, I)
      = [\cA_{[v_{k+1} \mapsto \tfalse]}]\\
    \rightarrowsafeuniques &(W, \cA_{[v_{k+1} \mapsto \tfalse]}, T_r, I, \dots, I)
      = (W, \cA, F, T_r, I, \dots, I)\\
    \rightarrowsafeunique &(W, \cA, I_2, I, \dots, I)\\
    \rightarrowsafeunique &(W, \cA, T, I_1, \dots, I)
      = [\cA_{[v_{k+1} \mapsto \ttrue]}]\\
    \rightarrowsafeuniques &(W, \cA_{[v_{k+1} \mapsto \ttrue]}, F_r, I, \dots, I)
      = (W, \cA, T, F_r, I, \dots, I)\\
    \rightarrowsafeunique &(W, \cA, F_r, I, \dots, I)\\
\end{array}
\]
  Each $\rightarrowsafeunique$ can be checked directly, and each $\rightarrowsafeuniques$ is from the inductive hypothesis.
  Therefore, the lemma follows by induction.
\end{proof}

\begin{theorem}\label{thm:reduction_from_empty_assignment}
  Recall that
  $[\emptyset] = (\ps{:}W, \pp_1{:}I_1, \pp_2{:}I, \dots,
  \pr_{m+1}{:}I)$. 
  Then: {\rm (1)} 
  $\Deltainit \rightarrowsafeunique [\emptyset]$;
  {\rm (2)} 
  If $\models \bqbf$, then $[\emptyset] \rightarrowsafeuniques
  \Deltainit$; 
  and {\rm (3)} If $\not\models \bqbf$, then $[\emptyset] \rightarrowsafeuniques (\ps{:}\Tbad, \pp_1{:}I, \dots, \pr_{m+1}{:}I)$.
\end{theorem}
\begin{proof}
  The first statement is true by inspection of the transition system. The latter two statements are true by Lemma~\ref{thm:pspace_var_reduction} with $k = 0$.
\end{proof}


We can now use the above lemmas to prove \PSPACE-hardness for checking safety, liveness and deadlock-freedom.

\subsection{\PSPACE-Hardness of Safety}
\label{sec:safety_pspace_hard}

\begin{lemma}\label{thm:unique_reduction_unsafety}
  If $\Delta \rightarrow \Delta_1 \rightarrowsafeuniques \Delta$ 
    and $\Delta \rightarrow^* \Delta'$
    and $\Delta$ is a safe state
    and $\Delta'$ is not a safe state
    then $\Delta \rightarrow \Delta_2 \rightarrow^* \Delta'$
    for some $\Delta_2 \neq \Delta_1$.
\end{lemma}
\begin{proof}
  By contradiction. Assume that $\Delta$, $\Delta_1$, $\Delta'$ are as described,
  but there is no $\Delta_2 \neq \Delta_1$ such that $\Delta \rightarrow \Delta_2 \rightarrow^* \Delta'$.
  Fix the shortest reduction sequence $\Delta = \Delta^0 \rightarrow \Delta^1 \dots \rightarrow \Delta^n = \Delta'$.
  As $\Delta$ is a safe state and $\Delta'$ is not a safe state, $\Delta \neq \Delta'$ and so $n > 0$.
  From our contradiction assumption, $\Delta^1 = \Delta_1$.
  As $\Delta_1 \rightarrowsafeuniques \Delta$, fix $\Delta^{\dag}_{i}$ such that $\Delta^{\dag}_{1} \rightarrowsafeunique \Delta^{\dag}_{2} \rightarrowsafeunique \dots \rightarrowsafeunique \Delta^{\dag}_{m} = \Delta$.
  By the definition of $\rightarrowsafeunique$ (Definition~\ref{def:unique_safe_reduction_paths}), $\Delta^{\dag}_{i} = \Delta^i$ and $\Delta^i$ is a safe state for all $i \in \{1, \dots, m\}$,
  so $\Delta^m = \Delta$.
  But then $\Delta = \Delta^m \rightarrow \dots \rightarrow \Delta^n = \Delta'$ is a shorter reduction sequence, contradicting our assumption of minimality.
\end{proof}

\begin{theorem}\label{thm:safety_pspace_hard}
  Checking for safety is \PSPACE-hard.
\end{theorem}
\begin{proof} Set $\Tbad = \tbra{\pp_1}{\kf{unsafe}: \tend}$.
  \begin{itemize}
    \item If $\models \bqbf$, then by Theorem~\ref{thm:reduction_from_empty_assignment}, $[\emptyset] \rightarrowsafeuniques \Deltainit$. Assume for contradiction that $\Deltainit \rightarrow \Delta'$ for some $\Delta'$ that is not a safe state. We have $\Deltainit \rightarrow [\emptyset]$, so by Lemma \ref{thm:unique_reduction_unsafety}, $\Deltainit \rightarrow \Delta_2$ for some $\Delta_2 \neq [\emptyset]$. But no such $\Delta_2$ exists, so $\Deltainit$ is safe by Lemma \ref{thm:safety_from_reductions}.
    \item If $\not\models \bqbf$, then by Theorem~\ref{thm:reduction_from_empty_assignment}, $\Deltainit \rightarrow [\emptyset] \rightarrowsafeuniques (\Tbad, I, \dots, I)$. But $(\Tbad, I, \dots, I) \redBra{\ps}{\pp_1}{\kf{unsafe}}$ and $(\Tbad, I, \dots, I) \not\redSelBra{\pp_1}{\ps}{\kf{unsafe}}$, so it is not a safe state. Thus $\Deltainit$ is not safe.
  \end{itemize}
  Therefore, $\Deltainit$ is safe if and only if $\models \bqbf$. As QBF is \PSPACE-complete, we conclude that checking for safety is \PSPACE-hard.
\end{proof}

\subsection{\PSPACE-Hardness of Deadlock-Freedom}
\label{sec:deadlock_freedom_pspace_hard}

The proof that deadlock-freedom is \PSPACE-hard is similar to the above proof for safety, as unique reduction paths also guarantee deadlock-freedom. First, we prove the analogue of Lemma~\ref{thm:unique_reduction_unsafety} for deadlock-freedom.

\begin{lemma}\label{thm:unique_reduction_deadlocks}
  If $\Delta \rightarrow \Delta_1 \rightarrowuniques \Delta$ 
    and $\Delta \rightarrow^* \Delta' \not\rightarrow$
    then $\Delta \rightarrow \Delta_2 \rightarrow^* \Delta'$
    for some $\Delta_2 \neq \Delta_1$.
\end{lemma}
\begin{proof}
  By contradiction. Assume that $\Delta$, $\Delta_1$, $\Delta'$ are as described,
  but there is no $\Delta_2 \neq \Delta_1$ such that $\Delta \rightarrow \Delta_2 \rightarrow^* \Delta'$.
  Fix the shortest reduction sequence $\Delta = \Delta^0 \rightarrow \Delta^1 \dots \rightarrow \Delta^n = \Delta'$.
  As $\Delta$ is a safe state and $\Delta'$ is not a safe state, $\Delta \neq \Delta'$ and so $n > 0$.
  From our contradiction assumption, $\Delta^1 = \Delta_1$.
  As $\Delta_1 \rightarrowuniques \Delta$, fix $\Delta^{\dag}_{i}$ such that $\Delta^{\dag}_{1} \rightarrowunique \Delta^{\dag}_{2} \rightarrowunique \dots \rightarrowunique \Delta^{\dag}_{m} = \Delta$.
  By the definition of $\rightarrowunique$ (Definition~\ref{def:unique_safe_reduction_paths}), $\Delta^{\dag}_{i} = \Delta^i$ and $\Delta^i \rightarrow$ for all $i \in \{1, \dots, m\}$,
  so $\Delta^m = \Delta$.
  But then $\Delta = \Delta^m \rightarrow \dots \rightarrow \Delta^n = \Delta'$ is a shorter reduction sequence, contradicting our assumption of minimality.
\end{proof}

\begin{theorem}\label{thm:deadlock_freedom_pspace_hard}
  Checking for deadlock-freedom is \PSPACE-hard.
\end{theorem}
\begin{proof}
  Set $\Tbad = \tend$.
  \begin{itemize}
    \item If $\models \bqbf$, then by Theorem~\ref{thm:reduction_from_empty_assignment}, $[\emptyset] \rightarrowsafeuniques \Deltainit$. Assume for contradiction that $\Deltainit \rightarrow \Delta'$ for some $\Delta'\not\rightarrow$. We have $\Deltainit \rightarrow [\emptyset]$, so by Lemma \ref{thm:unique_reduction_unsafety}, $\Deltainit \rightarrow \Delta_2$ for some $\Delta_2 \neq [\emptyset]$. But no such $\Delta_2$ exists, so $\Deltainit$ is deadlock-free by Lemma~\ref{thm:unique_reduction_deadlocks}.
    \item If $\not\models \bqbf$, then by Theorem~\ref{thm:reduction_from_empty_assignment}, $\Deltainit \rightarrow [\emptyset] \rightarrowsafeuniques (\Tbad, I, \dots, I)$. But $(\Tbad, I, \dots, I) \not\rightarrow$, so $\Deltainit$ is not deadlock-free (because $\unfold{I} \neq \tend$).
  \end{itemize}
  Therefore, $\Deltainit$ is deadlock-free if and only if $\models \bqbf$. As QBF is \PSPACE-complete, we conclude that checking for deadlock-freedom is \PSPACE-hard.
\end{proof}

\subsection{\PSPACE-Hardness of Liveness}
\label{sec:liveness_pspace_hard}

First, we prove some properties about liveness that will be useful in our proof.

\begin{lemma}\label{thm:liveness_implies_df}
  If $\Delta$ is live then $\Delta$ is deadlock-free.
\end{lemma}
\begin{proof}
By \cite[Theorem~4]{YH2024}.
\end{proof}



\begin{definition}[Participants of Reductions]\label{def:reduction_participants}
  Define the participants of a reduction: $\pt{\redOutIn{\pp}{\pq}} = \{\pp, \pq\}$ and $\pt{\redSelBra{\pp}{\pq}{l}} = \{\pp, \pq\}$.
\end{definition}

\begin{restatable}{lemma}{livemovement}
   \label{thm:single_selection_live_paths}
Let $(\Delta_i)_{i \in \mathbb{N}}$ be an infinite path such that $\Delta_i \trans{\cxell_i} \Delta_{i+1}$ for all $i$.
If $\bigcup_{i \geq j} \pt{\cxell_i} = \pt{\Delta}$ for all $j \in \mathbb{N}$, then $(\Delta_i)_{i \in \mathbb{N}}$ is a live path.
\end{restatable}
\begin{proof}
  Let $i \in \mathbb{N}$, and let $\pp \in \pt{\Delta}$. Let $\Delta = \pp: \T, \Delta'$.
  By assumption there exists $j \geq i$ such that $\pp \in \pt{\cxell_j}$.
  Choose the smallest such $j$.
  Then, no transition before $\cxell_j$ affects $\pp$, so $\Delta_j(\pp) = \T$.
  We consider the syntax of $\unfold{\T}$:

  \case $\T = \tend$. Then there are no transitions involving $\pp$.
  
  \case $\tin{\pq}{S} \T'$. By considering the rules of Definition~\ref{def:composition_reduction}, $\cxell_j = \OutIn{\pq}{\pp}$, as desired.

  \case $\tout{\pq}{S} \T'$. Similar to the above: $\cxell_j = \OutIn{\pp}{\pq}$.

  \case $\tselsub{\pq}{l_i: \T_i}{i \in \I}$. We conclude that $\cxell_j = \SelBra{\pp}{\pq}{l_k}$ for some $k \in \I$.

  \case $\tbrasub{\pq}{l_i: \T_i}{i \in \I}$. We conclude that $\cxell_j = \SelBra{\pq}{\pp}{l_k}$ for some $k \in \I$.

All transitions in the definition of live path fall into one of the above cases, so the path is live.
\end{proof}

\begin{theorem}\label{thm:liveness_pspace_hard}
  Checking for liveness is \PSPACE-hard.
\end{theorem}
\begin{proof}
  We use the same construction as the deadlock-freedom reduction (Theorem~\ref{thm:deadlock_freedom_pspace_hard}): set $\Tbad = \tend$.
  \begin{itemize}
    \item If $\models\bqbf$, then by the same argument as Theorem~\ref{thm:deadlock_freedom_pspace_hard}, $\Delta \rightarrowunique [\emptyset] \rightarrowuniques \Delta$.
    Thus there is a unique infinite path $\Delta \rightarrowunique^\omega$.
    We will show that this path is live.

    Furthermore, the reduction sequence $\Delta \rightarrow (W, I_1, I, \dots, I) \rightarrowunique \dots \rightarrowunique (W, F, \dots, F, R, \dots, R)$ involves all participants in at least one transition.
    As $\Delta_j \rightarrowuniques \Delta$ for all $\Delta_j$, this reduction sequence occurs after $\Delta_j$, therefore $\bigcup_{i \geq j} \pt{\cxell_i} = \pt{\Delta}$ for all $j \in \mathbb{N}$.
    We can then apply Lemma~\ref{thm:single_selection_live_paths} to conclude that $\Delta \rightarrowunique^\omega$ is a live path.
    As all reachable contexts have outgoing transitions (due to the path being deterministic), there are no finite fair paths.
    Furthermore, we have shown that the only infinite path is live, so all fair paths are live.
    We conclude that $\Delta$ is live.
    \item If $\not\models \bqbf$, then by the proof of Theorem~\ref{thm:deadlock_freedom_pspace_hard}, $\Delta$ is not deadlock-free. By Lemma~\ref{thm:liveness_implies_df}, $\Delta$ is not live.
  \end{itemize}

Therefore, $\Delta$ is live if and only if $\models \bqbf$. As QBF is \PSPACE-complete, we conclude that checking for liveness is \PSPACE-hard.
\end{proof}

\end{document}